\documentclass[final]{llncs}
\pdfoutput=1

\usepackage[utf8]{inputenc}
\usepackage[T1]{fontenc}

\usepackage{mathptmx}

\usepackage{keyval}
\usepackage{footmisc}
\usepackage{ragged2e}
\usepackage{xspace}
\usepackage[usenames,dvipsnames,svgnames,table]{xcolor}
\usepackage[final,stretch=10,shrink=10]{microtype}
\usepackage{ifdraft}
\usepackage[inline]{enumitem}
\usepackage{footnote}
\usepackage{tikz}
\usetikzlibrary{backgrounds,snakes,arrows,decorations.markings,calc}
\usepackage{varwidth}
\usepackage{algpseudocode}
\usepackage{algorithm}
\usepackage{soul} %
\usepackage{adjustbox}

\usepackage{captcont}
\usepackage{caption}
\DeclareCaptionLabelSeparator{dot}{.\enspace}
\clearcaptionsetup{figure}
\usepackage[nottoc,section]{tocbibind}

\usepackage{macros}
\usepackage{cite}
\usepackage{amsmath}
\usepackage{amssymb}
\usepackage{thmtools}
\setlistdepth{9}
\setlist[itemize,1]{label=$\bullet$}
\setlist[itemize,2]{label=$\bullet$}
\setlist[itemize,3]{label=$\bullet$}
\setlist[itemize,4]{label=$\bullet$}
\setlist[itemize,5]{label=$\bullet$}
\setlist[itemize,6]{label=$\bullet$}
\setlist[itemize,7]{label=$\bullet$}
\setlist[itemize,8]{label=$\bullet$}
\setlist[itemize,9]{label=$\bullet$}

\renewlist{itemize}{itemize}{9}

\usepackage[english]{babel}
\usepackage[final]{hyperref}
\usepackage{color}
\definecolor{linkcolor}{rgb}{0,0,0.5}
\hypersetup{
 colorlinks=true,
  linkcolor=linkcolor,
  anchorcolor=linkcolor,
  citecolor=linkcolor,
  filecolor=linkcolor,
  menucolor=linkcolor,
  runcolor=linkcolor,
  urlcolor=linkcolor,
}

\usepackage{todonotes}

\usepackage{fancybox}
\usepackage{wrapfig}

\usepackage{fullpage}

\ifdraft{
\usepackage[scrtime]{prelim2e}
\usepackage{showlabels}
}{}

\makeatletter
\renewcommand*\l@author[2]{}
\renewcommand*\l@title[2]{}
\makeatletter

\DeclareSymbolFont{greek}     {OML}{cmm}{m}{it}
\DeclareMathSymbol{\alpha}{\mathord}{greek}{"0B}
\DeclareMathSymbol{\beta}{\mathord}{greek}{"0C}
\DeclareMathSymbol{\gamma}{\mathord}{greek}{"0D}
\DeclareMathSymbol{\delta}{\mathord}{greek}{"0E}
\DeclareMathSymbol{\epsilon}{\mathord}{greek}{"22}
\DeclareMathSymbol{\varepsilon}{\mathord}{greek}{"22}
\DeclareMathSymbol{\zeta}{\mathord}{greek}{"10}
\DeclareMathSymbol{\eta}{\mathord}{greek}{"11}
\DeclareMathSymbol{\theta}{\mathord}{greek}{"12}
\DeclareMathSymbol{\iota}{\mathord}{greek}{"13}
\DeclareMathSymbol{\kappa}{\mathord}{greek}{"14}
\DeclareMathSymbol{\lambda}{\mathord}{greek}{"15}
\DeclareMathSymbol{\mu}{\mathord}{greek}{"16}
\DeclareMathSymbol{\nu}{\mathord}{greek}{"17}
\DeclareMathSymbol{\xi}{\mathord}{greek}{"18}
\DeclareMathSymbol{\pi}{\mathord}{greek}{"19}
\DeclareMathSymbol{\rho}{\mathord}{greek}{"1A}
\DeclareMathSymbol{\sigma}{\mathord}{greek}{"1B}
\DeclareMathSymbol{\tau}{\mathord}{greek}{"1C}
\DeclareMathSymbol{\upsilon}{\mathord}{greek}{"1D}
\DeclareMathSymbol{\phi}{\mathord}{greek}{"27}
\DeclareMathSymbol{\chi}{\mathord}{greek}{"1F}
\DeclareMathSymbol{\psi}{\mathord}{greek}{"20}
\DeclareMathSymbol{\omega}{\mathord}{greek}{"21}
\DeclareMathSymbol{\vartheta}{\mathord}{greek}{"23}
\DeclareMathSymbol{\varpi}{\mathord}{greek}{"24}
\DeclareMathSymbol{\varrho}{\mathord}{greek}{"25}
\DeclareMathSymbol{\varsigma}{\mathord}{greek}{"26}
\DeclareMathSymbol{\varphi}{\mathord}{greek}{"27}

\begin{document}
\hyphenation{Brow-serID}
\hyphenation{in-fra-struc-ture}
\hyphenation{brow-ser}
\hyphenation{doc-u-ment}
\hyphenation{Chro-mi-um}
\hyphenation{meth-od}
\hyphenation{sec-ond-ary}
\hyphenation{Java-Script}
\hyphenation{Mo-zil-la}
\hyphenation{post-Mes-sage}
\hyphenation{RP-doc}
\hyphenation{IdP-doc}
\hyphenation{in-dis-tin-guish-abil-i-ty}

\title{\spresso: A Secure, Privacy-Respecting\\ Single Sign-On System for the Web}
\author{Daniel Fett \and Ralf Küsters \and Guido Schmitz}
\institute{University of Trier, Germany\\\email{\{fett,kuesters,schmitzg\}@uni-trier.de}}

\maketitle

\begin{abstract}
  Single sign-on (SSO) systems, such as OpenID and OAuth, allow web sites, so-called relying parties (RPs), to delegate user authentication to identity providers (IdPs), such as Facebook or Google. These systems are very popular, as they provide a convenient means for users to log in at RPs and move much of the burden of user authentication from RPs to IdPs.

  There is, however, a downside to current systems, as they do not respect users' privacy: IdPs learn at which RP a user logs in. With one exception, namely Mozilla's BrowserID system (a.k.a. Mozilla Persona), current SSO systems were not even designed with user privacy in mind. Unfortunately, recently discovered attacks, which exploit design flaws of BrowserID, show that BrowserID does not provide user privacy either.
  
  In this paper, we therefore propose the first privacy-respecting SSO system for the web, called \spresso (for Secure Privacy-REspecting Single Sign-On). The system is easy to use, decentralized, and platform independent. It is based solely on standard HTML5 and web features and uses no browser extensions, plug-ins, or other executables.

  Existing SSO systems and the numerous attacks on such systems illustrate that the design of secure SSO systems is highly non-trivial. We therefore also carry out a formal analysis of \spresso based on an expressive model of the web in order to formally prove that \spresso enjoys strong authentication and privacy properties.
\end{abstract}

\ifdraft{
\listoftodos
}{ }

\tableofcontents

\newpage

\section{Introduction}
\label{sec:introduction}

\noindent Web-based Single Sign-On (SSO) systems allow a
user to identify herself to a so-called relying party (RP),
which provides some service, using an identity that is
managed by an identity provider (IdP), such as Facebook or
Google. If an RP uses an SSO system, a user does not need a
password to log in at the RP. Instead, she is authenticated
by the IdP, which exchanges some data with the RP so that
the RP is convinced of the user's identity. When logged in
at the IdP already, a user can even log in at the RP by one
click without providing any password. This makes SSO
systems very attractive for users. These systems are also
very convenient for RPs as much of the burden of user
authentication, including, for example, the handling of
user passwords and lost passwords, is shifted to the IdPs.
This is why SSO systems are very popular and widely used on
the web. Over the last years, many different SSO systems
have been developed, with OpenID~\cite{openid-2-spec} (used by Google, Yahoo,
AOL, and Wordpress, for example) and OAuth~\cite{rfc6749-oauth2} (used by
Twitter, Facebook, PayPal, Microsoft, GitHub, and LinkedIn, for
example) being the most prominent of
such systems; other SSO systems include SAML/Shibboleth,
CAS, and WebAuth.

There is, however, a downside to these systems: with one
exception, none of the existing SSO systems have been
designed to respect users’ privacy. That is, the IdP always
knows at which RP the user logs in, and hence, which
services the user uses. In fact, exchanging user data
between IdPs and RPs directly in every login process is a
key part of the protocols in OpenID and OAuth, for example,
and thus, IdPs can easily track users.

The first system so far which was designed with the intent
to respect users' privacy was the BrowserID system
\cite{mozilla/persona,mozilla/persona/beta2-announce}, which is a
relatively new system developed by Mozilla and is also
known by its marketing name \emph{Persona}.

Unfortunately, in
\cite{FettKuestersSchmitz-ESORICS-BrowserID-Primary-2015}
severe attacks against BrowserID were discovered, which
show that the privacy of BrowserID is completely broken:
these attacks allow malicious IdPs and in some versions of
the attacks even arbitrary parties to check the login
status of users at any RP with little effort (see
Section~\ref{sec:main-features} for some more details on
these attacks). Even worse, these attacks exploit design
flaws of BrowserID that, as discussed in
\cite{FettKuestersSchmitz-ESORICS-BrowserID-Primary-2015},
cannot be fixed without a major redesign of the system, and
essentially require building a new system. As further
discussed in Section~\ref{sec:discussion}, besides the lack
of privacy there are also other issues that motivate the
design of a new system.

The goal of this work is therefore to design the (first)
SSO system which respects users' privacy in the sense
described above, i.e., IdPs (even completely malicious
ones) should not be able to track at which RPs
users log in. Moreover, the history of SSO systems shows
that it is highly non-trivial to design secure SSO systems,
not only w.r.t.~privacy requirements, but even
w.r.t.~authentication requirements. Attacks easily go
unnoticed and in fact numerous attacks on SSO systems,
including attacks on OAuth, OpenID, Google ID, Facebook
Connect, SAML, and BrowserID have been uncovered which
compromise the security of many services and users at once
\cite{BansalBhargavanetal-POST-2013-WebSpi,BansalBhargavanMaffeis-CSF-2012,baiLeietal-NDSS-2013-authscan,WangChenWang-SP-2012-SSO,SomorovskyMayerSchwenkKampmannJensen-USENIX-2012,SovisKohlarSchwenk-Sicherheit-2010-OpenID-signatures,SantsaiHaekyBeznosov-CS-2012,ZhouEvans-USENIX-SSOScan-2014,Wangetal-USENIX-Explicating-SDKs-2013}. Besides
designing and implementing a privacy-respecting SSO system,
we therefore also carry out a formal security analysis of
the system based on an expressive model of the web
infrastructure in order to provide formal security
guarantees. More specifically, the contributions of our
work are as follows.

\subsubsection{Contributions of this Paper.} In this work,
we propose the system \spresso (for Secure
Privacy-REspecting Single Sign-On). This is the first SSO
system which respects user's privacy. The system allows
users to log in to RPs with their email addresses. A user
is authenticated to an RP by the IdP hosting the user's
email address. This is done in such a way that the IdP does
not learn at which RP the user wants to log in.

Besides strong authentication and privacy guarantees (see
also below), \spresso is designed in such a way that it can be
used across browsers, platforms, and devices. For this
purpose, \spresso is based solely on standard HTML5 and web
features and uses no browser extensions, plug-ins, or
browser-independent executables. 

Moreover, as further discussed in
Section~\ref{sec:main-features}, \spresso is designed as an
open and decentralized system. For example, in contrast to
OAuth, \spresso does not require any prior coordination or
setup between RPs and IdPs: users can log in at any RP with
any email address with \spresso support.

We formally prove that \spresso enjoys strong
authentication and privacy properties. Our analysis is
based on an expressive Dolev-Yao style model of the web
infrastructure~\cite{FettKuestersSchmitz-SP-2014}. This web model is
designed independently of a specific web application and
closely mimics published (de-facto) standards and
specifications for the web, for instance, the HTTP/1.1 and
HTML5 standards and associated (proposed) standards. It is
the most comprehensive web model to date. Among others,
HTTP(S) requests and responses, including several headers,
such as cookie, location, strict transport security (STS),
and origin headers, are modeled. The model of web browsers
captures the concepts of windows, documents, and iframes,
including the complex navigation rules, as well as new
technologies, such as web storage and cross-document
messaging (postMessages). JavaScript is modeled in an
abstract way by so-called scripting processes which can be
sent around and, among others, can create iframes and
initiate XMLHTTPRequests (XHRs). Browsers may be corrupted
dynamically by the adversary.

So far, this web model has been employed to analyze
trace-based properties only, namely, authentication
properties. In this work, we formulate, for the first time,
strong indistinguishability/privacy properties for web
applications. Our general definition is not tailored to a
specific web application, and hence, should be useful
beyond our analysis of \spresso. These properties require
that an adversary should not be able to distinguish two
given systems. In order to formulate these properties we
slightly modify and extend the web model.

Finally, we formalize \spresso in the web model and
formally state and prove strong authentication and privacy
properties for \spresso.  The authentication properties we
prove are central to any SSO system, where our formulation
of these properties follows the one in
\cite{FettKuestersSchmitz-SP-2014}. As for the privacy
property, we prove that a malicious IdP cannot distinguish
whether an honest user logs in at one RP or another. The
analysis we carry out in this work is also interesting by
itself, as web applications have rarely been analyzed based
on an expressive web model (see
Section~\ref{sec:relatedwork}).

\subsubsection{Structure of this Paper.} 
In Section~\ref{sec:description-spresso}, we describe our
system and discuss and motivate design choices. We then, in
Section~\ref{sec:web-model}, briefly recall the general web
model from \cite{FettKuestersSchmitz-SP-2014} and explain
the modifications and extensions we made. The mentioned
strong but general definition of
indistinguishability/privacy for web applications is
presented in Section~\ref{sec:form-defin-priv}. In
Section~\ref{sec:formal-model-spresso}, we provide the
formal model of \spresso, based on which we state and
analyze privacy and authentication of \spresso in
Sections~\ref{sec:proving-privacy} and \ref{sec:prov-auth},
respectively. Further related work is discussed in
Section~\ref{sec:relatedwork}. We conclude in
Section~\ref{sec:conclusion}. All details and proofs are
available in the appendix. An online demo and the source code of
\spresso are available at
\cite{spresso-sourcecode-2015}.

%
\begin{figure*}[hp!]
  \resizebox{0.9\linewidth}{!}{
  \trimbox{2.5cm 0cm 1cm 0cm}{ 
    \centering
    \input{from-paper/final-flow}
  }}\\
  \centering
  \raisebox{0.5ex}{\tikz{\draw [->] (0,0) -- (0.4,0);}}
  HTTPS messages, \raisebox{0.5ex}{\tikz{\draw
      [->,color=blue,>=latex] (0,0) -- (0.4,0);}} \xhrs
  (over HTTPS), \raisebox{0.5ex}{\tikz{\draw
      [->,color=red,dashed] (0,0) -- (0.4,0);}} \pms,
  \raisebox{0.5ex}{\tikz{\draw [->,snake=snake,segment
      length=2ex,segment amplitude=0.2ex] (0,0) --
      (0.4,0);}} browser commands\\[-1ex]

\caption{\label{fig:protocol-flow} \spresso Login Flow.}
\end{figure*}

\section{Description of \spresso}
\label{sec:description-spresso}

In this section, we first briefly describe the main
features of \spresso.  We then provide a detailed
description of the system in Section~\ref{sec:login-flow},
with further implementation details given in
Section~\ref{sec:implementationdetails}. To provide
additional intuition and motivation for the design of
\spresso, in Section~\ref{sec:discussion} we discuss
potential attacks against \spresso and why they are
prevented.

\subsection{Main Features}
\label{sec:main-features}

\spresso enjoys the following key features:

\subsubsection{Strong Authentication and Privacy.}
\spresso is designed to satisfy strong authentication
and privacy properties. 

Authentication is the most fundamental security property of
an SSO system.  That is, i) an adversary should not be able
to log in to an RP, and hence, use the service of the RP,
as an honest user, and ii) an adversary should not be able
to log in the browser of an honest user under an
adversary's identity (identity injection). Depending on the
service provided by the RP, a violation of ii) could allow
the adversary to track an honest user or to obtain user
secrets. We note that in the past, attacks on
authentication have been found in almost all deployed SSO
systems (e.g., OAuth, OpenID, and BrowserID
\cite{SantsaiHaekyBeznosov-CS-2012,SovisKohlarSchwenk-Sicherheit-2010-OpenID-signatures,SantsaiBeznosov-CCS-2012-OAuth,FettKuestersSchmitz-SP-2014,FettKuestersSchmitz-ESORICS-BrowserID-Primary-2015,ZhouEvans-USENIX-SSOScan-2014,Wangetal-USENIX-Explicating-SDKs-2013}).

While authentication assumes the involved RP and IdP to be
honest, privacy is concerned with malicious IdPs. This
property requires that (malicious) IdPs should not be able
to track at which RPs specific users log in. As already
mentioned in the introduction, so far, except for
BrowserID, no other SSO system was designed to provide
privacy. (In fact, exchanging user data between IdPs and
RPs directly is a key part of the protocols in OpenID and
OAuth, for example, and hence, in such protocols, IdPs can
easily track at which RP a user logs in.) However,
BrowserID failed to provide privacy: As shown in
\cite{FettKuestersSchmitz-ESORICS-BrowserID-Primary-2015}, a
subtle attack allowed IdPs (and in some versions of the
attack even arbitrary parties) to check the login status of
users at any RP. More specifically, by running a malicious
JavaScript within the user's browser, an IdP can, for any
RP, check whether the user is logged in at that RP by
triggering the (automatic) login process and testing
whether a certain iframe is created during this process or
not. The (non-)existence of this iframe immediately reveals
the user's login status. Hence, a malicious IdP can track
at which RP a user is logged in. As we discuss in
\cite{FettKuestersSchmitz-ESORICS-BrowserID-Primary-2015}, this
could not be fixed without a major redesign of BrowserID.
Our work could be considered such a major redesign. While
\spresso shares some basic concepts with BrowserID,
\spresso is, however, not based on BrowserID, but a new
system built from scratch (see the discussion in
Section~\ref{sec:discussion}).

The above shows that the design of a secure SSO system is
non-trivial and that attacks are very easy to overlook. As
already mentioned in the introduction, we therefore not
only designed and implemented \spresso to meet strong
authentication and privacy properties, but also perform a
formal analysis of \spresso in an expressive model of the
web infrastructure in order to show that \spresso in fact
meets these properties.

\subsubsection{An Open and Decentralized System.}
We created \spresso as a decentralized, open system. In
\spresso, users are identified by their email addresses,
and email providers certify the users'
authenticity. Compared to OpenID, users do not need to
learn a new, complicated identifier --- an approach similar
to that of BrowserID. But unlike in BrowserID, there is no
central authority in \spresso (see also the discussion in
Section~\ref{sec:discussion}). In contrast to OAuth,
\spresso does not require any prior coordination or setup
between RPs and IdPs: Users can log in at any RP with any
email address with \spresso support. For email addresses
lacking \spresso support, a seamless fallback can be
provided, as discussed later.

\subsubsection{Adherence to Web Standards.} \spresso is
based solely on standard HTML5 and web features and uses no
browser extensions, plug-ins, or other client-side
executables. This guarantees that \spresso can be used across
browsers, platforms and devices, including both desktop
computers and mobile platforms, without installing any
software (besides a browser). Note that on smartphones, for
example, browsers usually do no support extensions or
plug-ins.

\subsection{Login Flow}
\label{sec:login-flow}

We now explain \spresso by a typical login flow in the
system. \spresso knows three distinct types of parties:
relying parties (RPs), i.e., web sites where a user wishes
to log in, identity providers (IdPs), providing to RPs a
proof that the user owns an email address (identity), and
forwarders (FWDs), who forward messages from IdPs to RPs
within the browser. We start with a brief overview
of the login flow and then present the flow in detail.

\subsubsection{Overview.} 
On a high level, the login flow consists of the following
steps: First, on the RP web site, the user enters her email
address. RP then creates what we call a \emph{tag} by
encrypting its own domain name and a nonce with a freshly
generated symmetric key. This tag along with the user's
email address is then forwarded to the IdP. Due to the
privacy requirement, this is done via the user's browser
in such a way that the IdP does not learn from which RP
this data was received. Note also that the tag contains
RP's domain in encrypted form only. The IdP then signs the
tag and the user's email address (provided that the user is
logged in at the IdP, otherwise the user first has to log
in). This signature is called the \emph{identity assertion
  (IA)}. The IA is then transferred to the RP (again via
the user's browser), which checks the signature and
consistency of the data signed and then considers the user
with the given email address to be logged in. We note that
passing the IA to the RP is done using an FWD (the RP
determines which one is used) as it is important that the
IA is delivered to the correct RP (RP document). The IdP
cannot ensure this, because, again due to the privacy
requirements, IdP is not supposed to know the intended RP.

\subsubsection{Detailed Flow.} 
We now take a detailed look at the \spresso login flow. We
refer to the steps of the protocol as depicted in
Figure~\ref{fig:protocol-flow}.  We use the names RP, IdP,
and FWD for the servers of the respective parties. We use
RPdoc, RPRedirDoc, IdPdoc, and FWDdoc as names for HTML
documents delivered by the respective parties. The login
flow involves the servers RP, IdP, and FWD as well as the
user's browser (gray background), in which different
windows/iframes are created: first, the window containing
RPdoc (which is present from the beginning), second, the
login dialog created by RPdoc (initially containing
RPRedirDoc and later IdPdoc), and third, an iframe inside
the login dialog where the document FWDdoc from FWD is
loaded.

As the first step in the protocol, the user opens the login
page at RP~\refprotostep{get-RPdoc}. The actual login then
starts when the user enters her email
address~\refprotostep{input-email}. RPdoc sends this
address in a POST request to
RP~\refprotostep{post-start-login}. RP identifies the IdP
(from the domain in the email address) and retrieves a
\emph{support document} from
IdP~\refprotostep{get-spresso-info-1}. This document is
retrieved from a fixed URL
\nolinkurl{https://IdPdomain/.well-known/spresso-info} and
contains a public (signature verification) key of the IdP.
RP now selects new nonces/symmetric keys $\mi{rpNonce}$,
$\mi{iaKey}$, $\mi{tagKey}$, and
$\mi{loginSessionToken}$~\refprotostep{choose-rp-nonces}
and creates the tag $\mi{tag}$ by encrypting RP's domain
$\str{RPDomain}$ and the nonce $\mi{rpNonce}$ under
$\mi{tagKey}$~\refprotostep{create-tag}. Using standard
Dolev-Yao notation (see also Section~\ref{sec:web-model}),
we denote this term by
\[\mi{tag}:= \encs{\an{\str{RPDomain},
    \mi{rpNonce}}}{\mi{tagKey}}\,.\] RP
further selects an FWD (e.g., a fixed one from its
settings). Now, RP stores $\mi{tag}$, $\mi{iaKey}$, the FWD
domain, and the email address in its session data store
under the session key $\mi{loginSessionToken}$ and sends
$\mi{tagKey}$, $\mi{FWDDomain}$, and
$\mi{loginSessionToken}$ as response to the POST request by
RPdoc~\refprotostep{send-start-login}.

RPdoc now opens the login dialog. Ultimately, this window
contains the login dialog from IdP (IdPdoc) so that the
user can log in to IdP (if not logged in already). However,
to preserve the user's privacy (see the discussion in
Section~\ref{sec:discussion}), RPdoc does not launch
the dialog with the URL of IdPdoc immediately. Instead,
RPdoc opens the login dialog with the URL of RPRedirDoc and
attaches the
$\mi{loginSessionToken}$~\refprotostep{open-IdPdoc}.
RPRedirDoc is loaded from RP (\refprotostep{get-redir-ld}
and~\refprotostep{send-redir-ld}) and redirects the login
dialog to IdPdoc (\refprotostep{redir-to-IdPdoc} and
\refprotostep{get-IdPdoc}), passing the user's email
address, the tag, the FWD domain, and the iaKey from RP, as
stored under the session key $\mi{loginSessionToken}$, to
IdPdoc.\footnote{This data is passed to IdPdoc in the
  fragment identifier of the URL (a.k.a.~\emph{hash}), and
  therefore, it is not necessarily sent to IdP.}

After the browser loaded IdPdoc from IdP, the user enters
her password\footnote{In fact, the
  IdP can as well offer any other form of authentication,
  e.g., TLS client authentication or two-factor
  authentication.} matching her email
address~\refprotostep{enter-secret}. The password, the email address, the
tag, and the FWD domain are now sent to
IdP~\refprotostep{post-loginxhr}. After IdP verified the
user credentials~\refprotostep{check-credentials}, it
creates the identity assertion as the signature \[\mi{ia}
:= \sig{\an{\mi{tag}, \mi{email},
    \mi{FWDDomain}}}{k_\text{IdP}}\] using its private
signing key $k_\text{IdP}$~\refprotostep{create-ia} and
then returns $\mi{ia}$ to
IdPdoc~\refprotostep{send-loginxhr}. We note that $\mi{ia}$
contains the signature only, not the data that was signed.

To avoid that the FWD learns the IA (we discuss this
further in Section~\ref{sec:discussion}), IdPdoc now
encrypts the IA using the iaKey
\refprotostep{create-eia}:
\[\mi{eia} := \encs{\mi{ia}}{\mi{iaKey}}\ .\]
Then, IdPdoc opens an iframe with the URL of FWDdoc,
passing the tag and the encrypted IA to FWDdoc. After the
iframe is loaded~\refprotostep{get-FWDdoc}, FWDdoc sends a
postMessage\footnote{postMessages are messages that are
  sent between different windows in one browser.} to its
parent's opener window, which is
RPdoc~\refprotostep{pm-ready}. This postMessage with the
sole content ``ready'' triggers RPdoc to send the
$\mi{tagKey}$ to FWDdoc, where in the postMessage the
origin\footnote{An origin is defined by a domain name plus
  the information whether the connection to this domain is
  via HTTP or HTTPS.} of FWD with HTTPS is declared to be
the only allowed receiver of this
message~\refprotostep{pm-tag-key}. FWDdoc uses the key to
decrypt the tag and thereby learns the intended receiver
(RP) of the IA~\refprotostep{decrypt-tag}. As its last
action, FWD forwards the encrypted IA $\mi{eia}$ via postMessage to RPdoc
(using RP's HTTPS origin as the only allowed
receiver)~\refprotostep{pm-eia}.

RPdoc receives $\mi{eia}$ and sends it along with the
$\mi{loginSessionToken}$ to
RP~\refprotostep{post-login}. RP then decrypts $\mi{eia}$,
retrieves $\mi{ia'}$ and checks whether $\mi{ia'}$ is a
valid signature for $\an{\mi{tag}, \mi{email},
  \mi{FWDDomain}}$ under the verification key
$\pub(k_{IdP})$ of the IdP, where $\mi{tag}$, $\mi{email}$,
and $\mi{FWDDomain}$ are taken from the session data
identified by
$\mi{loginSessionToken}$~\refprotostep{decrypt-ia}.

Now, the user identified by the email address is logged in.
The mechanism that is used to persist this logged-in state
(if any) at this point is out of the scope of \spresso. In
our analysis, as a model for a standard
session-based login, we assume that RP creates a
session for the user's browser, identified by some freshly
chosen token (the \emph{RP service
  token})~\refprotostep{choose-service-token} and sends
this token to the browser~\refprotostep{send-login}.

\subsection{Implementation Details}\label{sec:implementationdetails}
We developed a proof-of-concept implementation of \spresso
in about 700 lines of JavaScript and HTML code. It contains
all presented features of \spresso itself and a typical
IdP. The implementation (source code and online demo) is available at
\cite{spresso-sourcecode-2015}. Our model presented in
Section~\ref{sec:formal-model-spresso} closely follows this
implementation.

The three servers (RP, IdP and FWD) are written in
JavaScript and are based on node.js and its built-in
\emph{crypto} API. On the client-side we use the Web
Cryptography API. For encryption we employ AES-256 in 
GCM mode to provide authenticity. Signatures are
created/verified using RSA-SHA256.

\subsection{Discussion}
\label{sec:discussion}

\noindent In order to provide more intuition and motivation
for the design of \spresso, and in particular its security and
privacy properties, we first informally discuss some
potential attacks on our system and what measures we took
when designing and implementing \spresso to prevent these
attacks. These attacks also illustrate the complexity and
difficulty of designing a secure and privacy-respecting
web-based SSO system. In Sections~\ref{sec:proving-privacy}
and~\ref{sec:prov-auth}, we formally prove that \spresso
provides strong authentication and privacy properties in a
detailed model of the web infrastructure. We also discuss
other aspects of \spresso, including usability and
performance. We conclude this section with a comparison of
\spresso and BrowserID.

\subsubsection{Malicious RP: Impersonation Attack.}
An attacker could try to launch a man in the middle attack
against \spresso by playing the role of an RP (RP server
and RPdoc) to the user. Such an attacker would run a
malicious server at his RP domain, say, $\str{RPa}$, and
also deliver a malicious script (instead of the honest
RPdoc script) to the user's browser. Now assume that the
user wants to log in with her email address at $\str{RPa}$
and is logged in at the IdP corresponding to the email
address already.  Then, the attacker (outside of the user's
browser) could first initiate the login process at
$\str{RPb}$ using the user's email address. The attacker's
RP could then create a tag of the form
$\encs{\an{\str{RPb}, \mi{rpNonce}}}{\mi{tagKey}}$ using
the domain of an honest RP $\str{RPb}$, instead of
$\str{RPa}$.  The IdP would hence create an IA for this tag
and the user's email address and deliver this IA to the
user's browser. If this IA were now indeed be delivered to
the attacker's RP window (which is running a malicious
RPdoc script), the attacker could use the IA to finish the
log in process at $\str{RPb}$ (and obtain the service token
from $\str{RPb}$), and thus, log in at $\str{RPb}$ as the
honest user.

However, assuming that FWD is honest (see below for a
discussion of malicious FWDs), FWD prevents this kind of
attack: FWD forwards the (encrypted) IA via a postMessage
only to the domain listed in the tag (so, in this case,
$\str{RPb}$), which in the attack above is not the domain
of the document loaded in the attacker's RP window
($\str{RPa}$). The IA is therefore not transmitted to the
attacker. The same applies when the attacker tries to
navigate the RP window to its own domain, i.e., to
$\str{RPa}$, before Step~\refprotostep{pm-eia}. Our formal
analysis presented in the following sections indeed proves
that such attacks are excluded in \spresso. We note that in
order to make sure that the postMessage is delivered to the
correct RP window (technically, a window with the expected
origin), FWD uses a standard feature of the postMessage
mechanism which allows to specify the origin of the
intended recipient of a postMessage.

\subsubsection{Malicious IdP.} A malicious IdP could try to
log the user in under an identity that is not her own. An
attack of this kind on BrowserID was shown in~\cite{FettKuestersSchmitz-SP-2014}. However, in \spresso,
the IdP cannot select or alter the identity with which the
user is logged in. Instead, the identity is fixed by RP
after Step~\refprotostep{create-tag} and checked in
Step~\refprotostep{decrypt-ia}. Again, our formal analysis
shows that such attacks are indeed not possible in \spresso.

The IdP could try to undermine the user's privacy by trying
to find out which RP requests the IA. However, in \spresso,
the IdP cannot gather such information: From the
information available to it ($\mi{email}$, $\mi{tag}$,
$\mi{FWDDomain}$ plus any information it can gather from
the browser's state), it cannot infer the RP.\footnote{If
  only a few RPs use a specific FWD, $\mi{FWDDomain}$ would
  reveal some information. However, this is easy to avoid
  in practice: the set of FWDs all (or many) RPs trust
  should be big enough and RPs could randomly choose one of
  these FWDs for every login process.} It could further try
to corellate the sources and times of HTTPS requests for
the support document with user logins. To minimize this
side channel, we suggest caching the support
document at each RP and automatic refreshing of this cache
(e.g., an RP could cache the document for 48 hours and
after that period automatically refresh the cache).
Additionally, RPs should use the Tor network (or similar
means) when retrieving the support document in order to
hide their IP addresses. Assuming that support documents
have been obtained from IdPs independently of specific
login requests by users, our formal analysis shows that
\spresso in fact enjoys a very strong privacy property (see
Sections~\ref{sec:form-defin-priv} and
\ref{sec:proving-privacy}).

In BrowserID, malicious IdPs (in fact, any party who can
run malicious scripts in the user's browser) can check the
presence or absence of certain iframes in the login
process, leading to the privacy break mentioned earlier.
Again, our formal analysis implies that this is not
possible for \spresso.

\subsubsection{Malicious FWD.} A malicious FWD could
cooperate with or act as a malicious RP and thereby enable
the man in the middle attack discussed above, undermining
the authentication guarantees of the system. Also, a
malicious FWD could collaborate with a malicious IdP and
send information about the RP to the IdP, and hence,
undermine privacy. 

Therefore, for our system to provide authentication and
privacy, we require that FWDs behave honestly. Below we
discuss ways to force FWD to behave honestly. We suspect
that there is no way to avoid the use of FWDs or other
honest components in a practical SSO system which is
supposed to provide not only authentication but also
privacy: In our system, after
Step~\refprotostep{send-loginxhr} of the flow, IdPdoc must
return the IA to the RP. There are two constraints: First,
the IA should only be forwarded to a document that in fact
is RP's document. Otherwise, it could be misused to log in
at RP under the user's identity by any other party, which
would break authentication. Second, RP's identity should
not be revealed to IdP, which is necessary for
privacy. Currently, there is no browser mechanism to
securely forward the IA to RP without disclosing RP's
identity to IdP (but see below).

\subsubsection{Enforcing Honest FWDs.} 
Before we discuss existing and upcoming technologies to
enforce honest behavior of FWDs, we first note that in
\spresso, an FWD is chosen by the RP to which a user wants
to log in. So the RP can choose the FWD it trusts. The RP
certainly has a great interest in the trustworthiness of
the FWD: As mentioned, a malicious FWD could allow an
attacker to log in as an honest user (and hence, misuse
RP's service and undermine confidentiality and integrity of
the user's data stored at RP), something an RP would
definitely want to prevent. Second, we also note that FWD
does not learn a user's email address: the IA, which is
given to FWD and which contains the user's email address,
is encrypted with a symmetric key unknown to
FWD.\footnote{We note that IA is a signature anyway, so
  typically a signed hash of a message. Hence, for common
  signature schemes, already from the IA itself FWD is not
  able to extract the user's email address. In addition,
  \spresso even encrypts the IA to make sure that this is
  the case no matter which signature scheme is used.}
Therefore, \spresso does not provide FWD with information
to track at which RP a specific user logs
in.\footnote{A malicious FWD could try to set cookies and
  do browser fingerprinting to the track the behavior of
  specific browsers. Still it does not obtain the user's
  email address.}

Now, as for \emph{enforcing} honest FWDs, first note that
an honest FWD server is supposed to always deliver the same
fixed JavaScript to a user's browsers. This JavaScript code
is very short (about 50 lines of code). If this code is
used, it is not only ensured that FWD preserves
authentication and privacy, but also that no tracking data
is sent back to the FWD server.

Using current technology, a user could use a browser
extension which again would be very simple and which would
make sure that in fact only this specific JavaScript is
delivered by FWD (upon the respective request). As a
result, FWD would be forced to behave honestly, without the
user having to trust FWD. Another approach would be an
extension that replaces FWD completely, which could also
lead to a simplified protocol. In both cases, \spresso
would provide authentication and privacy without having to
trust any FWD. Both solutions have the common problem that
they do not work on all platforms, because not on all
platforms browsers support extensions. The first solution
(i.e., the extension checks only that correct JavaScript is
loaded) would at least still work for users on such
platforms, albeit with reduced security and privacy
guarantees.

A native web technology called \emph{subresource integrity} (SRI)\footnote{\url{http://www.w3.org/TR/SRI/}} is
currently under development at the W3C\@. SRI allows a
document to create an iframe with an attribute
\emph{integrity} that takes a hash value. The browser now
would guarantee that the document loaded into the iframe
hashes to exactly the given value. So, essentially the
creator of the iframe can enforce the iframe to be loaded
with a specific document. This would enable \spresso to
automatically check the integrity of FWDdoc without any
extensions.

\subsubsection{Referer Header and Privacy.}
The \emph{Referer} [sic!] header is set by browsers to show
which page caused a navigation to another page. It is set
by all common browsers. To preserve privacy, when the
loading of IdPdoc is initiated by RPdoc, it is important
that the Referer header is \emph{not} set, because it would
contain RP's domain, and consequently, IdP would be able to
read off from the Referer header to which RP the user wants
to log in, and hence, privacy would be broken.  With HTML5,
a special attribute for links in HTML was introduced, which
causes the Referer header to be suppressed
(\texttt{rel="noreferrer"}). However, when such a link is
used to open a new window, the new window does not have a
handle on the opening window (opener) anymore. But having a
handle is essential for \spresso, as the postMessage in
Step~\refprotostep{pm-ready} is sent to the opener window
of IdPdoc. To preserve the opener handle while at the same
time hiding the referer, we first open the new window with
a redirector document loaded from RP
(Step~\refprotostep{open-IdPdoc}) and then navigate this
window to IdPdoc (using a link with the noreferrer
attribute set and triggered by JavaScript in
Step~\refprotostep{redir-to-IdPdoc}). This causes the
Referer header to be cleared, while the opener handle is
preserved.\footnote{Another option would have been to use a
  data URI instead of loading the redirector document from
  RPdoc and to use a Refresh header contained in a meta tag
  for getting rid of the Referer header. This however
  showed worse cross-browser compatibility, and the Refresh
  header lacks standardization.}  Our formal analysis
implies that with this solution indeed privacy is
preserved.

\subsubsection{Cross-Site Request Forgery.}
Cross-site request forgery is particularly critical at RP,
where it could be used to log a user in under an identity
that is not her own. For RP, \spresso therefore employs a
session token that is not stored in a cookie, but only in
the state of the JavaScript, avoiding cross-origin and
cross-domain cookie attacks. Additionally, RP checks the
Origin header of the login request to make sure that no
login can be triggered by a third party (attacker) web
page. Our formal analysis implies that cross-site request
forgery and related attacks are not possible in \spresso.

\subsubsection{Phishing.}
It is important to notice that in \spresso the user can
verify the location and TLS certificate of IdPdoc's window
by checking the location bar of her browser. The user can
therefore check where she enters her password, which would
not be possible if IdPdoc was loaded in an iframe. Setting
strict transport security headers can further help in
avoiding phishing attacks.

\subsubsection{Tag Length Side Channel.} The length of the
tag created in Step~\refprotostep{create-tag} depends on
the length of $\str{RPDomain}$. Since the tag is given to
IdP, IdP might try to infer $\str{RPDomain}$ from the
length of the tag. However, according to RFC 1035, domain
names may at most be 253 characters long. Therefore, by
appropriate padding (e.g., encrypting always nine 256 Bit
plaintext blocks)\footnote{Eight 256 bit blocks are
  sufficient for all domain names. We need an additional
  block for $\mi{rpNonce}$.}  the length of the tag will
not reveal any information about $\str{RPDomain}$.

\subsubsection{Performance.}
\spresso uses only standard browser features, employs only
symmetric encryption/decryption and signatures, and
requires (in a minimal implementation) eight HTTPS
re\-quests/re\-sponses --- all of which pose no significant
performance overhead to any modern web application, neither
for the browser nor for any of the servers. In our
prototypical and unoptimized implementation, a login
process takes less than 400\,ms plus the time for entering
email address and password.

\subsubsection{Usability.}
In \spresso, users are identified by their email addresses
(an identifier many users easily memorize) and email
pro\-vid\-ers serve as identity pro\-vid\-ers. Many web
applications today already use the email address as the
primary identifier along with a password for the specific
web site: When a user signs up, a URL with a secret token
is sent to the user's email address. The user has to check
her emails and click on the URL to confirm that she has
control over the email address. She also has to create a
password for this web site. \spresso could seamlessly be
integrated into this sign up scheme and greatly simplify
it: If the email provider (IdP) of the user supports
\spresso, an \spresso login flow can be launched directly
once the user entered her email address and clicked on the
login button, avoiding the need for a new user password and
the email confirmation; and if the user is logged in at the
IdP already, the user does not even have to enter a
password. Otherwise, or if a user has JavaScript disabled,
an automatic and seamless fallback to the classical
token-based approach is possible (as RP can detect whether
the IdP supports \spresso in
Step~\refprotostep{get-spresso-info-1} of the protocol). In
contrast to other login systems, such as Google ID, the user would not
even have to decide whether to log in with \spresso or not
due to the described seamless integration of \spresso. Due
to the privacy guarantees (which other SSO systems do not
have), using \spresso would not be disadvantageous for the
user as her IdPs cannot track to which RPs the user logs
in.

The above illustrates that, using \spresso, signing up to a
web site is very convenient: The user just enters her email
address at the RP's web site and presses the login button
(if already logged in at the respective IdP, no password is
necessary). Also, with \spresso the user is free to use any
of her email addresses.

\subsubsection{Extendability.} \spresso could be extended
to have the IdP sign (in addition to the email address)
further user attributes in the IA, which then might be
used by the RP.

\subsubsection{Operating FWD.} Operating an FWD is very
cheap, as the only task is to serve one static file.  Any
party can act as an FWD. Users and RPs might feel most
confident if an FWD is operated by widely trusted
non-profit organizations, such as Mozilla or the EFF.

\subsubsection{Comparison with BrowserID.}
BrowserID was the first and so far only SSO system designed
to provide privacy (IdPs should not be able to tell at
which RPs user's log in). Nonetheless, as already mentioned
(see Section~\ref{sec:main-features}), severe attacks were
discovered in
\cite{FettKuestersSchmitz-ESORICS-BrowserID-Primary-2015} which
show that the privacy promise of BrowserID is broken: not
only IdPs but even other parties can track the login
behavior of users. Regaining privacy would have required a
major redesign of the system, resulting in essentially a
completely new system, as pointed out in
\cite{FettKuestersSchmitz-ESORICS-BrowserID-Primary-2015}. Also,
BrowserID has the disadvantage that it relies on a single
trusted server (\nolinkurl{login.persona.org}) which is
quite complex, with several server interactions necessary
in every login process, and most importantly, by design,
gets full information about the login behavior of users
(the user's email address and the RP at which the user
wants to log in).\footnote{In \spresso, we require that FWD
  behaves honestly. In a login process, however, the FWD
  server needs to provide only a fixed single and very
  simple JavaScript, no further server interaction is
  necessary.  Also, FWD does not get full information and
  RP in every login process may choose any FWD it
  trusts. Moreover, as discussed above, there are means to
  force FWD to provide the expected JavaScript.} Finally,
BrowserID is a rather complex SSO system (with at least 64
network and inter-frame messages in a typical login
flow\footnote{Counting HTTP request and responses as well
  as postMessages, leaving out any user requests for GUI
  elements or other non-necessary resources.} compared to
only 19 in \spresso). This complexity implies that security
vulnerability go unnoticed more easily. In fact, several
attacks on BrowserID breaking authentication and privacy
claims were discovered (see
\cite{FettKuestersSchmitz-SP-2014,FettKuestersSchmitz-ESORICS-BrowserID-Primary-2015}).

This is why we designed and built \spresso from scratch,
rather than trying to redesign BrowserID. The design of
\spresso is in fact very different to BrowserID. For example,
except for HTTPS and signatures of IdPs, \spresso uses only
symmetric encryption, whereas in BrowserID, users (user's
browers) have to create public/private key pairs and IdPs
sign the user's public keys. The entities in \spresso are
different to those in BrowserID as well, e.g., \spresso does
not rely on the mentioned single, rather complex, and
essentially omniscient trusted party, resulting in a
completely different protocol flow. The design of \spresso is
much slimmer than the one of BrowserID.

%
\section{Web Model}
\label{sec:web-model}

\begin{sloppypar}
Our formal security analysis of \spresso (presented in the
next sections) is based on the general Dolev-Yao style web
model in~\cite{FettKuestersSchmitz-SP-2014}. As mentioned
in the introduction, we changed some details in this model
to facilitate the definition of
in\-dis\-tin\-guish\-abil\-i\-ty/pri\-va\-cy properties (see
Section~\ref{sec:form-defin-priv}). In particular, we
simplified the handling of nonces and removed
non-deterministic choices wherever possible. Also, we added
the HTTP \emph{Referer} header and the HTML5 \emph{noreferrer}
attribute for links.
\end{sloppypar}

Here, we only present a very brief version of the web
model. The full model, including our changes, is provided
in
Appendices~\ref{app:web-model}--\ref{app:deta-descr-brows}.

\subsection{Communication
  Model}\label{sec:communicationmodel}

The main entities in the communication model are
\emph{atomic processes}, which are used to model web
browsers, web servers, DNS servers as well as web and
network attackers. Each \ap listens to one or more (IP)
addresses. A set of \aps forms what is called a
\emph{system}. Atomic processes can communicate via events,
which consist of a message as well as a receiver and a
sender address. In every step of a run, one event is chosen
non-deterministically from the current ``pool'' of events
and is delivered to one of the \aps that listens to the
receiver address of that event. The atomic process can then
process the event and output new events, which are added to
the pool of events, and so on. More specifically, messages,
processes, etc.~are defined as follows.

\subsubsection{Terms, Messages and Events.} 
As usual in Dolev-Yao models (see, e.g.,
\cite{AbadiFournet-POPL-2001}), messages are expressed as
formal terms over a signature. The signature $\Sigma$ for
the terms and messages considered in the web model
contains, among others, constants (such as (IP) addresses,
ASCII strings, and nonces), sequence and projection
symbols, and further function symbols, including those for
(a)symmetric encryption/decryption and digital signatures.
Messages are defined to be ground terms (terms without
variables). For example (see also
Section~\ref{sec:login-flow} where we already use the term
notation to describe messages), $\pub(k)$ denotes the
public key which belongs to the private key $k$. To provide
another example of a message, in the web model, an HTTP
request is represented as a ground term containing a nonce,
a method (e.g., $\mGet$ or $\mPost$), a domain name, a
path, URL parameters, request headers (such as
$\str{Cookie}$), and a message body. For instance, an HTTP
$\mGet$ request for the URL \url{http://example.com/show?p=1} is
modeled as the term
\[\mi{r} := \langle \cHttpReq, n_1, \mGet, \str{example.com},
\str{/show}, \an{\an{\str{p},1}}, \an{}, \an{} \rangle\,,\]
where headers and body are empty. An HTTPS request for $r$
is of the form
$\ehreqWithVariable{r}{k'}{\pub(k_\text{example.com})}$, where
$k'$ is a fresh symmetric key (a nonce) generated by the
sender of the request (typically a browser); the responder
is supposed to use this key to encrypt the response.

\emph{Events} are terms of the form $\an{a,f,m}$ where $a$
and $f$ are receiver/sender (IP) addresses, and $m$ is a
message, for example, an HTTP(S) message as above or a DNS
request/response.

The \emph{equational theory} associated with the signature
$\Sigma$ is defined as usual in Dolev-Yao models. The
theory induces a congruence relation $\equiv$ on terms. It
captures the meaning of the function symbols in
$\Sigma$. For instance, the equation in the equational
theory which captures asymmetric decryption is $\dec{\enc
  x{\pub(y)}}{y}=x$. With this, we have that, for example,
\[\dec{\ehreqWithVariable{r}{k'}{\pub(k_\text{example.com})}}{k_\text{example.com}}\equiv
\an{r,k'}\,,\] i.e., these two terms are
equivalent w.r.t.~the equational theory.

\subsubsection{Atomic Processes, Systems and Runs.} Atomic
Dolev-Yao processes, systems, and runs of systems are
defined as follows.

An \emph{atomic Dolev-Yao (DY) process} is a tuple \[p =
(I^p, Z^p, R^p, s^p_0)\] where $I^p$ is the set of addresses
the process listens to, $Z^p$ is a set of states (formally,
terms), $s^p_0\in Z^p$ is an initial state, and $R^p$ is a
relation that takes an event and a state as input and
(non-deterministically) returns a new state and a sequence
of events. This relation models a computation step of the
process, which upon receiving an event in a given state
non-deterministically moves to a new state and outputs a
set of events. It is required that the events and states in
the output can be computed (more formally, derived in the
usual Dolev-Yao style) from the current input event and
state. We note that in \cite{FettKuestersSchmitz-SP-2014}
the definition of an atomic process also contained a set of
nonces which the process may use. Instead of such a set, we
now consider a global sequence of (unused) nonces and new
nonces generated by an atomic process are taken from this
global sequence.

The so-called \emph{attacker process} is an atomic DY
process which records all messages it receives and outputs
all events it can possibly derive from its recorded
messages. Hence, an attacker process is the maximally
powerful DY process. It carries out all attacks any DY
process could possibly perform and is parametrized by the
set of sender addresses it may use. Attackers may corrupt
other DY processes (e.g., a browser).

A \emph{system} is a set of atomic processes. A
\emph{configuration} $(S,E,N)$ of this system consists of the current
states of all atomic processes in the system ($S$), the pool of
waiting events ($E$, here formally modeled as a sequence of
events; in \cite{FettKuestersSchmitz-SP-2014}, the pool was
modeled as a multiset), and the mentioned sequence of
unused nonces ($N$).

A \emph{run} of a system for an initial sequence of events
$E^0$ is a sequence of configurations, where each
configuration (except for the initial one) is obtained by
delivering one of the waiting events of the preceding
configuration to an atomic process $p$ (which listens to
the receiver address of the event), which in turn performs
a computation step according to its relation $R^p$. The
initial configuration consists of the initial states of the
atomic processes, the sequence $E^0$, and an initial
infinite sequence of unused nonces.

\subsubsection{Scripting Processes.}
The web model also defines scripting processes, which model
client-side scripting technologies, such as JavaScript.

A \emph{scripting process} (or simply, a \emph{script}) is
defined similarly to a DY process. It is called by the
browser in which it runs. The browser provides it with
state information $s$, and the script then, according to
its computation relation, outputs a term $s'$, which
represents the new internal state and some command which is
interpreted by the browser (see also below). Again, it is
required that a script's output is derivable from its
input.

Similarly to an attacker process, the so-called
\emph{attacker script} $\Rasp$ may output everything that
is derivable from the input.

\subsection{Web System}\label{sec:websystem}

A web system formalizes the web
infrastructure and web applications. Formally, a \emph{web
  system} is a tuple \[(\websystem,
\scriptset,\mathsf{script}, E^0)\] with the following
components:

\begin{itemize}
\item The first component, $\websystem$, denotes a system
  (a set of DY processes as defined above) and contains
  honest processes, web attacker, and network attacker
  processes. While a web attacker can listen to and send
  messages from its own addresses only, a network attacker
  may listen to and spoof all addresses (and therefore is
  the maximally powerful attacker). Attackers may corrupt
  other parties. In the analysis of a concrete web system,
  we typically have one network attacker only and no web
  attackers (as they are subsumed by the network attacker),
  or one or more web attackers but then no network
  attacker. Honest processes can either be web browsers,
  web servers, or DNS servers. The modeling of web servers
  heavily depends on the specific application. The web
  browser model, which is independent of a specific web
  application, is presented below.
\item The second component, $\scriptset$, is a finite set
  of scripts, including the attacker script $\Rasp$. In a
  concrete model, such as our \spresso model, the set
  $\scriptset\setminus\{\Rasp\}$ describes the set of
  honest scripts used in the web application under
  consideration while malicious scripts are modeled by the
  ``worst-case'' malicious script, $\Rasp$.
\item The third component, $\mathsf{script}$, is an
  injective mapping from a script in $\scriptset$ to its
  string representation $\mathsf{script}(s)$ (a constant in
  $\Sigma$) so that it can be part of a messages, e.g., an
  HTTP response. 
\item Finally, $E^0$ is a sequence of events, which always
  contains an infinite number of events of the form
  $\an{a,a,\trigger}$ for every IP address $a$ in the web
  system. 
\end{itemize}
A \emph{run} of the web system is a run of
  $\websystem$ initiated by $E^0$.

\subsection{Web Browsers}\label{sec:web-browsers}

We now sketch the model of the web browser, with full
details provided in Appendix~\ref{app:deta-descr-brows}. A web browser is modeled as a
DY process $(I^p, Z^p, R^p, s^p_0)$.

An honest browser is thought to be used by one honest user,
who is modeled as part of the browser. User actions are
modeled as non-deterministic actions of the web browser.
For example, the browser itself non-deterministically
follows the links in a web page. User data (i.e., passwords
and identities) is stored in the initial state of the
browser and is given to a web page when needed, similar to
the AutoFill feature in browsers.

Besides the user identities and passwords, the state of a
web browser (modeled as a term) contains a tree of open
windows and documents, lists of cookies, localStorage and
sessionStorage data, a DNS server address, and other data.

In the browser state, the $\mi{windows}$ subterm is the
most complex one. It contains a window subterm for every
open window (of which there may be many at a time), and
inside each window, a list of documents, which represent
the history of documents that have been opened in that
window, with one of these documents being active, i.e.,
this document is presented to the user and ready for
interaction. A document contains a script loaded from a web
server and represents one loaded HTML page. A document also
contains a list of windows itself, modeling
iframes. Scripts may, for example, navigate or create
windows, send \xhrs and postMessages, submit forms,
set/change cookies, localStorage, and sessionStorage data,
and create iframes.  When activated, the browser provides a
script with all data it has access to, such as a (limited)
view on other documents and windows, certain cookies as
well as localStorage and sessionStorage.

Figure~\ref{fig:browser-structure} shows a brief overview
of the browser relation $R^p$ which defines how browsers
behave. For example, when a $\str{TRIGGER}$ message is
delivered to the browser, the browser non-deterministically
choses an $\mi{action}$. If, for instance, this action is
1, then an active document is selected
non-deterministically, and its script is triggered. The
script (with inputs as outlined above), can now output a
command, for example, to follow a hyperlink
($\str{HREF}$). In this case, the browser will follow this
link by first creating a new DNS request.  Once a response
to that DNS request arrives, the actual HTTP request (for
the URL defined by the script) will be sent out. After a
response to that HTTP request arrives, the browser creates
a new document from the contents of the response. Complex
navigation and security rules ensure that scripts can only
manipulate specific aspects of the browser's
state. Browsers can become corrupted, i.e., be taken over
by web and network attackers.  The browser model comprises
two types of corruption: \emph{close-corruption}, modeling
that a browser is closed by the user, and hence, certain
data is removed (e.g., session cookies and opened windows),
before it is taken over by the attacker, and \emph{full
  corruption}, where no data is removed in advance. Once
corrupted, the browser behaves like an attacker process.

\begin{figure}[tb]
  \centering
  \begin{minipage}{0.5\linewidth}
    \footnotesize
    { \underline{\textsc{Processing Input Message
          \hlExp{$m$}}}
      \begin{itemize}[itemsep=0.2ex,label=,leftmargin=0pt]
      \item $m = \fullcorrupt$:
        $\mi{isCorrupted} := \fullcorrupt$
      \item $m = \closecorrupt$:
        $\mi{isCorrupted} := \closecorrupt$
      \item $m = \trigger$:
        non-deterministically choose $\mi{action}$ from
        $\{1,2,3\}$
        \begin{itemize}[itemsep=0.1ex,label=,leftmargin=1em]
        \item $\mi{action} = 1$: \parbox[t]{20em}{Call
          script
          of some active document.\\ Outputs 
          new state and $\mi{command}$.}
          \begin{itemize}[itemsep=0.4ex,label=,leftmargin=1em]
          \item $\mi{command} = \tHref$:
            $\rightarrow$ \emph{Initiate request}
          \item $\mi{command} = \tIframe$:
            Create subwindow, $\rightarrow$ \emph{Initiate
              request}
          \item $\mi{command} = \tForm$:
            $\rightarrow$ \emph{Initiate request}
          \item $\mi{command} = \tSetScript$:
            Change script in given document.
          \item $\mi{command} =
                \tSetScriptState$: \parbox[t]{10em}{Change state of script\\
            in given document.}
          \item $\mi{command} =
                \tXMLHTTPRequest$: $\rightarrow$
            \emph{Initiate request}
          \item $\mi{command} = \tBack$ or $\tForward$: Navigate given window.
          \item $\mi{command} = \tClose$:
            Close given window.
          \item $\mi{command} = \tPostMessage$: Send \pm to specified document.
          \end{itemize}
        \item $\mi{action} = 2$:
          $\rightarrow$ \emph{Initiate request to some URL
            in new window}
        \item $\mi{action} = 3$:
          $\rightarrow$ \emph{Reload some document}
        \end{itemize}
    \item $m=$ DNS response: send
      corresponding HTTP request
    \item $m=$ HTTP(S) response:
      (decrypt,) find reference.
      \begin{itemize}[itemsep=0.2ex,label=,leftmargin=1em]
      \item\emph{reference to window:} create document
        in window
      \item \emph{reference to
          document:} \parbox[t]{15em}{add response
          body to document's\\ script input}
      \end{itemize}
    \end{itemize}
  }
\end{minipage}\vspace{-0.5em}
  \caption{The basic structure of the web browser relation
    $R^p$ with an extract of the most important processing
    steps, in the case that the browser is not already
    corrupted.}\label{fig:browser-structure}
\end{figure}

%
\section{Indistinguishability of Web Systems}
\label{sec:form-defin-priv}

We now define the indistinguishability of web systems. This
definition is not tailored towards a specific web application,
and hence, is of independent interest.

Our definition follows the idea of trace equivalence in
Dolev-Yao models (see, e.g.,
\cite{ChevalComonLundhDelaune-CCS-2011}), which in turn is
an abstract version of cryptographic indistinguishability. 

Intuitively, two web systems are indistinguishable if the
following is true: whenever the attacker performs the same
actions in both systems, then the sequence of messages he
obtains in both runs look the same from the attacker's
point of view, where, as usual in Dolev-Yao models, two
sequences are said to ``look the same'' when they are
statically equivalent \cite{AbadiFournet-POPL-2001} (see
below). More specifically, since, in general, web systems
allow for non-deterministic actions (also of honest
parties), the sequence of actions of the attacker might
induce a set of runs. Then indistinguishability says that
for all actions of the attacker and for every run induced
by such actions in one system, there exists a run in the
other system, induced by the same attacker actions, such
that the sequences of messages the attacker obtains in both
runs look the same to the attacker.

Defining the actions of attackers in web systems requires
care because the attacker can control different components
of such a system, but some only partially: A web attacker
(unlike a network attacker) controls only part of the
network. Also an attacker might control certain servers
(web servers and DNS servers) and browsers. Moreover, he
might control certain scripts running in honest browsers,
namely all attacker scripts $\Rasp$ running in browsers;
dishonest browsers are completely controlled by the
attacker anyway.

We model a single action of the attacker by what we call a
\emph{(web system) command}; not to be confused with
commands output by a script to the browser.  A command is
of the form
\[\an{i,j,\tau_\text{process},\mi{cmd}_\text{switch},\mi{cmd}_\text{window},\tau_\text{script},\mi{url}}\,.\] The
first component $i \in \mathbb{N}$ determines which event
from the pool of events is processed. If this event could
be delivered to several processes (recall that a network
attacker, if present, can listen to all addresses), then
$j$ determines the process which actually gets to process
the event. Now, there are different cases depending on the
process to which the event is delivered and depending on
the event itself. We denote the process by $p$ and the
event by $e$: i) If $p$ is corrupted (it is a web attacker,
network attacker, some corrupted browser or server), then
the new state of this process and its output are determined
by the term $\tau_\text{process}$, i.e., this term is
evaluated with the current state of the process and the
input $e$. ii) If $p$ is an honest browser and $e$ is not a
trigger message (e.g., a DNS or HTTP(S) response), then the
browser processes $e$ as usual (in a deterministic
way). iii) If $p$ is an honest browser and $e$ is a trigger
message, then there are three actions a browser can
(non-deterministically) choose from: open a new window,
reload a document, or run a script. The term
$\mi{cmd}_\text{switch} \in \{1,2,3\}$ selects one of these
actions. If it chooses to open a new window, a document
will be loaded from the URL $\mi{url}$. In the remaining
two cases, $\mi{cmd}_\text{window}$ determines the window
which should be reloaded or in which a script is
executed. If a script is executed and this script is the
attacker script, then the output of this script is derived
(deterministically) by the term $\tau_\text{script}$, i.e.,
this term is evaluated with the data provided by the
browser. The resulting command, if any, is processed
(deterministically) by the browser. If the script to be
executed is an honest script (i.e., not $\Rasp$), then this
script is evaluated and the resulting command is processed
by the browser. (Note that the script might perform
non-deterministic actions.) iv) If $p$ is an honest process
(but not a browser), then the process evaluates $e$ as
usual. (Again, the computation might be non-deterministic,
as honest processes might be non-deterministic.)

We call a finite sequence of commands a
\emph{schedule}. Given a web system
$\completewebsystem=(\websystem,
\scriptset,\mathsf{script}, E^0)$, a schedule $\sigma$
induces a set of (finite) runs in the obvious way. We
denote this set by
$\sigma(\completewebsystem)$. Intuitively, a schedule
models the attacker actions in a run. Note that we consider
a very strong attacker. He not only determines the actions
of all dishonest processes and all attacker scripts, but
also schedules all events, not only events intended for the
attacker; clearly, the attacker does not get to see
explicitly events not intended for him.

Before we can define indistinguishability of two web
systems, we need to, as mentioned above, recall the
definition of static equivalence of two messages $t_1$ and
$t_2$. We say that the messages $t_1$ and $t_2$ are
\emph{statically equivalent}, written $t_1\approx t_2$, if
and only if, for all terms $M(x)$ and $N(x)$ which contain
one variable $x$ and do not use nonces, we have that
$M(t_1)\equiv N(t_1)$ iff $M(t_2)\equiv N(t_2)$. That is,
every test performed by the attacker yields the same result
for $t_1$ and $t_2$, respectively. For example, if $k$ and
$k'$ are nonces, and $r$ and $r'$ are
different constants, then
\[\ehreqWithVariable{r}{k'}{\pub(k)}\approx
\ehreqWithVariable{r'}{k'}{\pub(k)}\ .\]
Intuitively, this is the case because the attacker does not
know the private key $k$.

We also need the following terminology.  If
$(\websystem, \scriptset,\mathsf{script}, E^0)$ is
a web system and $p$ is an attacker process in
$\websystem$, then we say that $(\websystem,
\scriptset,\mathsf{script}, E^0,p)$ is a \emph{web
  system with a distinguished attacker process $p$}. If
$\rho$ is a finite run of this system, we denote by
$\rho(p)$ the state of $p$ at the end of this run. In our
indistinguishability definition, we will consider the state
of the distinguished attacker process only. This is
sufficient since the attacker can send all its data to this
process.

Now, we are ready to define indistinguishability of web
systems in a natural way.

\begin{definition}\label{def:indistwebsystems}
  Let $\completewebsystem_0$ and $\completewebsystem_1$ be
  two web system each with a distinguished attacker process
  $p_0$ and $p_1$, respectively. We say that these systems
  are \emph{indistinguishable}, written
  $\completewebsystem_0\approx \completewebsystem_1$, iff
  for every schedule $\sigma$ and every $i\in \{0,1\}$, we
  have that for every run $\rho\in
  \sigma(\completewebsystem_i)$ there exists a run
  $\rho'\in \sigma(\completewebsystem_{1-i})$ such that
  $\rho(p_i)\approx \rho'(p_{1-i})$.
\end{definition}

%

\section{Formal Model of \spresso}
\label{sec:formal-model-spresso}

We now present the formal model of \spresso, which closely
follows the description in
Section~\ref{sec:description-spresso} and the
implementation of the system. This model is the basis for
our formal analysis of privacy and authentication
properties presented in Sections~\ref{sec:proving-privacy}
and \ref{sec:prov-auth}. 

We model \spresso as a web system (in the sense of
Section~\ref{sec:websystem}).  We call
$\spressowebsystem=(\bidsystem, \scriptset,
\mathsf{script}, E^0)$ an \emph{\spresso web system} if it
is of the form described in what follows.

The set $\bidsystem=\mathsf{Hon}\cup \mathsf{Web} \cup
\mathsf{Net}$ consists of a finite set of web attacker
processes (in $\mathsf{Web}$), at most one
network attacker process (in $\mathsf{Net}$), a finite set
$\fAP{FWD}$ of forwarders, a finite set $\fAP{B}$ of web
browsers, a finite set $\fAP{RP}$ of web servers for the
relying parties, a finite set $\fAP{IDP}$ of web servers
for the identity providers, and a finite set $\fAP{DNS}$ of
DNS servers, with $\mathsf{Hon} := \fAP{B} \cup \fAP{RP}
\cup \fAP{IDP} \cup \fAP{FWD} \cup \fAP{DNS}$.
Figure~\ref{fig:scripts-in-w} shows the set of scripts $\scriptset$
and their respective string representations that are defined by the
mapping $\mathsf{script}$.
The set $E^0$ contains only the trigger events as specified in
Section~\ref{sec:websystem}.

\begin{figure}[tb]
  \centering
  \begin{tabular}{|@{\hspace{1ex}}l@{\hspace{1ex}}|@{\hspace{1ex}}l@{\hspace{1ex}}|}\hline 
    \hfill $s \in \scriptset$\hfill  &\hfill  $\mathsf{script}(s)$\hfill  \\\hline\hline
    $\Rasp$ & $\str{att\_script}$  \\\hline
    $\mi{script\_rp}$ & $\str{script\_rp}$  \\\hline
    $\mi{script\_rp\_redir}$ & $\str{script\_rp\_redir}$  \\\hline
    $\mi{script\_idp}$ &  $\str{script\_idp}$  \\\hline
    $\mi{script\_fwd}$ & $\str{script\_fwd}$ \\\hline
  \end{tabular}
  
  \caption{List of scripts in $\scriptset$ and their respective string
    representations.}
  \label{fig:scripts-in-w}
\end{figure}

We now sketch the processes and the scripts in $\bidsystem$
and $\scriptset$ (see Appendix~\ref{app:model-spresso-auth}
for full details). As mentioned, our modeling closely
follows the description in
Section~\ref{sec:description-spresso} and the
implementation of the system:
\begin{itemize}[leftmargin=1.2em]

\item Browsers (in $\fAP{B}$) are defined as described in Section~\ref{sec:web-browsers}.

\item A relying party (in $\fAP{RP}$) is a web server. RP
  knows four distinct paths: $\mathtt{/}$, where it serves
  the index web page ($\str{script\_rp}$),
  $\mathtt{/startLogin}$, where it only accepts POST
  requests and mainly issues a fresh RP nonce,
  $\mathtt{/redir}$, where it only accepts requests with a
  valid login session token and serves
  $\str{script\_rp\_redir}$ to redirect the browser to the
  IdP, and $\mathtt{/login}$, where it also only accepts
  POST requests with login data obtained during the login
  process by $\str{script\_rp}$ running in the browser.  It
  checks this data and, if the data is considered to be
  valid, it issues a service token. The RP keeps a list of
  such tokens in its state. Intuitively, a client having
  such a token can use the service of the RP.

\item Each IdP (in $\fAP{IDP}$) is a web server. It knows
  three distinct paths:
  $\mathtt{/.well\mhyphen{}known/spresso\mhyphen{}login}$,
  where it serves the login dialog web page
  ($\str{script\_idp}$),
  $\mathtt{/.well\mhyphen{}known/spresso\mhyphen{}info}$,
  where it serves the support document containing its
  public key, and $\mathtt{/sign}$, where it issues a
  (signed) identity assertion. Users can authenticate to
  the IdP with their credentials and IdP tracks the state
  of the users with sessions.  Only authenticated users can
  receive IAs from the IdP.

\item Forwarders (in $\fAP{FWD}$) are web servers that have
  only one state (i.e., they are stateless) and serve only
  the script $\str{script\_fwd}$, except if they become corrupted. 

\item Each DNS server (in $\fAP{DNS}$) contains the assignment of
domain names to IP addresses and answers DNS requests
accordingly.
\end{itemize}

\noindent Besides the browser, RPs, IdPs, and FWDs can
become corrupted: If they receive the message $\corrupt$,
they start collecting all incoming messages in their state
and when triggered send out some message that is derivable
from their state and collected input messages, just like an
attacker process.

%

\section{Privacy of \spresso}
\label{sec:proving-privacy}

In our privacy analysis, we show that an identity provider
in \spresso cannot learn where its users log in. We
formalize this property as an indistinguishability
property: an identity provider (modeled as a web attacker)
cannot distinguish between a user logging in at one relying
party and the same user logging in at a different relying
party.

\subsubsection{Definition of Privacy of \spresso.}
\label{sec:formal-model-spresso-priv}
The web systems considered for the privacy of \spresso are
the web systems $\spressowebsystem$ defined in
Section~\ref{sec:formal-model-spresso} which now contain
one or more web attackers, no network attackers, one honest
DNS server, one honest forwarder, one browser, and two
honest relying parties $r_1$ and $r_2$. All honest parties
may not become corrupted and use the honest DNS server for
address resolving. Identity providers are assumed to be
dishonest, and hence, are subsumed by the web attackers
(which govern all identities). The web attacker subsumes
also potentially dishonest forwarders, DNS servers, relying
parties, and other servers. The honest relying parties are
set up such that they already contain the public signing
keys (used to verify identity assertions) for each domain
registered at the DNS server, modeling that these have been
cached by the relying parties, as discussed in
Section~\ref{sec:login-flow}.

In order to state the privacy property, we replace the
(only) honest browser in the above described web systems by
a slightly extended browser, which we call a
\emph{challenge browser}: This browser may not
become corrupted and is parameterized by a domain $r$ of a
relying party. When it is to assemble an HTTP(S) request for the special domain
$\str{CHALLENGE}$, then instead of putting together and
sending out the request for $\str{CHALLENGE}$ it takes the
domain $r$. However, this is done only for the first
request to $\str{CHALLENGE}$. Further requests to this domain are not
altered (and would fail, as the domain $\str{CHALLENGE}$
is not listed in the honest DNS server).

We denote web systems as described above by
$\spressoprivwebsystem(r)$, where $r$ is the domain of the
relying party given to the challenge browser in this
system.

We can now define privacy of $\spresso$. We note that it is
not important which attacker process in
$\spressoprivwebsystem(\cdot)$ is the distinguished one (in
the sense of Section~\ref{sec:form-defin-priv}).

\begin{definition}\label{def:privacyspresso}
  We say that \spresso is \emph{IdP-private} iff for every
  web system $\spressoprivwebsystem(\cdot)$ and domains
  $r_1$ and $r_2$ of relying parties as described above, we
  have that $\spressoprivwebsystem(r_1)\approx
  \spressoprivwebsystem(r_2)$, i.e.,
  $\spressoprivwebsystem(r_1)$ and
  $\spressoprivwebsystem(r_2)$ are indistinguishable.
\end{definition}
Note that there are many different situations where the
honest browser in $\spressoprivwebsystem(\cdot)$ could be
triggered to send an HTTP(S) request to
$\str{CHALLENGE}$. This could, for example, be triggered by
the user who enters a URL in the location bar of the
browser, a location header (e.g., determined by the
adversary), an (attacker) script telling the browser to
follow a link or create an iframe, etc.

Now, the above definition requires that in every stage of a
run and no matter how and by whom the $\str{CHALLENGE}$
request was triggered, no (malicious) IdP can tell whether
$\str{CHALLENGE}$ was replaced by $r_1$ or $r_2$, i.e.,
whether this resulted in a login request for $r_1$ or
$r_2$. Recall that the $\str{CHALLENGE}$ request is
replaced by the honest browser only once. This is the only
place in a run where the adversary does not know whether
this is a request to $r_1$ or $r_2$. Other requests in a
run, even to both $r_1$ and $r_2$, the adversary can
determine. Still, he should not be able to figure out what
happened in the $\str{CHALLENGE}$ request. Hence, this
definition captures in a strong sense the intuition that a
malicious IdP should not be able to distinguish whether a
user logs in/has logged in at $r_1$ or $r_2$.

\subsubsection{Analyzing Privacy of \spresso.}
\label{sec:analyz-priv-spresso}
The following theorem says that \spresso enjoys the
described privacy definition.

\begin{theorem}
  \spresso is IdP-private.
\end{theorem}

The full proof is provided in Appendix~\ref{app:formal-proof-privacy}. In the
proof, we define an equivalence relation between
configurations of $\spressoprivwebsystem(r_1)$ and
$\spressoprivwebsystem(r_2)$, comprising equivalences
between states and equivalences between events (in the pool
of waiting events). For the states, for each (type of an)
atomic DY process in the web system, we define how their
states are related. For example, the state of the FWD
server must be identical in both configurations.  As
another example, roughly speaking, the attacker's state is
the same up to subterms the attacker cannot decrypt.
Regarding (waiting) events, we distinguish between messages
that result (directly or indirectly) from a
$\str{CHALLENGE}$ request by the browser and other
messages. While the challenged messages may differ in
certain ways, other messages may only differ in parts that
the attacker cannot decrypt.

Given these equivalences, we then show by induction and an
exhaustive case distinction that, starting from equivalent
configurations, every schedule leads to equivalent
configurations. (We note that in
$\spressoprivwebsystem(\cdot)$ a schedule induces a single
run because in $\spressoprivwebsystem(\cdot)$ we do not
have non-deterministic actions that are not determined by a
schedule: honest servers and scripts perform only
deterministic actions.) As an example, we distinguish
between the potential receivers of an event. If, e.g., FWD
is a receiver of a message, given its identical state in
both configurations (as per the equivalence definition) and
the equivalence on the input event, we can immediately show
that the equivalence holds on the output message and state.
For other atomic DY processes, such as browsers and RPs,
this is much harder to show. For example, for browsers, we
need to distinguish between the different scripts that can
potentially run in the browser (including the attacker
script), the origins under which these scripts run, and the
actions they can perform.

For equivalent configurations of
$\spressoprivwebsystem(r_1)$ and
$\spressoprivwebsystem(r_2)$, we show that the attacker's
views are indistinguishable. Given that for all
$\spressoprivwebsystem(r_1)$ and
$\spressoprivwebsystem(r_2)$ every schedule leads to
equivalent configurations, we have that \spresso is
IdP-private.

%
\section{Authentication of \spresso}
\label{sec:prov-auth}

We show that \spresso satisfies two fundamental
authentication properties.

\subsubsection{Formal Model of \spresso for Authentication.}
\label{sec:formal-model-spresso-auth}
For the authentication analysis, we consider web systems as
defined in Section~\ref{sec:formal-model-spresso} which now
contain one network attacker, a finite set of browsers, a
finite set of relying parties, a finite set of identity
providers, and a finite set of forwarders. Browsers,
forwarders, and relying parties can become corrupted by the
network attacker. The network attacker subsumes all web
attackers and also acts as a (dishonest) DNS server to all
other parties.  We denote a web system in this class of web
systems by $\spressoauthwebsystem$.

\subsubsection{Defining Authentication for \spresso.}
\label{sec:defin-auth-spresso}
We state two fundamental authentication properties every
SSO system should satisfy. These properties are adapted
from~\cite{FettKuestersSchmitz-SP-2014}. 

Informally, these properties can be stated as follows:
\textbf{(A)} \emph{The attacker should not be able to use a
  service of an honest RP as an honest user}. In other
words, the attacker should not get hold of (be able to
derive from his current knowledge) a service token issued
by an honest RP for an ID of an honest user (browser), even
if the browser was closed and then later used by a
malicious user, i.e., after a $\str{CLOSECORRUPT}$ (see
Section~\ref{sec:web-browsers}). \textbf{(B)} \emph{The
  attacker should not be able to authenticate an honest
  browser to an honest RP with an ID that is not owned by
  the browser (identity injection)}. For both properties,
we clearly have to require that the forwarder used by the
honest RP is honest as well.

We call a web system $\spressoauthwebsystem$ \emph{secure
  w.r.t.~authentication} if the above conditions are
satisfied in all runs of the system. We refer the reader to
Appendix~\ref{app:form-secur-prop} for the formal definition of
(A) and (B).

\subsubsection{Analyzing Authentication of \spresso.}
\label{sec:analyz-auth-spresso}
We prove the following theorem:

\begin{theorem}\label{thm:authentication}
  Let $\spressoauthwebsystem$ be an \spresso web system as
  defined above. Then $\spressoauthwebsystem$ is secure
  w.r.t.~authentication.
\end{theorem}

In other words, the authentication properties~(A) and~(B)
are fulfilled for every \spresso web system.

For the proof, we first show some general properties of
$\spressoauthwebsystem$. In particular, we show that
encrypted communication over HTTPS between an honest
relying party and an honest IdP cannot be altered by the
(network) attacker, and, based on that, any honest relying
party always retrieves the ``correct'' public signature
verification key from honest IdPs. We then proceed to show
that for a service token to be issued by an honest RP, a
request of a specific form has to be received by the RP.

We then use these properties and the general web system
properties shown in the full version of
\cite{FettKuestersSchmitz-ESORICS-BrowserID-Primary-2015}
to prove properties~(A) and~(B) separately. In both cases,
we assume that the respective property is not satisfied and
lead this to a contradiction. Again, the full proof is
provided in Appendix~\ref{app:proof-spresso}.

%
\section{Further Related Work}\label{sec:relatedwork}

As mentioned in the introduction, many SSO systems have been
developed. However, unlike \spresso, none of them is
privacy-respecting.

Besides the design and implementation of \spresso, the formal analysis
of this system based on an expressive web model is an important part
of our work. The formal treatment of the security of web applications
is a young discipline. Of the few works in this area even less are
based on a general model that incorporates essential mechanisms of the
web. Early works in formal web security analysis (see, e.g.,
\cite{kerschbaum-SP-2007-XSRF-prevention,Jackson-TACAS-2002-Alloy,ArmandoEtAl-FMSE-2008,SantsaiHaekyBeznosov-CS-2012,ChariJutlaRoy-IACR-2011})
are based on very limited models developed specifically for the
application under scrutiny. The first work to consider a general model
of the web, written in the finite-state model checker Alloy, is the
work by Akhawe et al.\cite{AkhawBarthLamMitchellSong-CSF-2010}.
Inspired by this work, Bansal et
al.\cite{BansalBhargavanetal-POST-2013-WebSpi,BansalBhargavanMaffeis-CSF-2012}
built a more expressive model, called WebSpi, in ProVerif
\cite{Blanchet-CSFW-2001}, a tool for symbolic cryptographic protocol
analysis. These models have successfully been applied to web standards
and applications. Recently, Kumar \cite{Kumar-RAID-2014} presented a
high-level Alloy model and applied it to SAML single sign-on. The web
model presented in~\cite{FettKuestersSchmitz-SP-2014},
which we further extend and refine here, is the most comprehensive web
model to date (see also the discussion in
\cite{FettKuestersSchmitz-SP-2014}). In fact, this is the only model
in which we can analyze SPRESSO. For example, other models do not
incorporate a precise handling of windows, documents, or iframes;
cross-document messaging (postMessages) are not included at all.

\section{Conclusion}\label{sec:conclusion}

In this paper, we proposed the first privacy-respecting
(web-based) SSO system, where the IdP cannot track at which
RP a user logs in. Our system, \spresso, is open and
decentralized. Users can log in at any RP with any email
address with SPRESSO support, allowing for seamless and
convenient integration into the usual login process. Being
solely based on standard HTML5 and web features, SPRESSO
can be used across browsers, platforms, and devices.

We formally prove that SPRESSO indeed enjoys strong
authentication and privacy properties. This is important
since, as discussed in the paper, numerous attacks on other
SSO systems have been discovered. These attacks demonstrate
that designing a secure SSO system is non-trivial and
security flaws can easily go undetected when no rigorous
analysis is carried out.

As mentioned in Section~\ref{sec:relatedwork}, there have
been only very few analysis efforts, based on expressive
models of the web infrastructure, on web applications in
general and SSO systems in particular in the literature so
far. Therefore, the analysis carried out in this paper is
also of independent interest.

Our work is the first to analyze privacy properties based
on an expressive web model, in fact the most expressive
model to date.  The general indistinguishability/privacy
definition we propose, which is not tailored to any
specific web application, will be useful beyond the
analysis performed in this paper.

%

\bibliographystyle{abbrv}

\appendix
\section{The Web Model}\label{app:web-model}

In this section, we present the model of the web infrastructure as
proposed in~\cite{FettKuestersSchmitz-SP-2014}
and~\cite{FettKuestersSchmitz-TR-BrowserID-Primary-2015}, along with
the following changes and additions:

\begin{itemize}
\item The set of waiting events is replaced by a (infinite) sequence
  of waiting events. The sequence initially only contains an infinite
  number of trigger events (interleaved by receiver). All new events
  output by processes are added in the front of the sequence.
\item We write events as terms, i.e., $(a{:}f{:}m)$ becomes $\an{a,f,m}$.
\item $E_0$ and $E$ in runs/processing steps are now infinite
  sequences instead of multi-sets.
\item In runs, the index of states and events (and nonces) are now written superscript (instead of subscript).
\item For atomic DY processes, we replace the set of output messages
  by a sequence term (as defined in the equational theory) of the form
  $\an{\an{a,f,m}, \an{a',f',m'}, \dots}$. Each time such a sequence
  is output by any DY process, its elements are prepended to the
  sequence of waiting events.
\item We introduce a global sequence of nonces $(n_1, n_2, \dots)$.
  Whenever any DY process outputs special placeholders $\nu_1, \nu_2,
  \dots$ (in its state or output messages), these placeholders are
  replaced by freshly chosen nonces from the global set of nonces.
\item A similar approach applies to scripts (running inside browsers).
  Instead of receiving and using a fresh set of nonces each time they
  are called by the browser, scripts now get no dedicated set of
  nonces as inputs, but instead may output operators $\mu_1, \mu_2,
  \dots$. After the script run has finished, these are replaced by
  ``fresh'' $\nu$ placeholders by the browser (i.e., $\nu$
  placeholders the browser itself does not use otherwise.)
\item We therefore remove the sets of nonces from DY processes.
\item We remove the function symbol $\unsig{\cdot}$ which extracted
  the signed term from a signature. Instead, we added a new function
  symbol $\checksigThree{\cdot}{\cdot}{\cdot}$ that checks that a
  given term was signed.
\item For an accurate privacy analysis, we introduce the
  \emph{Referer}\footnote{A spelling error in the early HTTP
    standards.} header and associated document property. We also
  introduce the \emph{location} document property.
\item For the script command for following a link ($\tHref$) we add
  the option to avoid sending the \emph{referer} header (as a model
  for the \texttt{rel="noreferrer"} attribute for links in
  HTML5).\footnote{Note that in practice, all major browsers except
    for the Internet Explorer support this property.}
\item DNS responses now not only contain the IP address of the domain
  for which the DNS request was sent, but also the domain itself. This
  is a more realistic model.
\end{itemize}

\subsection{Communication Model}\label{app:communication-model}

We here present details and definitions on the basic concepts of the
communication model.

\subsubsection{Terms, Messages and Events} 
The signature $\Sigma$ for the terms and
messages considered in this work is the union of the following
pairwise disjoint sets of function symbols:
\begin{itemize}
\item constants $C = \addresses\,\cup\, \mathbb{S}\cup
  \{\True,\bot,\notdef\}$ where the three sets are pairwise disjoint,
  $\mathbb{S}$ is interpreted to be the set of ASCII strings
  (including the empty string $\varepsilon$), and $\addresses$ is
  interpreted to be a set of (IP) addresses,
\item function symbols for public keys, (a)symmetric
  en\-cryp\-tion/de\-cryp\-tion, and signatures: $\mathsf{pub}(\cdot)$,
  $\enc{\cdot}{\cdot}$, $\dec{\cdot}{\cdot}$, $\encs{\cdot}{\cdot}$,
  $\decs{\cdot}{\cdot}$, $\sig{\cdot}{\cdot}$,
  $\checksig{\cdot}{\cdot}$, and $\unsig{\cdot}$,
\item $n$-ary sequences $\an{}, \an{\cdot}, \an{\cdot,\cdot},
  \an{\cdot,\cdot,\cdot},$ etc., and
\item projection symbols $\pi_i(\cdot)$ for all $i \in \mathbb{N}$.
\end{itemize}
For strings (elements in $\mathbb{S}$), we use a
specific font. For example, $\cHttpReq$ and $\cHttpResp$
are strings. We denote by $\dns\subseteq \mathbb{S}$ the
set of domains, e.g., $\str{example.com}\in \dns$.  We
denote by $\methods\subseteq \mathbb{S}$ the set of methods
used in HTTP requests, e.g., $\mGet$, $\mPost\in \methods$.

The equational theory associated with the signature
$\Sigma$ is given in Figure~\ref{fig:equational-theory}.

\begin{figure}
\begin{align}
\dec{\enc x{\pub(y)}}{y} &= x\\
\decs{\encs x{y}}{y} &= x\\
\checksigThree{\sig{x}{y}}{x}{\pub(y)} &= \True\\
\pi_i(\an{x_1,\dots,x_n}) &= x_i \text{\;\;if\ } 1 \leq i \leq n \\
\proj{j}{\an{x_1,\dots,x_n}} &= \notdef \text{\;\;if\ } j
\not\in \{1,\dots,n\}
\end{align}
\caption{Equational theory for $\Sigma$.}\label{fig:equational-theory}
\end{figure}

\begin{definition}[Nonces and Terms]\label{def:terms}
  By $X=\{x_0,x_1,\dots\}$ we denote a set of variables and by
  $\nonces$ we denote an infinite set of constants (\emph{nonces})
  such that $\Sigma$, $X$, and $\nonces$ are pairwise disjoint. For
  $N\subseteq\nonces$, we define the set $\gterms_N(X)$ of
  \emph{terms} over $\Sigma\cup N\cup X$ inductively as usual: (1) If
  $t\in N\cup X$, then $t$ is a term. (2) If $f\in \Sigma$ is an
  $n$-ary function symbol in $\Sigma$ for some $n\ge 0$ and
  $t_1,\ldots,t_n$ are terms, then $f(t_1,\ldots,t_n)$ is a term.
\end{definition}

By $\equiv$ we denote the congruence relation on $\terms(X)$ induced
by the theory associated with $\Sigma$. For example, we have that
$\pi_1(\dec{\enc{\an{\str{a},\str{b}}}{\pub(k)}}{k})\equiv \str{a}$.

\begin{definition}[Ground Terms, Messages, Placeholders, Protomessages]\label{def:groundterms-messages-placeholders-protomessages}
  By $\gterms_N=\gterms_N(\emptyset)$, we denote the set of all terms
  over $\Sigma\cup N$ without variables, called \emph{ground terms}.
  The set $\messages$ of messages (over $\nonces$) is defined to be
  the set of ground terms $\gterms_{\nonces}$. 
  
  We define the set $V_{\text{process}} = \{\nu_1, \nu_2, \dots\}$ of
  variables (called placeholders). The set $\messages^\nu :=
  \gterms_{\nonces}(V_{\text{process}})$ is called the set of \emph{protomessages},
  i.e., messages that can contain placeholders.
\end{definition}

\begin{example}
  For example, $k\in \nonces$ and $\pub(k)$ are messages, where $k$
  typically models a private key and $\pub(k)$ the corresponding
  public key. For constants $a$, $b$, $c$ and the nonce $k\in
  \nonces$, the message $\enc{\an{a,b,c}}{\pub(k)}$ is interpreted to
  be the message $\an{a,b,c}$ (the sequence of constants $a$, $b$,
  $c$) encrypted by the public key $\pub(k)$.
\end{example}

\begin{definition}[Normal Form]\gs{Kann man das so schreiben?}
  Let $t$ be a term. The \emph{normal form} of $t$ is acquired by
  reducing the function symbols from left to right as far as possible
  using the equational theory shown in
  Figure~\ref{fig:equational-theory}. For a term $t$, we denote its
  normal form as $t\nf$.
\end{definition}

\begin{definition}[Pattern Matching]
  Let $\mi{pattern} \in \terms(\{*\})$ be a term containing the
  wildcard (variable $*$). We say that a term $t$ \emph{matches}
  $\mi{pattern}$ iff $t$ can be acquired from $\mi{pattern}$ by
  replacing each occurrence of the wildcard with an arbitrary term
  (which may be different for each instance of the wildcard). We write
  $t \sim \mi{pattern}$.

  For a term $t'$ we write $t'|\, \mi{pattern}$ to denote the term
  that is acquired from $t'$ by removing all immediate subterms of
  $t'$ that do not match $\mi{pattern}$.
\end{definition}

\begin{example}
  For example, for a pattern $p = \an{\top,*}$ we have that $\an{\top,42} \sim p$, $\an{\bot,42} \not\sim p$, and \[\an{\an{\bot,\top},\an{\top,23},\an{\str{a},\str{b}},\an{\top,\bot}} |\, p = \an{\an{\top,23},\an{\top,\bot}}\ .\]
\end{example}

\begin{definition}[Variable Replacement]
  Let $N\subseteq \nonces$, $\tau \in \gterms_N(\{x_1,\ldots,x_n\})$,
  and $t_1,\ldots,t_n\in \gterms_N$. By
  $\tau[t_1\!/\!x_1,\ldots,t_n\!/\!x_n]$ we denote the (ground) term obtained
  from $\tau$ by replacing all occurrences of $x_i$ in $\tau$ by
  $t_i$, for all $i\in \{1,\ldots,n\}$.
\end{definition}

\begin{definition}[Events and Protoevents]
  An \emph{event (over $\addresses$ and $\messages$)} is a term of the
  form $\an{a, f, m}$, for $a$, $f\in \addresses$ and $m \in
  \messages$, where $a$ is interpreted to be the receiver address and
  $f$ is the sender address. We denote by $\events$ the set of all
  events. Events over $\addresses$ and $\messages^\nu$ are called
  \emph{protoevents} and are denoted $\events^\nu$. By
  $2^{\events\an{}}$ (or $2^{\events^\nu\an{}}$, respectively) we
  denote the set of all sequences of (proto)events, including the
  empty sequence (e.g., $\an{}$, $\an{\an{a, f, m}, \an{a', f', m'},
    \dots}$, etc.). 
\end{definition}

\subsubsection{Atomic Processes, Systems and Runs} 

An atomic process takes its current state and an
event as input, and then (non-deterministi\-cally) outputs a new state
and a set of events.
\begin{definition}[Generic Atomic Processes and Systems]\label{def:atomic-process-and-process}
  A \emph{(generic) \ap} is a tuple $p = (I^p, Z^p, R^p, s^p_0)$ where
  $I^p \subseteq \addresses$, $Z^p \in \terms$ is a set of states,
  $R^p\subseteq (\events \times Z^p) \times (2^{\events^\nu\an{}}
  \times \terms(V_{\text{process}}))$ (input event and old state map to sequence of
  output events and new state), and $s^p_0\in Z^p$ is the initial
  state of $p$. For any new state $s$ and any sequence of nonces
  $(\eta_1, \eta_2, \dots)$ we demand that $s[\eta_1/\nu_1,
  \eta_2/\nu_2, \dots] \in Z^p$. A \emph{system} $\process$ is a
  (possibly infinite) set of \aps.
\end{definition}

\begin{definition}[Configurations]
  A \emph{configuration of a system $\process$} is a tuple $(S, E, N)$
  where the state of the system $S$ maps every atomic process
  $p\in \process$ to its current state $S(p)\in Z^p$, the sequence of
  waiting events $E$ is an infinite sequence\footnote{Here: Not in the
    sense of terms as defined earlier.} $(e_1, e_2, \dots)$ of events
  waiting to be delivered, and $N$ is an infinite sequence of nonces
  $(n_1, n_2, \dots)$.
\end{definition}

\begin{definition}[Concatenating sequences]
  For a term $a = \an{a_1, \dots, a_i}$ and a sequence $b = (b_1, b_2,
  \dots)$, we define the \emph{concatenation} as $a \cdot b := (a_1,
  \dots, a_i, b_1, b_2, \dots)$.
  
\end{definition}

\begin{definition}[Subtracting from Sequences]
  For a sequence $X$ and a set or sequence $Y$ we define $X \setminus
  Y$ to be the sequence $X$ where for each element in $Y$, a
  non-deterministically chosen occurence of that element in $X$ is
  removed.
\end{definition}

\begin{definition}[Processing Steps]\label{def:processing-step}
  A \emph{processing step of the system $\process$} is of the form
  \[(S,E,N) \xrightarrow[p \rightarrow E_{\text{out}}]{e_\text{in}
    \rightarrow p} (S', E', N')\]
  where
  \begin{enumerate}
  \item $(S,E,N)$ and $(S',E',N')$ are configurations of $\process$,
  \item $e_\text{in} = \an{a, f, m} \in E$ is an event,
  \item $p \in \process$ is a process,
  \item $E_{\text{out}}$ is a sequence (term) of events
  \end{enumerate}
  such that there exists 
  \begin{enumerate}
  \item a sequence (term)
    $E^\nu_{\text{out}} \subseteq 2^{\events^\nu\an{}}$ of protoevents,
  \item a term $s^\nu \in \gterms_{\nonces}(V_{\text{process}})$, 
  \item a sequence $(v_1, v_2, \dots, v_i)$ of all placeholders appearing in $E^\nu_{\text{out}}$ (ordered lexicographically),
  \item a sequence $N^\nu = (\eta_1, \eta_2, \dots, \eta_i)$ of the first $i$ elements in $N$ 
  \end{enumerate}
  with
  \begin{enumerate}
  \item $((e_{\text{in}}, S(p)), (E^\nu_{\text{out}}, s^\nu)) \in R^p$
    and $a \in I^p$,
  \item $E_{\text{out}} = E^\nu_{\text{out}}[m_1/v_1, \dots, m_i/v_i]$
  \item $S'(p) = s^\nu[m_1/v_1, \dots, m_i/v_i]$ and $S'(p') = S(p')$ for all $p' \neq p$
  \item $E' = E_{\text{out}} \cdot (E \setminus \{e_{\text{in}}\})$ 
  \item $N' = N \setminus N^\nu$ 
  \end{enumerate}
  We may omit the superscript and/or subscript of the arrow.
\end{definition} 
Intuitively, for a processing step, we select one of the processes in
$\process$, and call it with one of the events in the list of waiting
events $E$. In its output (new state and output events), we replace
any occurences of placeholders $\nu_x$ by ``fresh'' nonces from $N$
(which we then remove from $N$). The output events are then prepended
to the list of waiting events, and the state of the process is
reflected in the new configuration.

\begin{definition}[Runs]
  Let $\process$ be a system, $E^0$ be sequence of events, and $N^0$ be
  a sequence of nonces. A \emph{run $\rho$ of a system $\process$
    initiated by $E^0$ with nonces $N^0$} is a finite sequence of
  configurations $((S^0, E^0, N^0),\dots,(S^n, E^n, N^n))$ or an infinite sequence
  of configurations $((S^0, E^0, N^0),\dots)$ such that $S^0(p) = s_0^p$ for
  all $p \in \process$ and $(S^i, E^i, N^i) \xrightarrow{} (S^{i+1},
  E^{i+1}, N^{i+1})$ for all $0 \leq i < n$ (finite run) or for all $i \geq 0$  
  (infinite run).

  We denote the state $S^n(p)$ of a process $p$ at the end of a run $\rho$ by $\rho(p)$.
\end{definition}

Usually, we will initiate runs with a set $E^0$ that contains infinite
trigger events of the form $\an{a, a, \str{TRIGGER}}$ for each $a \in
\addresses$, interleaved by address.

\subsubsection{Atomic Dolev-Yao Processes}  We next define
atomic Dolev-Yao processes, for which we require that the
messages and states that they output can be computed (more
formally, derived) from the current input event and
state. For this purpose, we first define what it means to
derive a message from given messages.

\begin{definition}[Deriving Terms]
  Let $M$ be a set of ground terms. We say that \emph{a
    term $m$ can be derived from $M$ with placeholders $V$} if there
  exist $n\ge 0$, $m_1,\ldots,m_n\in M$, and $\tau\in
  \gterms_{\emptyset}(\{x_1,\ldots,x_n\} \cup V)$ such that $m\equiv
  \tau[m_1/x_1,\ldots,m_n/x_n]$. We denote by $d_V(M)$ the set of all
  messages that can be derived from $M$ with variables $V$.
\end{definition}
For example, $a\in d_{\{\}}(\{\enc{\an{a,b,c}}{\pub(k)}, k\})$.

\begin{definition}[Atomic Dolev-Yao Process] \label{def:adyp} An \emph{atomic Dolev-Yao process
    (or simply, a DY process)} is a tuple $p = (I^p, Z^p,$ $R^p,
  s^p_0)$ such that $(I^p, Z^p, R^p, s^p_0)$ is an atomic process and
  (1) $Z^p \subseteq \gterms_{\nonces}$ (and hence, $s^p_0\in
  \gterms_{\nonces}$), and (2) for all events $e \in \events$,
  sequences of protoevents $E$, $s\in \gterms_{\nonces}$, $s'\in
  \gterms_{\nonces}(V_{\text{process}})$, with $((e, s), (E, s')) \in R^p$ it holds
  true that $E$, $s' \in d_{V_{\text{process}}}(\{e,s\})$.
\end{definition}

\begin{definition}[Atomic Attacker Process]\label{def:atomicattacker}
  An \emph{(atomic) attacker process for a set of sender addresses
    $A\subseteq \addresses$} is an atomic DY process $p = (I, Z, R,
  s_0)$ such that for all events $e$, and $s\in \gterms_{\nonces}$ we
  have that $((e, s), (E,s')) \in R$ iff $s'=\an{e, E, s}$ and
  $E=\an{\an{a_1, f_1, m_1}, \dots, \an{a_n, f_n, m_n}}$ with $n \in
  \mathbb{N}$, $a_1,\dots,a_n\in \addresses$, $f_0,\dots,f_n\in A$,
  $m_1,\dots,m_n\in d_{V_{\text{process}}}(\{e,s\})$.
\end{definition}

\subsection{Scripting Processes}
We define scripting processes, which model client-side scripting
technologies, such as JavaScript. Scripting processes are defined
similarly to DY processes.
\begin{definition}[Placeholders for Scripting Processes]\label{def:placeholder-sp}
  By $V_{\text{script}} = \{\lambda_1, \dots\}$ we denote an infinite set of variables
  used in scripting processes.
\end{definition}

\begin{definition}[Scripting Processes]\label{def:sp} A \emph{scripting process} (or
  simply, a \emph{script}) is a relation $R\subseteq \terms \times
  \terms(V_{\text{script}})$ such that for all $s \in \terms$, $s' \in \terms(V_{\text{script}})$ with
  $(s, s') \in R$ it follows that $s'\in d_{V_{\text{script}}}(s)$.
\end{definition}
A script is called by the browser which provides it with state
information (such as the script's last state and limited information
about the browser's state) $s$. The script then outputs a term $s'$,
which represents the new internal state and some command which is
interpreted by the browser. The term $s'$ may contain variables
$\lambda_1, \dots$ which the browser will replace by (otherwise
unused) placeholders $\nu_1,\dots$ which will be replaced by nonces
once the browser DY process finishes (effectively providing the script
with a way to get ``fresh'' nonces).

Similarly to an attacker process, we define the
\emph{attacker script} $\Rasp$: 
\begin{definition}[Attacker Script]
  The attacker script $\Rasp$ outputs everything that is derivable
  from the input, i.e., $\Rasp=\{(s, s')\mid s\in \terms, s'\in
  d_{V_{\text{script}}}(s)\}$.
\end{definition}

\subsection{Web System}\label{app:websystem}

The web infrastructure and web applications are formalized by what is
called a web system. A web system contains, among others, a (possibly
infinite) set of DY processes, modeling web browsers, web servers, DNS
servers, and attackers (which may corrupt other entities, such as
browsers).

\begin{definition}\label{def:websystem}
  A \emph{web system $\completewebsystem=(\websystem,
    \scriptset,\mathsf{script}, E^0)$} is a tuple with its
  components defined as follows:

  The first component, $\websystem$, denotes a system
  (a set of DY processes) and is partitioned into the
  sets $\mathsf{Hon}$, $\mathsf{Web}$, and $\mathsf{Net}$
  of honest, web attacker, and network attacker processes,
  respectively.  

  Every $p\in \mathsf{Web} \cup \mathsf{Net}$ is an
  attacker process for some set of sender addresses
  $A\subseteq \addresses$. For a web attacker $p\in
  \mathsf{Web}$, we require its set of addresses $I^p$ to
  be disjoint from the set of addresses of all other web
  attackers and honest processes, i.e., $I^p\cap I^{p'} =
  \emptyset$ for all $p' \in \mathsf{Hon} \cup
  \mathsf{Web}$. Hence, a web attacker cannot listen to
  traffic intended for other processes. Also, we require
  that $A=I^p$, i.e., a web attacker can only use sender
  addresses it owns. Conversely, a network attacker may
  listen to all addresses (i.e., no restrictions on $I^p$)
  and may spoof all addresses (i.e., the set $A$ may be
  $\addresses$).

  Every $p \in \mathsf{Hon}$ is a DY process which
  models either a \emph{web server}, a \emph{web browser},
  or a \emph{DNS server},\gs{to check:} as further described in the
  following subsections. Just as for web attackers, we
  require that $p$ does not spoof sender addresses and that
  its set of addresses $I^p$ is disjoint from those of
  other honest processes and the web attackers. 

  The second component, $\scriptset$, is a finite set of
  scripts such that $\Rasp\in \scriptset$. The third
  component, $\mathsf{script}$, is an injective mapping
  from $\scriptset$ to $\mathbb{S}$, i.e., by
  $\mathsf{script}$ every $s\in \scriptset$ is assigned its
  string representation $\mathsf{script}(s)$. 

  Finally, $E^0$ is an \gs{changed:} (infinite) sequence of events, containing an
  infinite number of events of the form $\an{a,a,\trigger}$
  for every $a \in \bigcup_{p\in \websystem} I^p$.

  A \emph{run} of $\completewebsystem$ is a run of
  $\websystem$ initiated by $E^0$.
\end{definition}

\section{Message and Data
  Formats}\label{app:message-data-formats}

We now provide some more details about data and message
formats that are needed for the formal treatment of the web
model and the analysis of BrowserID presented in the rest
of the appendix.

\subsection{Notations}\label{app:notation}

\begin{definition}[Sequence Notations]
  For a sequence $t = \an{t_1,\dots,t_n}$ and a set $s$ we
  use $t \subsetPairing s$ to say that $t_1,\dots,t_n \in
  s$.  We define $\left. x \inPairing t\right. \iff \exists
  i: \left. t_i = x\right.$.
  We write $t \plusPairing y$ to denote the sequence
  $\an{t_1,\dots,t_n,y}$.
  For a finite set $M$ with $M = \{m_1, \dots,m_n\}$ we use
  $\an{M}$ to denote the term of the form
  $\an{m_1,\dots,m_n}$. (The order of the elements does not
  matter; one is chosen arbitrarily.) 
\end{definition}

\begin{definition}\label{def:dictionaries}
  A \emph{dictionary over $X$ and $Y$} is a term of the
  form \[\an{\an{k_1, v_1}, \dots, \an{k_n,v_n}}\] where
  $k_1, \dots,k_n \in X$, $v_1,\dots,v_n \in Y$, and the
  keys $k_1, \dots,k_n$ are unique, i.e., $\forall i\neq j:
  k_i \neq k_j$. We call every term $\an{k_i,v_i}$, $i\in
  \{1,\ldots,n\}$, an \emph{element} of the dictionary with
  key $k_i$ and value $v_i$.  We often write $\left[k_1:
    v_1, \dots, k_i:v_i,\dots,k_n:v_n\right]$ instead of
  $\an{\an{k_1, v_1}, \dots, \an{k_n,v_n}}$. We denote the
  set of all dictionaries over $X$ and $Y$ by $\left[X
    \times Y\right]$.
\end{definition}
We note that the empty dictionary is equivalent to the
empty sequence, i.e.,  $[] = \an{}$.  Figure
\ref{fig:dictionaries} shows the short notation for
dictionary operations that will be used when describing the
browser atomic process. For a dictionary $z = \left[k_1:
  v_1, k_2: v_2,\dots, k_n:v_n\right]$ we write $k \in z$ to
say that there exists $i$ such that $k=k_i$. We write
$z[k_j] := v_j$ to extract elements. If $k \not\in z$, we
set $z[k] := \an{}$.

\begin{figure}[htb!]\centering
  \begin{align}
    \left[k_1: v_1, \dots, k_i:v_i,\dots,k_n:v_n\right][k_i] = v_i%
  \end{align}\vspace{-2.5em}
  \begin{align}
    \nonumber \left[k_1: v_1, \dots, k_{i-1}:v_{i-1},k_i: v_i, k_{i+1}:v_{i+1}\dots,k_n: v_n\right]-k_i =\\
         \left[k_1: v_1, \dots, k_{i-1}:v_{i-1},k_{i+1}:v_{i+1}\dots,k_n: v_n\right]
  \end{align}
  \caption{Dictionary operators with $1\le i\le n$.}\label{fig:dictionaries}
\end{figure}

Given a term $t = \an{t_1,\dots,t_n}$, we can refer to any
subterm using a sequence of integers. The subterm is
determined by repeated application of the projection
$\pi_i$ for the integers $i$ in the sequence. We call such
a sequence a \emph{pointer}:

\begin{definition}\label{def:pointer}
  A \emph{pointer} is a sequence of non-negative
  integers. We write $\tau.\ptr{p}$ for the application of
  the pointer $\ptr{p}$ to the term $\tau$. This operator
  is applied from left to right. For pointers consisting of
  a single integer, we may omit the sequence braces for
  brevity.
\end{definition}

\begin{example}
  For the term $\tau = \an{a,b,\an{c,d,\an{e,f}}}$ and the
  pointer $\ptr{p} = \an{3,1}$, the subterm of $\tau$ at
  the position $\ptr{p}$ is $c =
  \proj{1}{\proj{3}{\tau}}$. Also, $\tau.3.\an{3,1} =
  \tau.3.\ptr{p} = \tau.3.3.1 = e$.
\end{example}

To improve readability, we try to avoid writing, e.g.,
$\compn{o}{2}$ or $\proj{2}{o}$ in this document. Instead,
we will use the names of the components of a sequence that
is of a defined form as pointers that point to the
corresponding subterms. E.g., if an \emph{Origin} term is
defined as $\an{\mi{host}, \mi{protocol}}$ and $o$ is an
Origin term, then we can write $\comp{o}{protocol}$ instead
of $\proj{2}{o}$ or $\compn{o}{2}$. See also
Example~\ref{ex:url-pointers}.

\subsection{URLs}\label{app:urls}

\begin{definition}\label{def:url}
  A \emph{URL} is a term of the form $\an{\tUrl, \mi{protocol},
    \mi{host}, \mi{path}, \mi{parameters}}$ with $\mi{protocol}$
  $\in \{\http, \https\}$ (for \textbf{p}lain (HTTP) and
  \textbf{s}ecure (HTTPS)), $\mi{host} \in \dns$,
  $\mi{path} \in \mathbb{S}$ and $\mi{parameters} \in
  \dict{\mathbb{S}}{\terms}$. The set of all valid URLs
  is $\urls$.
\end{definition}

\begin{example} \label{ex:url-pointers}
  For the URL $u = \an{\tUrl, a, b, c, d}$, $\comp{u}{protocol} =
  a$. If, in the algorithm described later, we say $\comp{u}{path} :=
  e$ then $u = \an{\tUrl, a, b, c, e}$ afterwards. 
\end{example}

\subsection{Origins}\label{app:origins}
\begin{definition} An \emph{origin} is a term of the form
  $\an{\mi{host}, \mi{protocol}}$ with $\mi{host} \in
  \dns$ and $\mi{protocol} \in \{\http, \https\}$. We write
  $\origins$ for the set of all origins.  
\end{definition}

\begin{example}
  For example, $\an{\str{FOO}, \https}$ is the HTTPS origin
  for the domain $\str{FOO}$, while $\an{\str{BAR}, \http}$
  is the HTTP origin for the domain $\str{BAR}$.
\end{example}
\subsection{Cookies}\label{app:cookies}

\begin{definition} A \emph{cookie} is a term of the form
  $\an{\mi{name}, \mi{content}}$ where $\mi{name} \in
  \terms$, and $\mi{content}$ is a term of the form
  $\an{\mi{value}, \mi{secure}, \mi{session},
    \mi{httpOnly}}$ where $\mi{value} \in \terms$,
  $\mi{secure}$, $\mi{session}$, $\mi{httpOnly} \in
  \{\True, \bot\}$. We write $\cookies$ for the set of all
  cookies and $\cookies^\nu$ for the set of all cookies
  where names and values are defined over $\terms(V)$.
\end{definition}

If the $\mi{secure}$ attribute of a cookie is set, the
browser will not transfer this cookie over unencrypted HTTP
connections. If the $\mi{session}$ flag is set, this cookie
will be deleted as soon as the browser is closed. The
$\mi{httpOnly}$ attribute controls whether JavaScript has
access to this cookie.

Note that cookies of the form described here are only
contained in HTTP(S) requests. In responses, only the
components $\mi{name}$ and $\mi{value}$ are transferred as
a pairing of the form $\an{\mi{name}, \mi{value}}$.

\subsection{HTTP Messages}\label{app:http-messages-full}
\begin{definition}
  An \emph{HTTP request} is a term of the form shown in
  (\ref{eq:default-http-request}). An \emph{HTTP response}
  is a term of the form shown in
  (\ref{eq:default-http-response}).
  \begin{align}
    \label{eq:default-http-request}
    & \hreq{ nonce=\mi{nonce}, method=\mi{method},
      xhost=\mi{host}, xpath=\mi{path},
      parameters=\mi{parameters}, headers=\mi{headers},
      xbody=\mi{body}
    } \\
    \label{eq:default-http-response}
    & \hresp{ nonce=\mi{nonce}, status=\mi{status},
      headers=\mi{headers}, xbody=\mi{body} }
  \end{align}
  The components are defined as follows:
  \begin{itemize}
  \item $\mi{nonce} \in \nonces$ serves to map each
    response to the corresponding request 
  \item $\mi{method} \in \methods$ is one of the HTTP
    methods.
  \item $\mi{host} \in \dns$ is the host name in the HOST
    header of HTTP/1.1.
  \item $\mi{path} \in \mathbb{S}$ is a string indicating
    the requested resource at the server side
  \item $\mi{status} \in \mathbb{S}$ is the HTTP status
    code (i.e., a number between 100 and 505, as defined by
    the HTTP standard)
  \item $\mi{parameters} \in
    \dict{\mathbb{S}}{\terms}$ contains URL parameters
  \item $\mi{headers} \in \dict{\mathbb{S}}{\terms}$,
    containing request/response headers. The dictionary
    elements are terms of one of the following forms: 
    \begin{itemize}
    \item $\an{\str{Origin}, o}$ where $o$ is an origin
    \item $\an{\str{Set{\mhyphen}Cookie}, c}$ where $c$ is
      a sequence of cookies
    \item $\an{\str{Cookie}, c}$ where $c \in
      \dict{\mathbb{S}}{\terms}$ (note that in this header,
      only names and values of cookies are transferred)
    \item $\an{\str{Location}, l}$ where $l \in \urls$
    \item $\an{\str{Referer}, r}$ where $r \in \urls$
    \item $\an{\str{Strict{\mhyphen}Transport{\mhyphen}Security},\True}$
    \end{itemize}
  \item $\mi{body} \in \terms$ in requests and responses. 
  \end{itemize}
  We write $\httprequests$/$\httpresponses$ for the set of
  all HTTP requests or responses, respectively.
\end{definition}

\begin{example}[HTTP Request and Response]
  \begin{align}
    \label{eq:ex-request}
    \nonumber \mi{r} := & \langle
                   \cHttpReq,
                   n_1,
                   \mPost,
                   \str{example.com},
                   \str{/show},
                   \an{\an{\str{index,1}}},\\ & \quad
                   [\str{Origin}: \an{\str{example.com, \https}}],
                   \an{\str{foo}, \str{bar}}
                \rangle \\
    \label{eq:ex-response} \mi{s} := & \hresp{ nonce=n_1,
      status=200,
      headers=\an{\an{\str{Set{\mhyphen}Cookie},\an{\an{\str{SID},\an{n_2,\bot,\bot,\True}}}}},
      xbody=\an{\str{somescript},x}}
  \end{align}
  \noindent
  An HTTP $\mGet$ request for the URL
  \url{http://example.com/show?index=1} is shown in
  (\ref{eq:ex-request}), with an Origin header and a body
  that contains $\an{\str{foo},\str{bar}}$. A possible
  response is shown in (\ref{eq:ex-response}), which
  contains an httpOnly cookie with name $\str{SID}$ and
  value $n_2$ as well as the string representation
  $\str{somescript}$ of the scripting process
  $\mathsf{script}^{-1}(\str{somescript})$ (which should be
  an element of $\scriptset$) and its initial state
  $x$.
\end{example}

\subsubsection{Encrypted HTTP
  Messages.} \label{app:http-messages-encrypted-full}
For HTTPS, requests are encrypted using the public key of
the server.  Such a request contains an (ephemeral)
symmetric key chosen by the client that issued the
request. The server is supported to encrypt the response
using the symmetric key.

\begin{definition} An \emph{encrypted HTTP request} is of
  the form $\enc{\an{m, k'}}{k}$, where $k$, $k' \in
  \nonces$ and $m \in \httprequests$. The corresponding
  \emph{encrypted HTTP response} would be of the form
  $\encs{m'}{k'}$, where $m' \in \httpresponses$. We call
  the sets of all encrypted HTTP requests and responses
  $\httpsrequests$ or $\httpsresponses$, respectively.
\end{definition}

\begin{example}
  \begin{align}
    \label{eq:ex-enc-request} \ehreqWithVariable{r}{k'}{\pub(k_\text{example.com})} \\
    \label{eq:ex-enc-response} \ehrespWithVariable{s}{k'}
  \end{align} The term (\ref{eq:ex-enc-request}) shows an
  encrypted request (with $r$ as in
  (\ref{eq:ex-request})). It is encrypted using the public
  key $\pub(k_\text{example.com})$.  The term
  (\ref{eq:ex-enc-response}) is a response (with $s$ as in
  (\ref{eq:ex-response})). It is encrypted symmetrically
  using the (symmetric) key $k'$ that was sent in the
  request (\ref{eq:ex-enc-request}).
\end{example}

\subsection{DNS Messages}\label{app:dns-messages}
\begin{definition} A \emph{DNS request} is a term of the form
$\an{\cDNSresolve, \mi{domain}, \mi{n}}$ where $\mi{domain}$ $\in
\dns$, $\mi{n} \in \nonces$. We call the set of all DNS requests
$\dnsrequests$.
\end{definition}

\begin{definition} A \emph{DNS response} is a term of the form
$\an{\cDNSresolved, \mi{domain}, \mi{result}, \mi{n}}$ with $\mi{domain}$ $\in
\dns$, $\mi{result} \in
\addresses$, $\mi{n} \in \nonces$. We call the set of all DNS
responses $\dnsresponses$.
\end{definition}

DNS servers are supposed to include the nonce they received
in a DNS request in the DNS response that they send back so
that the party which issued the request can match it with
the request.

\subsection{DNS Servers}\label{app:DNSservers}

Here, we consider a flat DNS model in which DNS queries are
answered directly by one DNS server and always with the
same address for a domain. A full (hierarchical) DNS system
with recursive DNS resolution, DNS caches, etc.~could also
be modeled to cover certain attacks on the DNS system
itself.

\begin{definition}
  A \emph{DNS server} $d$ (in a flat DNS model) is modeled
  in a straightforward way as an atomic DY process
  $(I^d, \{s^d_0\}, R^d, s^d_0)$. It has a finite set of
  addresses $I^d$ and its initial (and only) state $s^d_0$
  encodes a mapping from domain names to addresses of the
  form
$$s^d_0=\langle\an{\str{domain}_1,a_1},\an{\str{domain}_2, a_2}, \ldots\rangle \ .$$ DNS
queries are answered according to this table (otherwise
ignored).
\end{definition}

\section{Detailed Description of the Browser Model}
\label{app:deta-descr-brows}
Following the informal description of the browser model in
\gs{beim zusammenkleben aufpassen:}Section~\ref{sec:web-browsers}, we now present a formal
model. We start by introducing some notation and
terminology. 

\subsection{Notation and Terminology (Web Browser State)}

Before we can define the state of a web browser, we first
have to define windows and documents. 

\begin{sloppypar}
  \begin{definition} A \emph{window} is a term of the form
    $w = \an{\mi{nonce}, \mi{documents}, \mi{opener}}$ with
    $\mi{nonce} \in \nonces$,
    $\mi{documents} \subsetPairing \documents$ (defined
    below), $\mi{opener} \in \nonces \cup \{\bot\}$ where
    $\comp{d}{active} = \True$ for exactly one
    $d \inPairing \mi{documents}$ if $\mi{documents}$ is
    not empty (we then call $d$ the \emph{active document
      of $w$}). We write $\windows$ for the set of all
    windows. We write $\comp{w}{activedocument}$ to denote
    the active document inside window $w$ if it exists and
    $\an{}$ else.
  \end{definition}
\end{sloppypar}

We will refer to the window nonce as \emph{(window)
  reference}.

The documents contained in a window term to the left of the
active document are the previously viewed documents
(available to the user via the ``back'' button) and the
documents in the window term to the right of the currently
active document are documents available via the ``forward''
button.

A window $a$ may have opened a top-level window $b$ (i.e.,
a window term which is not a subterm of a document
term). In this case, the \emph{opener} part of the term $b$
is the nonce of $a$, i.e., $\comp{b}{opener} =
\comp{a}{nonce}$.

\begin{sloppypar}
  \begin{definition} A \emph{document} $d$ is a term of the
    form
    \begin{align*}
      \an{\mi{nonce}, \mi{location}, \mi{referrer}, \mi{script},
      \mi{scriptstate},\mi{scriptinputs}, \mi{subwindows},
      \mi{active}}  
    \end{align*}
    where $\mi{nonce} \in \nonces$,
    $\mi{location} \in \urls$,
    $\mi{referrer} \in \urls \cup \{\bot\}$,
    $\mi{script} \in \terms$,
    $\mi{scriptstate} \in \terms$,
    $\mi{scriptinputs} \in \terms$,
    $\mi{subwindows} \subsetPairing \windows$,
    $\mi{active} \in \{\True, \bot\}$. A \emph{limited
      document} is a term of the form
    $\an{\mi{nonce}, \mi{subwindows}}$ with $\mi{nonce}$,
    $\mi{subwindows}$ as above. A window
    $w \inPairing \mi{subwindows}$ is called a
    \emph{subwindow} (of $d$). We write $\documents$ for
    the set of all documents. For a document term $d$ we
    write $d.\str{origin}$ to denote the origin of the
    document, i.e., the term
    $\an{d.\str{location}.\str{host},
      d.\str{location}.\str{protocol}} \in \origins$.
  \end{definition}
\end{sloppypar}

We will refer to the document nonce as \emph{(document)
  reference}.

We can now define the set of states of web browsers. Note
that we use the dictionary notation that we introduced in
Definition~\ref{def:dictionaries}.

\begin{definition} The
  \emph{set of states $Z^p$ of a web browser atomic process}
  $p$ consists of the terms of the form
  \begin{align*} \langle\mi{windows}, \mi{ids},
    \mi{secrets}, \mi{cookies}, \mi{localStorage},
    \mi{sessionStorage}, \mi{keyMapping}, \\\mi{sts},
    \mi{DNSaddress}, \mi{pendingDNS},
    \mi{pendingRequests}, \mi{isCorrupted}\rangle
  \end{align*} where
  \begin{itemize}
  \item $\mi{windows} \subsetPairing \windows$,
  \item $\mi{ids} \subsetPairing \terms$,
  \item $\mi{secrets} \in \dict{\origins}{\nonces}$,
  \item $\mi{cookies}$ is a dictionary over $\dns$ and
    dictionaries of $\cookies$,
  \item $\mi{localStorage} \in \dict{\origins}{\terms}$,
  \item $\mi{sessionStorage} \in \dict{\mi{OR}}{\terms}$ for $\mi{OR} := \left\{\an{o,r}
    \middle|\, o \in \origins,\, r \in \nonces\right\}$,
  \item $\mi{keyMapping} \in \dict{\dns}{\terms}$,
  \item $\mi{sts} \subsetPairing \dns$,
  \item $\mi{DNSaddress} \in \addresses$,
  \item $\mi{pendingDNS} \in \dict{\nonces}{\terms}$,
  \item $\mi{pendingRequests} \in$ $\terms$,
  \item and $\mi{isCorrupted} \in \{\bot, \fullcorrupt,$ $
    \closecorrupt\}$.
  \end{itemize} 
\end{definition}

\begin{definition} For two window terms $w$ and $w'$ we
  write $w \windowChildOf w'$ if \\
  \[w \inPairing \comp{\comp{w'}{activedocument}}{subwindows}\text{\ .}\]
We write
  $\windowChildOfX$ for the transitive closure.
\end{definition}

In the following description of the web browser relation
$R^p$ we will use the helper functions
$\mathsf{Subwindows}$, $\mathsf{Docs}$, $\mathsf{Clean}$,
$\mathsf{CookieMerge}$ and $\mathsf{AddCookie}$. 

Given a browser state $s$, $\mathsf{Subwindows}(s)$ denotes
the set of all pointers\footnote{Recall the definition of a
  pointer in Definition~\ref{def:pointer}.} to windows in
the window list $\comp{s}{windows}$, their active
documents, and (recursively) the subwindows of these
documents. We exclude subwindows of inactive documents and
their subwindows. With $\mathsf{Docs}(s)$ we denote the set
of pointers to all active documents in the set of windows
referenced by $\mathsf{Subwindows}(s)$.
\begin{definition} 
  For a browser state $s$ we denote by
  $\mathsf{Subwindows}(s)$ the minimal set of
  pointers that satisfies the
  following conditions: (1) For all windows $w \inPairing
  \comp{s}{windows}$ there is a $\ptr{p} \in
  \mathsf{Subwindows}(s)$ such that $\compn{s}{\ptr{p}} =
  w$. (2) For all $\ptr{p} \in \mathsf{Subwindows}(s)$, the
  active document $d$ of the window $\compn{s}{\ptr{p}}$
  and every subwindow $w$ of $d$ there is a pointer
  $\ptr{p'} \in \mathsf{Subwindows}(s)$ such that
  $\compn{s}{\ptr{p'}} = w$.

  Given a browser state $s$, the set $\mathsf{Docs}(s)$ of
  pointers to active documents is the minimal set such that
  for every $\ptr{p} \in \mathsf{Subwindows}(s)$, there is
  a pointer $\ptr{p'} \in \mathsf{Docs}(s)$ with
  $\compn{s}{\ptr{p'}} =
  \comp{\compn{s}{\ptr{p}}}{activedocument}$.
\end{definition}

By $\mathsf{Subwindows}^+(s)$ and $\mathsf{Docs}^+(s)$ we
denote the respective sets that also include the inactive
documents and their subwindows.

The function $\mathsf{Clean}$ will be used to determine
which information about windows and documents the script
running in the document $d$ has access to.
\begin{definition} Let $s$ be a browser state and $d$ a
  document.  By $\mathsf{Clean}(s, d)$ we denote the term
  that equals $\comp{s}{windows}$ but with all inactive
  documents removed (including their subwindows etc.) and
  all subterms that represent non-same-origin documents
  w.r.t.~$d$ replaced by a limited document $d'$ with the
  same nonce and the same subwindow list. Note that
  non-same-origin documents on all levels are replaced by
  their corresponding limited document.
\end{definition}

The function $\mathsf{CookieMerge}$ merges two sequences of
cookies together: When used in the browser,
$\mi{oldcookies}$ is the sequence of existing cookies for
some origin, $\mi{newcookies}$ is a sequence of new cookies
that was output by some script. The sequences are merged
into a set of cookies using an algorithm that is based on
the \emph{Storage Mechanism} algorithm described in
RFC6265.
\begin{definition} \label{def:cookiemerge} For a sequence
  of cookies (with pairwise different names)
  $\mi{oldcookies}$ and a sequence of cookies
  $\mi{newcookies}$, the set
  $\mathsf{CookieMerge}(\mi{oldcookies}, \mi{newcookies})$
  is defined by the following algorithm: From
  $\mi{newcookies}$ remove all cookies $c$ that have
  $c.\str{content}.\str{httpOnly} \equiv \True$. For any
  $c$, $c' \inPairing \mi{newcookies}$, $\comp{c}{name}
  \equiv \comp{c'}{name}$, remove the cookie that appears
  left of the other in $\mi{newcookies}$. Let $m$ be the
  set of cookies that have a name that either appears in
  $\mi{oldcookies}$ or in $\mi{newcookies}$, but not in
  both. For all pairs of cookies $(c_\text{old},
  c_\text{new})$ with $c_\text{old} \inPairing
  \mi{oldcookies}$, $c_\text{new} \inPairing
  \mi{newcookies}$, $\comp{c_\text{old}}{name} \equiv
  \comp{c_\text{new}}{name}$, add $c_\text{new}$ to $m$ if
  $\comp{\comp{c_\text{old}}{content}}{httpOnly} \equiv
  \bot$ and add $c_\text{old}$ to $m$ otherwise. The result
  of $\mathsf{CookieMerge}(\mi{oldcookies},
  \mi{newcookies})$ is $m$.
\end{definition}

The function $\mathsf{AddCookie}$ adds a cookie $c$
received in an HTTP response to the sequence of cookies
contained in the sequence $\mi{oldcookies}$. It is again
based on the algorithm described in RFC6265 but simplified
for the use in the browser model.
\begin{definition} \label{def:addcookie} For a sequence of cookies (with pairwise different
  names) $\mi{oldcookies}$ and a cookie $c$, the sequence
  $\mathsf{AddCookie}(\mi{oldcookies}, c)$ is defined by the
  following algorithm: Let $m := \mi{oldcookies}$. Remove
  any $c'$ from $m$ that has $\comp{c}{name} \equiv
  \comp{c'}{name}$. Append $c$ to $m$ and return $m$.
\end{definition}

The function $\mathsf{NavigableWindows}$ returns a set of
windows that a document is allowed to navigate. We closely
follow \cite{html5}, Section~5.1.4 for this definition.
\begin{definition} The set $\mathsf{NavigableWindows}(\ptr{w}, s')$
  is the set $\ptr{W} \subseteq
  \mathsf{Subwindows}(s')$ of pointers to windows that the
  active document in $\ptr{w}$ is allowed to navigate. The
  set $\ptr{W}$ is defined to be the minimal set such that
  for every $\ptr{w'}
  \in \mathsf{Subwindows}(s')$ the following is true: %
\begin{itemize}
\item If
  $\comp{\comp{\compn{s'}{\ptr{w}'}}{activedocument}}{origin}
  \equiv
  \comp{\comp{\compn{s'}{\ptr{w}}}{activedocument}}{origin}$
  (i.e., the active documents in $\ptr{w}$ and $\ptr{w'}$ are
  same-origin), then $\ptr{w'} \in \ptr{W}$, and
\item If ${\compn{s'}{\ptr{w}} \childof
    \compn{s'}{\ptr{w'}}}$ $\wedge$ $\nexists\, \ptr{w}''
  \in \mathsf{Subwindows}(s')$ with $\compn{s'}{\ptr{w}'}
  \childof \compn{s'}{\ptr{w}''}$ ($\ptr{w'}$ is a
  top-level window and $\ptr{w}$ is an ancestor window of
  $\ptr{w'}$), then $\ptr{w'} \in \ptr{W}$, and
\item If $\exists\, \ptr{p} \in \mathsf{Subwindows}(s')$
  such that $\compn{s'}{\ptr{w}'} \windowChildOfX
  \compn{s'}{\ptr{p}}$ \\$\wedge$
  $\comp{\comp{\compn{s'}{\ptr{p}}}{activedocument}}{origin}
  =
  \comp{\comp{\compn{s'}{\ptr{w}}}{activedocument}}{origin}$
  ($\ptr{w'}$ is not a top-level window but there is an
  ancestor window $\ptr{p}$ of $\ptr{w'}$ with an active
  document that has the same origin as the active document
  in $\ptr{w}$), then $\ptr{w'} \in \ptr{W}$, and
\item If $\exists\, \ptr{p} \in \mathsf{Subwindows}(s')$ such
  that $\comp{\compn{s'}{\ptr{w'}}}{opener} =
  \comp{\compn{s'}{\ptr{p}}}{nonce}$ $\wedge$ $\ptr{p} \in
  \ptr{W}$ ($\ptr{w'}$ is a top-level window---it has an
  opener---and $\ptr{w}$ is allowed to navigate the opener
  window of $\ptr{w'}$, $\ptr{p}$), then $\ptr{w'} \in
  \ptr{W}$. 
\end{itemize}
\end{definition}

\subsection{Description of the Web Browser Atomic
  Process}\label{app:descr-web-brows}
We will now describe the relation $R^p$ of a standard HTTP
browser $p$. We define $\left(\left(\an{\an{a,f,m}},
    s\right), \left(M, s'\right)\right)$ to belong to $R^p$
if\/f the non-deterministic algorithm presented below, when
given $\left(\an{a,f,m}, s\right)$ as input, terminates
with \textbf{stop}~$M$,~$s'$, i.e., with output $M$ and
$s'$. Recall that $\an{a,f,m}$ is an (input) event and $s$
is a (browser) state, $M$ is a sequence of (output)
protoevents, and $s'$ is a new (browser) state (potentially
with placeholders for nonces).

\paragraph{Notations.} The notation $\textbf{let}\ n \leftarrow N$ is used to
describe that $n$ is chosen non-de\-ter\-mi\-nis\-tic\-ally from the
set $N$.  We write $\textbf{for each}\ s \in M\ \textbf{do}$
to denote that the following commands (until \textbf{end
  for}) are repeated for every element in $M$, where the
variable $s$ is the current element. The order in which the
elements are processed is chosen non-deterministically. %
We will write, for example, 
\begin{algorithmic}
  \LetST{$x,y$}{$\an{\str{Constant},x,y} \equiv
    t$}{doSomethingElse}
\end{algorithmic} \setlength{\parindent}{1em}
for some variables $x,y$, a string
$\str{Constant}$, and some term $t$ to express that $x :=
\proj{2}{t}$, and $y := \proj{3}{t}$ if $\str{Constant}
\equiv \proj{1}{t}$ and if $|\an{\str{Constant},x,y}| =
|t|$,  and that otherwise
$x$ and $y$ are not set and doSomethingElse is executed.

\paragraph{Placeholders.} In several places throughout the
algorithms presented next we use placeholders to generate
``fresh'' nonces as described in our communication model
(see Definition~\ref{def:terms}).
Figure~\ref{fig:browser-placeholder-list} shows a list of
all placeholders used.

\begin{figure}[htb]
  \centering
  \begin{tabular}{|@{\hspace{1ex}}l@{\hspace{1ex}}|@{\hspace{1ex}}l@{\hspace{1ex}}|}\hline 
    \hfill Placeholder\hfill  &\hfill  Usage\hfill  \\\hline\hline
    $\nu_1$ & Algorithm~\ref{alg:browsermain}, new window nonces  \\\hline
    $\nu_2$ & Algorithm~\ref{alg:browsermain}, new HTTP request nonce   \\\hline
    $\nu_3$ & Algorithm~\ref{alg:browsermain}, lookup key for pending HTTP requests entry  \\\hline
    $\nu_4$ & Algorithm~\ref{alg:runscript}, new HTTP request nonce (multiple lines)  \\\hline
    $\nu_5$ & Algorithm~\ref{alg:runscript}, new subwindow nonce  \\\hline
    $\nu_6$ & Algorithm~\ref{alg:processresponse}, new HTTP request nonce  \\\hline
    $\nu_7$ & Algorithm~\ref{alg:processresponse}, new document nonce   \\\hline
    $\nu_8$ & Algorithm~\ref{alg:send}, lookup key for pending DNS entry  \\\hline
    $\nu_9$ & Algorithm~\ref{alg:getnavigablewindow}, new window nonce  \\\hline
    $\nu_{10}, \dots$ & Algorithm~\ref{alg:runscript}, replacement for placeholders in scripting process output   \\\hline

  \end{tabular}
  
  \caption{List of placeholders used in browser algorithms.}
  \label{fig:browser-placeholder-list}
\end{figure}

Before we describe the main browser algorithm, we first
define some functions.

\subsubsection{Functions.} \label{app:proceduresbrowser} In
the description of the following functions we use $a$,
$f$, $m$, and $s$ as read-only global input
variables. All other variables are local variables or
arguments.

The following function, $\mathsf{GETNAVIGABLEWINDOW}$, is
called by the browser to determine the window that is
\emph{actually} navigated when a script in the window
$s'.\ptr{w}$ provides a window reference for navigation
(e.g., for opening a link). When it is given a window
reference (nonce) $\mi{window}$,
$\mathsf{GETNAVIGABLEWINDOW}$ returns a pointer to a
selected window term in $s'$:
\begin{itemize}
\item If $\mi{window}$ is the string $\wBlank$, a new
  window is created and a pointer to that window is
  returned.
\item If $\mi{window}$ is a nonce (reference) and there is
  a window term with a reference of that value in the
  windows in $s'$, a pointer $\ptr{w'}$ to that window term
  is returned, as long as the window is navigable by the
  current window's document (as defined by
  $\mathsf{NavigableWindows}$ above).
\end{itemize}
In all other cases, $\ptr{w}$ is returned instead (the
script navigates its own window).
\captionof{algorithm}{\label{alg:getnavigablewindow}
  Determine window for navigation.}
\begin{algorithmic}[1]
  \Function{$\mathsf{GETNAVIGABLEWINDOW}$}{$\ptr{w}$, $\mi{window}$, $\mi{noreferrer}$, $s'$}
    \If{$\mi{window} \equiv \wBlank$} \Comment{Open a new window when $\wBlank$ is used}
      \If{$\mi{noreferrer} \equiv \bot$}
        \Let{$w'$}{$\an{\nu_9, \an{}, \comp{\compn{s'}{\ptr{w}}}{nonce} }$}
      \Else
        \Let{$w'$}{$\an{\nu_9, \an{}, \bot}$}
      \EndIf
      \Append{$w'$}{$\comp{s'}{windows}$} \textbf{and} let
      $\ptr{w}'$ be a pointer to this new element in $s'$
      \State \Return{$\ptr{w}'$}
    \EndIf
    \LetNDST{$\ptr{w}'$}{$\mathsf{NavigableWindows}(\ptr{w},
      s')$}{$\comp{\compn{s'}{\ptr{w}'}}{nonce} \equiv
      \mi{window}$}{\textbf{return} $\ptr{w}$} %
    \State \Return{$\ptr{w'}$}
  \EndFunction
\end{algorithmic} \setlength{\parindent}{1em}

The following function takes a window reference as input
and returns a pointer to a window as above, but it checks
only that the active documents in both windows are
same-origin. It creates no new windows.
\captionof{algorithm}{\label{alg:getwindow} Determine same-origin window.}
\begin{algorithmic}[1]
  \Function{$\mathsf{GETWINDOW}$}{$\ptr{w}$, $\mi{window}$, $s'$}
    \LetNDST{$\ptr{w}'$}{$\mathsf{Subwindows}(s')$}{$\comp{\compn{s'}{\ptr{w}'}}{nonce} \equiv \mi{window}$}{\textbf{return} $\ptr{w}$} %
    \If{
      $\comp{\comp{\compn{s'}{\ptr{w}'}}{activedocument}}{origin}
      \equiv
      \comp{\comp{\compn{s'}{\ptr{w}}}{activedocument}}{origin}$
    }
      \State \Return{$\ptr{w}'$}
    \EndIf
    \State \Return{$\ptr{w}$}
  \EndFunction
\end{algorithmic} \setlength{\parindent}{1em}

The next function is used to stop any pending
requests for a specific window. From the pending requests
and pending DNS requests it removes any requests with the
given window reference $n$.
\captionof{algorithm}{\label{alg:cancelnav} Cancel pending requests for given window.}
\begin{algorithmic}[1]
  \Function{$\mathsf{CANCELNAV}$}{$n$, $s'$}
    \State \textbf{remove all} $\an{n, \mi{req}, \mi{key}, \mi{f}}$ \textbf{ from } $\comp{s'}{pendingRequests}$ \textbf{for any} $\mi{req}$, $\mi{key}$, $\mi{f}$
    \State \textbf{remove all} $\an{x, \an{n, \mi{message}, \mi{protocol}}}$ \textbf{ from } $\comp{s'}{pendingDNS}$ \textbf{for any} $\mi{x}$, $\mi{message}$, $\mi{protocol}$
    \State \Return{$s'$}
  \EndFunction
\end{algorithmic} \setlength{\parindent}{1em}

The following function takes an HTTP request
$\mi{message}$ as input, adds cookie and origin headers to
the message, creates a DNS request for the hostname given
in the request and stores the request in
$\comp{s'}{pendingDNS}$ until the DNS resolution
finishes. For normal HTTP requests, $\mi{reference}$ is a
window reference. For \xhrs, $\mi{reference}$ is a value of
the form $\an{\mi{document}, \mi{nonce}}$ where
$\mi{document}$ is a document reference and $\mi{nonce}$ is
some nonce that was chosen by the script that initiated the
request. $\mi{protocol}$ is either $\http$ or
$\https$. $\mi{origin}$ is the origin header value that is
to be added to the HTTP request.

\captionof{algorithm}{\label{alg:send} Prepare headers, do DNS resolution, save message. }
\begin{algorithmic}[1]
  \Function{$\mathsf{SEND}$}{$\mi{reference}$, $\mi{message}$, $\mi{protocol}$, $\mi{origin}$, $\mi{referrer}$, $s'$}
    \If{$\comp{\mi{message}}{host} \inPairing \comp{s'}{sts}$}
      \Let{$\mi{protocol}$}{$\https$}
    \EndIf
    \Let{ $\mi{cookies}$}{$\langle\{\an{\comp{c}{name}, \comp{\comp{c}{content}}{value}} | c\inPairing \comp{s'}{cookies}\left[\comp{\mi{message}}{host}\right]$} \label{line:assemble-cookies-for-request} \breakalgohook{1} $\wedge \left(\comp{\comp{c}{content}}{secure} \implies \left(\mi{protocol} = \https\right)\right) \}\rangle$ \label{line:cookie-rules-http}
    \Let{$\comp{\mi{message}}{headers}[\str{Cookie}]$}{$\mi{cookies}$}
    \If{$\mi{origin} \not\equiv \bot$}
      \Let{$\comp{\mi{message}}{headers}[\str{Origin}]$}{$\mi{origin}$}
    \EndIf
    \If{$\mi{referrer} \not\equiv \bot$}
      \Let{$\comp{\mi{message}}{headers}[\str{Referer}]$}{$\mi{referrer}$}
    \EndIf
    \Let{$\comp{s'}{pendingDNS}[\nu_8]$}{$\an{\mi{reference},
        \mi{message}, \mi{protocol}}$} \label{line:add-to-pendingdns}
    \State \textbf{stop} $\an{\an{\comp{s'}{DNSaddress},a,
    \an{\cDNSresolve, \mi{host}, n}}}$, $s'$
  \EndFunction
\end{algorithmic} \setlength{\parindent}{1em}
\noindent

The function $\mathsf{RUNSCRIPT}$ performs a script
execution step of the script in the document
$\compn{s'}{\ptr{d}}$ (which is part of the window
$\compn{s'}{\ptr{w}}$). A new script and document state is
chosen according to the relation defined by the script and
the new script and document state is saved. Afterwards, the
$\mi{command}$ that the script issued is interpreted. 

\captionof{algorithm}{\label{alg:runscript} Execute a script.}
\begin{algorithmic}[1]
  \Function{$\mathsf{RUNSCRIPT}$}{$\ptr{w}$, $\ptr{d}$, $s'$}
    \Let{$\mi{tree}$}{$\mathsf{Clean}(s', \compn{s'}{\ptr{d}})$} \label{line:clean-tree}

    \Let{$\mi{cookies}$}{$\langle\{\an{\comp{c}{name}, \comp{\comp{c}{content}}{value}} | c \inPairing \comp{s'}{cookies}\left[  \comp{\comp{\compn{s'}{\ptr{d}}}{origin}}{host}  \right]$
     \breakalgohook{1} $\wedge\,\comp{\comp{c}{content}}{httpOnly} = \bot$ \breakalgohook{1} $\wedge\,\left(\comp{\comp{c}{content}}{secure} \implies \left(\comp{\comp{\compn{s'}{\ptr{d}}}{origin}}{protocol} \equiv \https\right)\right) \}\rangle$} \label{line:assemble-cookies-for-script}
    \LetND{$\mi{tlw}$}{$\comp{s'}{windows}$ \textbf{such that} $\mi{tlw}$ is the top-level window containing $\ptr{d}$} 
    \Let{$\mi{sessionStorage}$}{$\comp{s'}{sessionStorage}\left[\an{\comp{\compn{s'}{\ptr{d}}}{origin}, \comp{\mi{tlw}}{nonce}}\right]$} %
    \Let{$\mi{localStorage}$}{$\comp{s'}{localStorage}\left[\comp{\compn{s'}{\ptr{d}}}{origin}\right]$}
    \Let{$\mi{secret}$}{$\comp{s'}{secrets}\left[\comp{\compn{s'}{\ptr{d}}}{origin}\right]$} \label{line:browser-secrets}
    \LetND{$R$}{$\mathsf{script}^{-1}(\comp{\compn{s'}{\ptr{d}}}{script})$} %
    \Let{$\mi{in}$}{$\langle\mi{tree}$, $\comp{\compn{s'}{\ptr{d}}}{nonce}, \comp{\compn{s'}{\ptr{d}}}{scriptstate}$, $\comp{\compn{s'}{\ptr{d}}}{scriptinputs}$, $\mi{cookies}$,  $\mi{localStorage}$, $\mi{sessionStorage}$, $\comp{s'}{ids}$, $\mi{secret}\rangle$}\label{line:browser-scriptinputs}
    \LetND{$\mi{state}'$}{$\terms(V)$, \breakalgohook{1}
      $\mi{cookies}' \gets \mathsf{Cookies}^\nu$, \breakalgohook{1}
      $\mi{localStorage}' \gets \terms(V)$,\breakalgohook{1}
      $\mi{sessionStorage}' \gets \terms(V)$,\breakalgohook{1}
      $\mi{command} \gets \terms(V)$, \breakalgohook{1} 
      $\mi{out}^\lambda := \an{\mi{state}', \mi{cookies}', \mi{localStorage}',$ $\mi{sessionStorage}', \mi{command}}$
      \breakalgohook{1} \textbf{such that} $(\mi{in}, \mi{out}^\lambda) \in R$}  \label{line:trigger-script} 
    \Let{$\mi{out}$}{$\mi{out}^\lambda[\nu_{10}/\lambda_1, \nu_{11}/\lambda_2, \dots]$}

    \Let{$\comp{s'}{cookies}\left[\comp{\comp{\compn{s'}{\ptr{d}}}{origin}}{host}\right]$}{$\langle\mathsf{CookieMerge}(\comp{s'}{cookies}\left[\comp{\comp{\compn{s'}{\ptr{d}}}{origin}}{host}\right]$, $\mi{cookies}')\rangle$} \label{line:cookiemerge}
    \Let{$\comp{s'}{localStorage}\left[\comp{\compn{s'}{\ptr{d}}}{origin}\right]$}{$\mi{localStorage}'$}
    \Let{$\comp{s'}{sessionStorage}\left[\an{\comp{\compn{s'}{\ptr{d}}}{origin}, \comp{\mi{tlw}}{nonce}}\right]$}{$\mi{sessionStorage}'$}
    \Let{$\comp{\compn{s'}{\ptr{d}}}{scriptstate}$}{$state'$}
    \Switch{$\mi{command}$}
      \Case{$\an{\tHref, \mi{url},
          \mi{hrefwindow}, \mi{noreferrer}}$} \df{$\nf$ Should we check types here?}
      \Let{$\ptr{w}'$}{$\mathsf{GETNAVIGABLEWINDOW}$($\ptr{w}$,
        $\mi{hrefwindow}$, $\mi{noreferrer}$, $s'$)} 
      \Let{$\mi{req}$}{$\hreq{ nonce=\nu_4, 
          method=\mGet, host=\comp{\mi{url}}{host},
          path=\comp{\mi{url}}{path},
          headers=\an{},
          parameters=\comp{\mi{url}}{parameters}, body=\an{}
        }$}
      \If{$\mi{noreferrer} \equiv \bot$}
        \Let{$\mi{referrer}$}{$\comp{\compn{s'}{\ptr{d}}}{location}$}
      \Else
        \Let{$\mi{referrer}$}{$\bot$}
      \EndIf
      \Let{$s'$}{$\mathsf{CANCELNAV}(\comp{\compn{s'}{\ptr{w}'}}{nonce}, s')$}
      \State \textsf{SEND}($\comp{\compn{s'}{\ptr{w}'}}{nonce}$, $\mi{req}$, $\comp{\mi{url}}{protocol}$, $\bot$, $\mi{referrer}$, $s'$) \label{line:send-href}
      \EndCase
      \Case{$\an{\tIframe, \mi{url}, \mi{window}}$}
        \Let{$\ptr{w}'$}{$\mathsf{GETWINDOW}(\ptr{w}, \mi{window}, s')$}
        \Let{$\mi{req}$}{$\hreq{
            nonce=\nu_4,
            method=\mGet,
            host=\comp{\mi{url}}{host},
            path=\comp{\mi{url}}{path},
            headers=\an{},
            parameters=\comp{\mi{url}}{parameters},
            body=\an{}
          }$}
        \Let{$\mi{referrer}$}{$s'.\ptr{w}'.\str{activedocument}.\str{location}$}
        \Let{$w'$}{$\an{\nu_5, \an{}, \bot}$}
        \Let{$\comp{\comp{\compn{s'}{\ptr{w}'}}{activedocument}}{subwindows}$}{ $\comp{\comp{\compn{s'}{\ptr{w}'}}{activedocument}}{subwindows} \plusPairing w'$}
        \State \textsf{SEND}($\nu_5$, $\mi{req}$, $\comp{\mi{url}}{protocol}$, $\bot$, $\mi{referrer}$, $s'$) \label{line:send-iframe}
      \EndCase
      \Case{$\an{\tForm, \mi{url}, \mi{method}, \mi{data}, \mi{hrefwindow}}$}
        \If{$\mi{method} \not\in \{\mGet, \mPost\}$} \footnote{The working draft for HTML5 allowed for DELETE and PUT methods in HTML5 forms. However, these have since been removed. See \url{http://www.w3.org/TR/2010/WD-html5-diff-20101019/\#changes-2010-06-24}.}
          \State \textbf{stop} $\an{}$, $s'$
        \EndIf
        \Let{$\ptr{w}'$}{$\mathsf{GETNAVIGABLEWINDOW}$($\ptr{w}$, $\mi{hrefwindow}$, $\bot$, $s'$)}
        \If{$\mi{method} = \mGet$}
          \Let{$\mi{body}$}{$\an{}$}
          \Let{$\mi{parameters}$}{$\mi{data}$}
          \Let{$\mi{origin}$}{$\bot$}
        \Else
          \Let{$\mi{body}$}{$\mi{data}$}
          \Let{$\mi{parameters}$}{$\comp{\mi{url}}{parameters}$}
          \Let{$\mi{origin}$}{$\comp{\compn{s'}{\ptr{d}}}{origin}$}
        \EndIf
        \Let{$\mi{req}$}{$\hreq{
            nonce=\nu_4,
            method=\mi{method},
            host=\comp{\mi{url}}{host},
            path=\comp{\mi{url}}{path},
            headers=\an{},
            parameters=\mi{parameters},
            xbody=\mi{body}
          }$}
        \Let{$\mi{referrer}$}{$\comp{\compn{s'}{\ptr{d}}}{location}$}
        \Let{$s'$}{$\mathsf{CANCELNAV}(\comp{\compn{s'}{\ptr{w}'}}{nonce}, s')$}
        \State \textsf{SEND}($\comp{\compn{s'}{\ptr{w}'}}{nonce}$, $\mi{req}$, $\comp{\mi{url}}{protocol}$, $\mi{origin}$, $\mi{referrer}$, $s'$) \label{line:send-form}
      \EndCase
      \Case{$\an{\tSetScript, \mi{window}, \mi{script}}$}
        \Let{$\ptr{w}'$}{$\mathsf{GETWINDOW}(\ptr{w}, \mi{window}, s')$}
        \Let{$\comp{\comp{\compn{s'}{\ptr{w}'}}{activedocument}}{script}$}{$\mi{script}$}
        \State \textbf{stop} $\an{}$, $s'$
      \EndCase
      \Case{$\an{\tSetScriptState, \mi{window}, \mi{scriptstate}}$}
        \Let{$\ptr{w}'$}{$\mathsf{GETWINDOW}(\ptr{w}, \mi{window}, s')$}
        \Let{$\comp{\comp{\compn{s'}{\ptr{w}'}}{activedocument}}{scriptstate}$}{$\mi{scriptstate}$}
        \State \textbf{stop} $\an{}$, $s'$
      \EndCase
      \Case{$\an{\tXMLHTTPRequest, \mi{url}, \mi{method}, \mi{data}, \mi{xhrreference}}$}
        \If{$\mi{method} \in \{\mConnect, \mTrace, \mTrack\} \wedge \mi{xhrreference} \not\in \{\nonces, \bot\}$} 
          \State \textbf{stop} $\an{}$, $s'$
        \EndIf
        \If{$\comp{\mi{url}}{host} \not\equiv \comp{\comp{\compn{s'}{\ptr{d}}}{origin}}{host}$  $\vee$ $\comp{\mi{url}}{protocol} \not\equiv \comp{\comp{\compn{s'}{\ptr{d}}}{origin}}{protocol}$} 
          \State \textbf{stop} $\an{}$, $s'$
        \EndIf
        \If{$\mi{method} \in \{\mGet, \mHead\}$}
          \Let{$\mi{data}$}{$\an{}$}
          \Let{$\mi{origin}$}{$\bot$}
        \Else
          \Let{$\mi{origin}$}{$\comp{\compn{s'}{\ptr{d}}}{origin}$}
        \EndIf
        \Let{$\mi{req}$}{$\hreq{
            nonce=\nu_4,
            method=\mi{method},
            host=\comp{\mi{url}}{host},
            path=\comp{\mi{url}}{path},
            headers={},
            parameters=\comp{\mi{url}}{parameters},
            xbody=\mi{data}
          }$}
        \Let{$\mi{referrer}$}{$\comp{\compn{s'}{\ptr{d}}}{location}$}
        \State \textsf{SEND}($\an{\comp{\compn{s'}{\ptr{d}}}{nonce}, \mi{xhrreference}}$, $\mi{req}$, $\comp{\mi{url}}{protocol}$, $\mi{origin}$, $\mi{referrer}$, $s'$)\label{line:send-xhr}
      \EndCase
      \Case{$\an{\tBack, \mi{window}}$} \footnote{Note that
        navigating a window using the back/forward buttons
        does not trigger a reload of the affected
        documents. While real world browser may chose to
        refresh a document in this case, we assume that the
        complete state of a previously viewed document is
        restored. A reload can be triggered
        non-deterministically at any point (in the main algorithm).}
      \Let{$\ptr{w}'$}{$\mathsf{GETNAVIGABLEWINDOW}$($\ptr{w}$,
        $\mi{window}$, $\bot$, $s'$)} \If{$\exists\, \ptr{j} \in
        \mathbb{N}, \ptr{j} > 1$ \textbf{such that}
        $\comp{\compn{\comp{\compn{s'}{\ptr{w'}}}{documents}}{\ptr{j}}}{active}
        \equiv \True$} %
      \Let{$\comp{\compn{\comp{\compn{s'}{\ptr{w'}}}{documents}}{\ptr{j}}}{active}$}{$\bot$}
      \Let{$\comp{\compn{\comp{\compn{s'}{\ptr{w'}}}{documents}}{(\ptr{j}-1)}}{active}$}{$\True$}
      \Let{$s'$}{$\mathsf{CANCELNAV}(\comp{\compn{s'}{\ptr{w}'}}{nonce},
        s')$}
        \EndIf
        \State \textbf{stop} $\an{}$, $s'$
      \EndCase
      \Case{$\an{\tForward, \mi{window}}$}
        \Let{$\ptr{w}'$}{$\mathsf{GETNAVIGABLEWINDOW}$($\ptr{w}$, $\mi{window}$, $\bot$, $s'$)}
        \If{$\exists\, \ptr{j} \in \mathbb{N} $ \textbf{such that} $\comp{\compn{\comp{\compn{s'}{\ptr{w'}}}{documents}}{\ptr{j}}}{active} \equiv \True$  $\wedge$  $\compn{\comp{\compn{s'}{\ptr{w'}}}{documents}}{(\ptr{j}+1)} \in \mathsf{Documents}$} %
          \Let{$\comp{\compn{\comp{\compn{s'}{\ptr{w'}}}{documents}}{\ptr{j}}}{active}$}{$\bot$}
          \Let{$\comp{\compn{\comp{\compn{s'}{\ptr{w'}}}{documents}}{(\ptr{j}+1)}}{active}$}{$\True$}
          \Let{$s'$}{$\mathsf{CANCELNAV}(\comp{\compn{s'}{\ptr{w}'}}{nonce}, s')$}
        \EndIf
        \State \textbf{stop} $\an{}$, $s'$
      \EndCase
      \Case{$\an{\tClose, \mi{window}}$}
        \Let{$\ptr{w}'$}{$\mathsf{GETNAVIGABLEWINDOW}$($\ptr{w}$, $\mi{window}$, $\bot$, $s'$)}
        \State \textbf{remove} $\compn{s'}{\ptr{w'}}$ from the sequence containing it 
        \State \textbf{stop} $\an{}$, $s'$
      \EndCase

      \Case{$\an{\tPostMessage, \mi{window}, \mi{message}, \mi{origin}}$}
        \LetND{$\ptr{w}'$}{$\mathsf{Subwindows}(s')$ \textbf{such that} $\comp{\compn{s'}{\ptr{w}'}}{nonce} \equiv \mi{window}$} %
        \If{$\exists \ptr{j} \in \mathbb{N}$ \textbf{such that} $\comp{\compn{\comp{\compn{s'}{\ptr{w'}}}{documents}}{\ptr{j}}}{active} \equiv \True$ \breakalgohook{3} $\wedge  (\mi{origin} \not\equiv \bot \implies \comp{\compn{\comp{\compn{s'}{\ptr{w'}}}{documents}}{\ptr{j}}}{origin} \equiv \mi{origin})$}    \label{line:append-pm-to-scriptinputs-condition} %
        \Let{$\comp{\compn{\comp{\compn{s'}{\ptr{w'}}}{documents}}{\ptr{j}}}{scriptinputs}$\breakalgohook{4}}{ $\comp{\compn{\comp{\compn{s'}{\ptr{w'}}}{documents}}{\ptr{j}}}{scriptinputs}$ \breakalgohook{4} $\plusPairing$
         $\an{\tPostMessage, \comp{\compn{s'}{\ptr{w}}}{nonce}, \comp{\compn{s'}{\ptr{d}}}{origin}, \mi{message}}$} \label{line:append-pm-to-scriptinputs}
        \EndIf
        \State \textbf{stop} $\an{}$, $s'$
      \EndCase
      \Case{else}
        \State \textbf{stop} $\an{}$, $s'$
      \EndCase
    \EndSwitch
  \EndFunction
\end{algorithmic} \setlength{\parindent}{1em}

The function $\mathsf{PROCESSRESPONSE}$ is responsible for
processing an HTTP response ($\mi{response}$) that was
received as the response to a request ($\mi{request}$) that
was sent earlier. In $\mi{reference}$, either a window or a
document reference is given (see explanation for
Algorithm~\ref{alg:send} above). Again, $\mi{protocol}$ is
either $\http$ or $\https$.

The function first saves any cookies that were contained in
the response to the browser state, then checks whether a
redirection is requested (Location header). If that is not
the case, the function creates a new document (for normal
requests) or delivers the contents of the response to the
respective receiver (for \xhr responses).
\captionof{algorithm}{\label{alg:processresponse} Process an HTTP response.}
\begin{algorithmic}[1]
\Function{$\mathsf{PROCESSRESPONSE}$}{$\mi{response}$, $\mi{reference}$, $\mi{request}$, $\mi{protocol}$, $s'$}
  \If{$\mathtt{Set{\mhyphen}Cookie} \in
    \comp{\mi{response}}{headers}$}
    \For{\textbf{each} $c \inPairing \comp{\mi{response}}{headers}\left[\mathtt{Set{\mhyphen}Cookie}\right]$, $c \in \mathsf{Cookies}$}
      \Let{$\comp{s'}{cookies}\left[\comp{\comp{\mi{request}}{url}}{host}\right]$}{$\mathsf{AddCookie}(\comp{s'}{cookies}\left[\comp{\comp{\mi{request}}{url}}{host}\right], c)$} \label{line:set-cookie}
    \EndFor
  \EndIf  
  \If{$\mathtt{Strict{\mhyphen}Transport{\mhyphen}Security} \in \comp{\mi{response}}{headers}$ $\wedge$ $\mi{protocol} \equiv \https$}
    \Append{$\comp{\mi{request}}{host}$}{$\comp{s'}{sts}$}
  \EndIf
  \If{$\str{Referer} \in \comp{request}{headers}$}
    \Let{$\mi{referrer}$}{$\comp{request}{headers}[\str{Referer}]$}
  \Else
    \Let{$\mi{referrer}$}{$\bot$}
  \EndIf

  \If{$\mathtt{Location} \in \comp{\mi{response}}{headers} \wedge \comp{\mi{response}}{status} \in \{303, 307\}$} \label{line:location-header} \footnote{The RFC for HTTPbis (currently in draft status), which obsoletes RFC 2616, does not specify whether a POST/DELETE/etc. request that was answered with a status code of 301 or 302 should be rewritten to a GET request or not (``for historic reasons'' that are detailed in Section~7.4.). 
As the specification is clear for the status codes 303 and 307 (and most browsers actually follow the specification in this regard), we focus on modeling these.}
    \Let{$\mi{url}$}{$\comp{\mi{response}}{headers}\left[\mathtt{Location}\right]$}
    \Let{$\mi{method}'$}{$\comp{\mi{request}}{method}$} \footnote{While the standard demands that users confirm redirections of non-safe-methods (e.g., POST), we assume that users generally confirm these redirections.}
    \Let{$\mi{body}'$}{$\comp{\mi{request}}{body}$} \footnote{If, for example, a GET request is redirected and the original request contained a body, this body is preserved, as HTTP allows for payloads in messages with all HTTP methods, except for the TRACE method (a detail which we omit). 
Browsers will usually not send body payloads for methods that do not specify semantics for such data in the first place.}
    \If{$\str{Origin} \in \comp{request}{headers}$}
      \Let{$\mi{origin}$}{$\an{\comp{request}{headers}[\str{Origin}], \an{\comp{request}{host}, \mi{protocol}}}$}
    \Else
      \Let{$\mi{origin}$}{$\bot$}
    \EndIf
    \If{$\comp{\mi{response}}{status} \equiv 303 \wedge \comp{\mi{request}}{method} \not\in \{\mGet, \mHead\}$}
      \Let {$\mi{method}'$}{$\mGet$}
      \Let{$\mi{body}'$}{$\an{}$}
    \EndIf
    \If{$\nexists\, \ptr{w} \in \mathsf{Subwindows}(s')$ \textbf{such that} $\comp{\compn{s'}{\ptr{w}}}{nonce} \equiv \mi{reference}$} \Comment{Do not redirect XHRs.}
      \State \textbf{stop} $\an{}$, $s$
    \EndIf
    \Let{$\mi{req}$}{$\hreq{
            nonce=\nu_6,
            method=\mi{method'},
            host=\comp{\mi{url}}{host},
            path=\comp{\mi{url}}{path},
            headers=\an{},
            parameters=\comp{\mi{url}}{parameters},
            xbody=\mi{body}'
          }$}
    \State \textsf{SEND}($\mi{reference}$, $\mi{req}$, $\comp{\mi{url}}{protocol}$, $\mi{origin}$, $\mi{referrer}$, $s'$)\label{line:send-redirect}
  \EndIf

  \If{$\exists\, \ptr{w} \in \mathsf{Subwindows}(s')$ \textbf{such that} $\comp{\compn{s'}{\ptr{w}}}{nonce} \equiv \mi{reference}$} \Comment{normal response}
    \Let{$\mi{location}$}{$\an{\cUrl, \mi{protocol}, \comp{\mi{request}}{host}, \comp{\mi{request}}{path}, \comp{\mi{request}}{parameters}}$}
    \If{$\mi{response}.\str{body} \not\sim \an{*,*}$}
      \State \textbf{stop} $\{\}$, $s'$
    \EndIf
    \Let{$\mi{script}$}{$\proj{1}{\comp{\mi{response}}{body}}$}
    \Let{$\mi{scriptstate}$}{$\proj{2}{\comp{\mi{response}}{body}}$}
    \Let{$d$}{$\an{\nu_7, \mi{location}, \mi{referrer}, \mi{script}, \mi{scriptstate}, \an{}, \an{}, \True}$} \label{line:take-script} \label{line:set-origin-of-document}
    \If{$\comp{\compn{s'}{\ptr{w}}}{documents} \equiv \an{}$}
      \Let{$\comp{\compn{s'}{\ptr{w}}}{documents}$}{$\an{d}$}
    \Else
      \LetND{$\ptr{i}$}{$\mathbb{N}$ \textbf{such that} $\comp{\compn{\comp{\compn{s'}{\ptr{w}}}{documents}}{\ptr{i}}}{active} \equiv \True$} %
      \Let{$\comp{\compn{\comp{\compn{s'}{\ptr{w}}}{documents}}{\ptr{i}}}{active}$}{$\bot$}
      \State \textbf{remove} $\compn{\comp{\compn{s'}{\ptr{w}}}{documents}}{(\ptr{i}+1)}$ and all following documents  from $\comp{\compn{s'}{\ptr{w}}}{documents}$
      \Append{$d$}{$\comp{\compn{s'}{\ptr{w}}}{documents}$}
    \EndIf
    \State \textbf{stop} $\{\}$, $s'$
  \ElsIf{$\exists\, \ptr{w} \in \mathsf{Subwindows}(s')$, $\ptr{d}$ \textbf{such that} $\comp{\compn{s'}{\ptr{d}}}{nonce} \equiv \proj{1}{\mi{reference}} $ \breakalgohook{1}  $\wedge$  $\compn{s'}{\ptr{d}} = \comp{\compn{s'}{\ptr{w}}}{activedocument}$} \label{line:process-xhr-response} \Comment{process XHR response}
    \Append{$\an{\tXMLHTTPRequest, \comp{\mi{response}}{body}, \proj{2}{\mi{reference}}}$}{$\comp{\compn{s'}{\ptr{d}}}{scriptinputs}$}
  \EndIf
\EndFunction
\end{algorithmic} \setlength{\parindent}{1em}

\subsubsection{Main Algorithm.}\label{app:mainalgorithmwebbrowserprocess}
This is the main algorithm of the browser relation.
It receives the message $m$ as input, as
well as $a$, $f$ and $s$ as above.

\captionof{algorithm}{\label{alg:browsermain} Main Algorithm}
\begin{algorithmic}[1]
\Statex[-1] \textbf{Input:} $\an{a,f,m},s$
  \Let{$s'$}{$s$}

  \If{$\comp{s}{isCorrupted} \not\equiv \bot$}
    \Let{$\comp{s'}{pendingRequests}$}{$\an{m, \comp{s}{pendingRequests}}$} \Comment{Collect incoming messages}
    \LetND{$m'$}{$d_{V}(s')$} %
    \LetND{$a'$}{$\addresses$} %
    \State \textbf{stop} $\an{\an{a',a,m'}}$, $s'$
  \EndIf
  \If{$m \equiv \trigger$} \Comment{A special trigger message. }
    \LetND{$\mi{switch}$}{$\{1,2,3\}$} \label{line:browser-switch} %
    \If{$\mi{switch} \equiv 1$} \Comment{Run some script.}
      \LetNDST{$\ptr{w}$}{$\mathsf{Subwindows}(s')$}{$\comp{\compn{s'}{\ptr{w}}}{documents} \neq \an{}$}{\textbf{stop} $\an{}$, $s'$} \label{line:browser-trigger-window}%
      \Let{$\ptr{d}$}{$\ptr{w} \plusPairing \str{activedocument}$}
      \State \textsf{RUNSCRIPT}($\ptr{w}$, $\ptr{d}$, $s'$)
    \ElsIf{$\mi{switch} \equiv 2$} \Comment{Create some new request.}
      \Let{$w'$}{$\an{\nu_1, \an{}, \bot}$}
      \Append{$w'$}{$\comp{s'}{windows}$}
      \LetND{$\mi{protocol}$}{$\{\http, \https\}$} \label{line:browser-choose-url} %
      \LetND{$\mi{host}$}{$\dns$} %
      \LetND{$\mi{path}$}{$\mathbb{S}$} %
      \LetND{$\mi{parameters}$}{$\dict{\mathbb{S}}{\mathbb{S}}$} %
      \Let{$\mi{req}$}{$\hreq{
          nonce=\nu_2,
          method=\mGet,
          host=\mi{host},
          path=\mi{path},
          headers=\an{},
          parameters=\mi{parameters},
          body=\an{}
        }$}
      \State \textsf{SEND}($\nu_1$, $\mi{req}$, $\mi{protocol}$, $\bot$, $s'$)\label{line:send-random}
    \ElsIf{$\mi{switch} \equiv 3$} \Comment{Reload some document.}
      \LetNDST{$\ptr{w}$}{$\mathsf{Subwindows}(s')$}{$\comp{\compn{s'}{\ptr{w}}}{documents} \neq \an{}$}{\textbf{stop} $\an{}$, $s'$} \label{line:browser-reload-window}%
      \Let{$\mi{url}$}{$s'.\ptr{w}.\str{activedocument}.\str{location}$}
      \Let{$\mi{req}$}{$\hreq{ nonce=\nu_2, 
          method=\mGet, host=\comp{\mi{url}}{host},
          path=\comp{\mi{url}}{path},
          headers=\an{},
          parameters=\comp{\mi{url}}{parameters}, body=\an{}
        }$}
      \Let{$\mi{referrer}$}{$s'.\ptr{w}.\str{activedocument}.\str{referrer}$}
      \Let{$s'$}{$\mathsf{CANCELNAV}(\comp{\compn{s'}{\ptr{w}}}{nonce}, s')$}
      \State \textsf{SEND}($\comp{\compn{s'}{\ptr{w}}}{nonce}$, $\mi{req}$, $\comp{\mi{url}}{protocol}$, $\bot$, $\mi{referrer}$, $s'$)

    \EndIf
  \ElsIf{$m \equiv \fullcorrupt$} \Comment{Request to corrupt browser}
    \Let{$\comp{s'}{isCorrupted}$}{$\fullcorrupt$}
    \State \textbf{stop} $\an{}$, $s'$
  \ElsIf{$m \equiv \closecorrupt$} \Comment{Close the browser}
    \Let{$\comp{s'}{secrets}$}{$\an{}$}  
    \Let{$\comp{s'}{windows}$}{$\an{}$}
    \Let{$\comp{s'}{pendingDNS}$}{$\an{}$}
    \Let{$\comp{s'}{pendingRequests}$}{$\an{}$}
    \Let{$\comp{s'}{sessionStorage}$}{$\an{}$}
    \State \textbf{let} $\comp{s'}{cookies} \subsetPairing \cookies$ \textbf{such that}  $(c \inPairing \comp{s'}{cookies}) {\iff} (c \inPairing \comp{s}{cookies} \wedge \comp{\comp{c}{content}}{session} \equiv \bot$)
    \Let{$\comp{s'}{isCorrupted}$}{$\closecorrupt$}
    \State \textbf{stop} $\an{}$, $s'$
  \ElsIf{$\exists\, \an{\mi{reference}, \mi{request}, \mi{key}, f}$
      $\inPairing \comp{s'}{pendingRequests}$ \breakalgohook{0}
      \textbf{such that} $\proj{1}{\decs{m}{\mi{key}}} \equiv \cHttpResp$ } %
    \Comment{Encrypted HTTP response}
    \Let{$m'$}{$\decs{m}{\mi{key}}$}
    \If{$\comp{m'}{nonce} \not\equiv \comp{\mi{request}}{nonce}$}
      \State \textbf{stop} $\an{}$, $s$
    \EndIf
    \State \textbf{remove} $\an{\mi{reference}, \mi{request}, \mi{key}, f}$ \textbf{from} $\comp{s'}{pendingRequests}$
    \State \textsf{PROCESSRESPONSE}($m'$, $\mi{reference}$, $\mi{request}$, $\https$, $s'$)
  \ElsIf{$\proj{1}{m} \equiv \cHttpResp$ $\wedge$ $\exists\, \an{\mi{reference}, \mi{request}, \bot, f}$ $\inPairing \comp{s'}{pendingRequests}$ \textbf{such that} $\comp{m'}{nonce} \equiv \comp{\mi{request}}{key}$ } %
    \State \textbf{remove} $\an{\mi{reference}, \mi{request}, \bot, f}$ \textbf{from} $\comp{s'}{pendingRequests}$
    \State \textsf{PROCESSRESPONSE}($m$, $\mi{reference}$, $\mi{request}$, $\http$, $s'$)
  \ElsIf{$m \in \dnsresponses$} \Comment{Successful DNS response}
      \If{$\comp{m}{nonce} \not\in \comp{s}{pendingDNS} \vee \comp{m}{result} \not\in \addresses \vee \comp{m}{domain} \not\equiv \comp{\proj{2}{\comp{s}{pendingDNS}}}{host}$}
        \State \textbf{stop} $\an{}$, $s$ \label{line:browser-dns-response-stop}
      \EndIf
      \Let{$\an{\mi{reference}, \mi{message}, \mi{protocol}}$}{$\comp{s}{pendingDNS}[\comp{m}{nonce}]$}
      \If{$\mi{protocol} \equiv \https$}
        \Append{$\langle\mi{reference}$, $\mi{message}$, $\nu_3$, $\comp{m}{result}\rangle$}{$\comp{s'}{pendingRequests}$} \label{line:add-to-pendingrequests-https}
        \Let{$\mi{message}$}{$\enc{\an{\mi{message},\nu_3}}{\comp{s'}{keyMapping}\left[\comp{\mi{message}}{host}\right]}$} \label{line:select-enc-key}
      \Else
        \Append{$\langle\mi{reference}$, $\mi{message}$, $\bot$, $\comp{m}{result}\rangle$}{$\comp{s'}{pendingRequests}$} \label{line:add-to-pendingrequests}
      \EndIf
      \Let{$\comp{s'}{pendingDNS}$}{$\comp{s'}{pendingDNS} - \comp{m}{nonce}$}
      \State \textbf{stop} $\an{\an{\comp{m}{result}, a, \mi{message}}}$, $s'$
  \EndIf
  \State \textbf{stop} $\an{}$, $s$

\end{algorithmic} \setlength{\parindent}{1em}

\section{Formal Model of \spresso}
\label{app:model-spresso-auth}

We here present the full details of our formal model of \spresso.
\gs{added:} For our analysis regarding our authentication and privacy
properties below, we will further restrict this generic model to suit
the setting of the respective analysis.

We model \spresso as a web system (in the sense of
Appendix~\ref{app:websystem}). \gs{this should be a definition:}\df{Would be too long for a definition.} We call a web system
$\spressowebsystem=(\bidsystem, \scriptset, \mathsf{script}, E^0)$ an
\emph{\spresso web system} if it is of the form described in what
follows.

\subsection{Outline}\label{app:outlinespressomodel}
The system $\bidsystem=\mathsf{Hon}\cup \mathsf{Web} \cup
\mathsf{Net}$ consists of web attacker processes (in $\mathsf{Web}$), network attacker processes
(in $\mathsf{Net}$), a finite set $\fAP{FWD}$ of forwarders, a finite set
$\fAP{B}$ of web browsers, a finite set $\fAP{RP}$ of web servers for
the relying parties, a finite set $\fAP{IDP}$ of web servers for
the identity providers, and a finite set $\fAP{DNS}$ of DNS servers, with $\mathsf{Hon} := \fAP{B} \cup \fAP{RP}
\cup \fAP{IDP} \cup \fAP{FWD} \cup \fAP{DNS}$. More details on the processes
in $\bidsystem$ are provided below.
Figure~\ref{fig:scripts-in-w} shows the set of scripts $\scriptset$
and their respective string representations that are defined by the
mapping $\mathsf{script}$.
The set $E^0$ contains only the trigger events as specified in
Appendix~\ref{app:websystem}.

\begin{figure}[htb]
  \centering
  \begin{tabular}{|@{\hspace{1ex}}l@{\hspace{1ex}}|@{\hspace{1ex}}l@{\hspace{1ex}}|}\hline 
    \hfill $s \in \scriptset$\hfill  &\hfill  $\mathsf{script}(s)$\hfill  \\\hline\hline
    $\Rasp$ & $\str{att\_script}$  \\\hline
    $\mi{script\_rp}$ & $\str{script\_rp}$  \\\hline
    $\mi{script\_rp\_redir}$ & $\str{script\_rp\_redir}$  \\\hline
    $\mi{script\_idp}$ &  $\str{script\_idp}$  \\\hline
    $\mi{script\_fwd}$ & $\str{script\_fwd}$ \\\hline
  \end{tabular}
  
  \caption{List of scripts in $\scriptset$ and their respective string
    representations.}
  \label{fig:scripts-in-w}
\end{figure}

This outlines $\spressowebsystem$. We will now define the DY processes in
$\spressowebsystem$ and their addresses, domain names, and secrets in more
detail. The scripts are defined in detail in Appendix~\ref{app:spresso-scripts}.

\subsection{Addresses and Domain Names}
The set $\addresses$ contains for every web attacker in $\fAP{Web}$, every network attacker in $\fAP{Net}$, every relying
party in $\fAP{RP}$, every identity provider in $\fAP{IDP}$, every
forwarder in $\fAP{FWD}$, every DNS server in $\fAP{DNS}$, and every browser in $\fAP{B}$ a finite set
of addresses each. By $\mapAddresstoAP$ we denote the corresponding
assignment from a process to its address. The set $\dns$ contains a
finite set of domains for every forwarder $\fAP{FWD}$, every
relying party in $\fAP{RP}$, every identity provider in $\fAP{IDP}$,
 every web attacker in $\fAP{Web}$, and every network attacker in $\fAP{Net}$. Browsers (in $\fAP{B})$ and DNS servers (in $\fAP{DNS}$) do not have a domain.

By $\mapAddresstoAP$ and $\mapDomain$ we denote the assignments from
atomic processes to sets of $\addresses$ and $\dns$, respectively.

\subsection{Keys and Secrets} The set $\nonces$ of nonces is
partitioned into four sets, an infinite sequence $N$, an infinite set
$K_\text{SSL}$, an infinite set $K_\text{sign}$, and a finite set
$\RPSecrets$. We thus have
\begin{align*}
\def\hereMaxHeightPhantom{\vphantom{K_{\text{p}}^\bidsystem}}
\nonces = 
\underbrace{N\hereMaxHeightPhantom}_{\text{infinite sequence}} 
\dot\cup \underbrace{K_{\text{SSL}}\hereMaxHeightPhantom}_{\text{finite}} 
\dot\cup \underbrace{K_{\text{sign}}\hereMaxHeightPhantom}_{\text{finite}} 
\dot\cup \underbrace{\RPSecrets\hereMaxHeightPhantom}_{\text{finite}}\ .
\end{align*}
The set $N$ contains the nonces that are available for each DY process
in $\bidsystem$ (it can be used to create a run of $\bidsystem$). 

The set $K_\text{SSL}$ contains the keys that will be used for SSL
encryption. Let $\mapSSLKey\colon \dns \to K_\text{SSL}$ be an injective
mapping that assigns a (different) private key to every domain.

The set $K_\text{sign}$ contains the keys that will be used by IdPs
for signing IAs. Let $\mapSignKey\colon \fAP{IdPs} \to K_\text{sign}$
be an injective mapping that assigns a (different) private key to every identity
provider.

The set $\RPSecrets$ is the
set of passwords (secrets) the browsers share with the identity
providers. 

\subsection{Identities}\label{app:spresso-pidp-identities}
Indentites are email addresses, which consist of a user name and a
domain part. For our model, this is defined as follows:
\begin{definition}
  An \emph{identity} (email address) $i$ is a term of the form
  $\an{\mi{name},\mi{domain}}$ with $\mi{name}\in \mathbb{S}$ and
  $\mi{domain} \in \dns$.

  Let $\IDs$ be the finite set of identities. By $\IDs^y$ we denote
  the set $\{ \an{\mi{name}, \mi{domain}} \in \IDs\,|\, \mi{domain}
  \in \mapDomain(y) \}$.

  We say that an ID is \emph{governed} by the DY process to which the
  domain of the ID belongs. Formally, we define the mapping $\mapGovernor:
  \IDs \to \bidsystem$, $\an{\mi{name}, \mi{domain}} \mapsto
  \mapDomain^{-1}(\mi{domain})$.
\end{definition}%

The governor of an ID will usually be an IdP, but could also be the
attacker. 

By $\mapIDtoPLI:\IDs \to \RPSecrets$ we denote the bijective mapping
that assigns secrets to all identities. 

Let $\mapPLItoOwner: \RPSecrets \to \fAP{B}$ denote the mapping that
assigns to each secret a browser that \emph{owns} this secret. Now, we
define the mapping $\mapIDtoOwner: \IDs \to \fAP{B}$, $i \mapsto
\mapPLItoOwner(\mapIDtoPLI(i))$, which assigns to each identity the
browser that owns this identity (we say that the identity belongs to
the browser).

\subsection{Tags and Identity Assertions}
\label{app:identity-assertions}

\begin{definition}\label{def:tag}
  A \emph{tag} is a term of the form $\encs{\an{o,
    n}}{k}$ for some domain $d$, a nonce
  $n \in \nonces$, and a nonce (here used as a symmetric key)
  $k$.
\end{definition}
\begin{definition}
  An \emph{identity assertion (IA)} is a term of the form $\sig{\an{t,
      e, d'}}{k}$ for a tag $t$, an email address (identity) $e$, a
  domain $d'$ and a nonce $k$. We call it an \emph{encrypted identity
    assertion (EIA)} if it is additionally (symmmetrically) encrypted
  (i.e., it is of the form $\encs{s}{k'}$ if $s$ is an IA and $k'$ is
  a nonce.

\end{definition}

\subsection{Corruption}\gs{DNS servers may not be corrupted?}
RPs, IdPs and FWDs can become corrupted: If they receive the message
$\corrupt$, they start collecting all incoming messages in their state
and (upon triggering) send out all messages that are derivable from
their state and collected input messages, just like the attacker
process. We say that an RP, an IdP or an forwarder is \emph{honest} if the according
part of their state ($s.\str{corrupt}$) is $\bot$, and that they are
corrupted otherwise.

We are now ready to define the processes in $\websystem$ as well as
the scripts in $\scriptset$ in more detail.

\subsection{Processes in $\bidsystem$ (Overview)}

We first provide an overview of the processes in $\bidsystem$. All
processes in $\websystem$ (except for DNS servers) contain in their initial states all public
keys and the private keys of their respective domains (if any). We
define $I^p=\mapAddresstoAP(p)$ for all $p\in \mathsf{Hon} \cup \mathsf{Web}$.

\subsubsection{Web Attackers.}  Each $\mi{wa} \in \mathsf{Web}$  is a
web attacker (see Appendix~\ref{app:websystem}), who
uses only his own addresses for sending and listening. 

\subsubsection{Network Attackers.}  Each $\mi{na} \in \mathsf{Net}$  is a
network attacker (see Appendix~\ref{app:websystem}), who
uses all addresses for sending and listening. 

\subsubsection{Browsers.} Each $b \in \fAP{B}$ is a web browser as
defined in Appendix~\ref{app:deta-descr-brows}. The initial state
contains all secrets owned by $b$, stored under the origin of the
respective IdP. See Appendix~\ref{app:browsers-spresso} for details.

\subsubsection{Relying Parties.} A relying party $r \in \fAP{RP}$ is a
web server. RP knows four distinct paths: $\mathtt{/}$, where it
serves the index web page ($\str{script\_rp}$),
$\mathtt{/startLogin}$, where it only accepts POST requests and mainly
issues a fresh RP nonce (details see below), $\mathtt{/redir}$, where it only accepts requests with a valid login session token and serves $\str{script\_rp\_redir}$ to redirect the browser to the IdP, and $\mathtt{/login}$,
where it also only accepts POST requests with login data obtained
during the login process by $\str{script\_rp}$ running in the browser.
It checks this data and, if the data is deemed ``valid'', it issues a
service token (again, for details, see below). The RP keeps a list of
such tokens in its state. Intuitively, a client having such a token
can use the service of the RP (for a specific identity record along
with the token). Just like IdPs, RPs can become corrupted.

\subsubsection{Identity Providers.} Each IdP is a web server. As outlined
in Section~\ref{sec:main-features}, users can authenticate to the IdP with
their credentials. IdP tracks the state of the users with
sessions. Authenticated users can receive IAs from the
IdP. When receiving a special message ($\corrupt$) IdPs can become
corrupted. Similar to the definition of corruption for the browser,
IdPs then start sending out all messages that are derivable from their
state.

\subsubsection{Forwarders.} FWDs are web servers that have only one state and
only serve the script $\str{script\_fwd}$. See
Appendix~\ref{app:fwd-spresso} for details.

\subsubsection{DNS.} Each $\mi{dns} \in \fAP{DNS}$ is a DNS server as defined in Appendix~\ref{app:DNSservers}. Their state contains the allocation of domain names to IP addresses.

\subsection{SSL Key Mapping}\label{app:common-data-structures}
Before we define the atomic DY processes in more detail, we first
define the common data structure that holds the mapping of domain
names to public SSL keys: For an atomic DY process $p$ we define
\[\mi{sslkeys}^p = \an{\left\{\an{d, \mapSSLKey(d)} \mid d \in
    \mapDomain(p)\right\}}.\]

\subsection{Web Attackers}\label{app:webattackers-spresso} Each $\mi{wa} \in \fAP{Web}$ is a web attacker. The initial state of each $\mi{wa}$ is $s_0^\mi{wa} =
\an{\mi{attdoms}, \mi{sslkeys}, \mi{signkeys}}$, where $\mi{attdoms}$
is a sequence of all domains along with the corresponding private keys
owned by $\mi{wa}$, $\mi{sslkeys}$ is a sequence of all domains and
the corresponding public keys, and $\mi{signkeys}$ is a sequence
containing all public signing keys for all IdPs. All other parties use
the attacker as a DNS server.

\subsection{Network Attackers}\label{app:networkattackers-spresso} As mentioned, each network attacker
$\mi{na}$ is modeled to be a network attacker as specified in
Appendix~\ref{app:websystem}. We allow it to listen to/spoof all
available IP addresses, and hence, define $I^\mi{na} =
\addresses$. The initial state is $s_0^\mi{na} =
\an{\mi{attdoms}, \mi{sslkeys}, \mi{signkeys}}$, where $\mi{attdoms}$
is a sequence of all domains along with the corresponding private keys
owned by the attacker $\mi{na}$, $\mi{sslkeys}$ is a sequence of all domains and
the corresponding public keys, and $\mi{signkeys}$ is a sequence
containing all public signing keys for all IdPs.

\subsection{Browsers}\label{app:browsers-spresso} 

Each $b \in \fAP{B}$ is a web browser as defined in
Appendix~\ref{app:deta-descr-brows}, with $I^b := \mapAddresstoAP(b)$
being its addresses.

To define the inital state, first let $\mi{ID}^b :=
\mapIDtoOwner^{-1}(b)$ be
 the set of all IDs of $b$, $\mi{ID}^{b,d} :=
\{i \mid \exists\, x:\ i = \an{x, d} \in \mi{ID}^b\}$ be the set of
IDs of $b$ for a domain $d$, and $\mi{SecretDomains}^b := \{d \mid
\mi{ID}^{b,d} \neq \emptyset \}$ be the set of all domains that $b$
owns identities for.

Then, the initial state $s_0^b$ is defined as follows: the key mapping
maps every domain to its public (ssl) key, according to the mapping
$\mapSSLKey$; the DNS address is $\mapAddresstoAP(p)$ with $p \in \bidsystem$;
the list of secrets contains an entry $\an{\an{d,\https}, s}$ for each
$d \in \mi{SecretDomains}^b$ and $s = \mapIDtoPLI(i)$ for some $i \in
\mi{ID}^{b,d}$ ($s$ is the same for all $i$); $\mi{ids}$ is
$\an{\mi{ID}^b}$; $\mi{sts}$ is empty.

\subsection{Relying Parties} \label{app:relying-parties-spresso}

A relying party $r \in \fAP{RP}$ is a web server modeled as an atomic
DY process $(I^r, Z^r, R^r, s^r_0)$ with the addresses $I^r :=
\mapAddresstoAP(r)$. Its initial state $s^r_0$ contains its domains,
the private keys associated with its domains, the DNS server address,
and the domain name of a forwarder. The full state additionally
contains the sets of service tokens and login session identifiers the
RP has issued. RP only accepts HTTPS requests.

RP manages two kinds of sessions: The \emph{login sessions}, which are
only used during the login phase of a user, and the \emph{service
  sessions} (we call the session identifier of a service session a
\emph{service token}). Service sessions allow a user to use RP's
services. The ultimate goal of a login flow is to establish such a
service session.

In a typical flow with one client, $r$ will first receive an HTTP GET
request for the path $\str{/}$. In this case, $r$ returns the script
$\str{script\_rp}$ (see below).

After the user entered her email address, $r$ will receive an HTTPS
POST XMLHTTPRequest for the path $\str{/startLogin}$. In this request,
it expects the email address the user entered. The relying party then
contacts the user's email provider to retrieve the \spresso support
document (where it extracts the public key of the IdP). After that, $r$ selects
the nonces $\mi{rpNonce}$, $\mi{iaKey}$, $\mi{tagKey}$, and
$\mi{loginSessionToken}$. It creates the $\mi{tag}$ as the (symmetric)
encryption of its own domain and the $\mi{rpNonce}$ with
$\mi{tagKey}$. It then returns to the browser the
$\mi{loginSessionToken}$, the $\mi{tagKey}$, and the domain of the
forwarder ($S(r).\str{FWDDomain}$).

When the RP document in the browser opens the login dialog, $r$
receives a third request, in this case a GET request for the path
$\str{/redir}$ with a parameter containing $\mi{loginSessionToken}$.
This is now used by $r$ to look up the user's session and redirect the
user to the IdP (this redirection serves mainly to hide the referer
string from the request to IdP). For this, $r$ sends the script
$\mi{script\_rp\_redir}$ and, in its initial script state, defines
that the script should redirect the user to the URL of the login
dialog (which is ``https://'' plus the domain of the user's email
address plus ``/.well-known/spresso-login'').

Finally, $r$ receives a last request in the login flow. This POST
request contains the encrypted IA and the $\mi{loginSessionToken}$. .
To conclude the login, $r$ looks up the user's login session, decrypts
the IA, and checks that it is a signature over the tag, the user's
email address, and the FWD domain. If successful, $r$ returns a new
service token, which is also stored in the state of $r$.

If $r$ receives a corrupt message, it becomes corrupt and acts like
the attacker from then on.

We now provide the formal definition of $r$ as an atomic DY process
$(I^r, Z^r, R^r, s^r_0)$. As mentioned, we define $I^r =
\mapAddresstoAP(r)$. Next, we define the set $Z^r$ of states of
$r$ and the initial state $s^r_0$ of $r$.

\begin{definition}
  A \emph{login session record} is a term of the form $\an{\mi{email},
    \mi{rpNonce}, \mi{iaKey}, \mi{tag}}$ with $\mi{email} \in \IDs$
  and $\mi{rpNonce}$, $\mi{iaKey}$, $\mi{tag} \in \nonces$.
\end{definition}

\begin{sloppypar}
  \begin{definition}\label{def:relying-parties}
    A \emph{state $s\in Z^r$ of an RP $r$} is a term of the form
    $\langle\mi{DNSAddress}$, $\mi{FWDDomain}$, $\mi{keyMapping}$,
    $\mi{sslkeys}$, $\mi{pendingDNS}$, $\mi{pendingRequests}$,
    $\mi{loginSessions}$, $\mi{serviceTokens}$, $\mi{wkCache}$,
    $\mi{corrupt}\rangle$ where $\mi{DNSAddress} \in \addresses$,
    $\mi{FWDDomain} \in \dns$,
    $\mi{keyMapping} \in \dict{\mathbb{S}}{\nonces}$,
    $\mi{sslkeys}=\mi{sslkeys}^r$,
    $\mi{pendingDNS} \in \dict{\nonces}{\terms}$,
    $\mi{pendingRequests} \in \dict{\nonces}{\terms}$,
    $\mi{serviceTokens}\in\dict{\nonces}{\mathbb{S}}$,
    $\mi{loginSessions} \in \dict{\nonces}{\terms}$ is a dictionary of
    login session records,
    $\mi{wkCache} \in \dict{\mathbb{S}}{\terms}$,
    $\mi{corrupt}\in\terms$.

    The \emph{initial state $s^r_0$ of $r$} is a state of $r$ with
    $s^r_0.\str{serviceTokens} = s^r_0.\str{loginSessions} =
    s^r_0.\str{wkCache} = \an{}$,
    $s^r_0.\str{corrupt} = \bot$, and $s^r_0.\str{keyMapping}$ is the
    same as the keymapping for browsers above.
  \end{definition}
\end{sloppypar}

We now specify the relation $R^r$. Just like
in Appendix~\ref{app:deta-descr-brows}, we describe this
relation by a non-deterministic algorithm. 

\captionof{algorithm}{\label{alg:sendstartloginresponse} Sending the response to a startLogin XMLHTTPRequest}
\begin{algorithmic}[1]
  \Function{$\mathsf{SENDSTARTLOGINRESPONSE}$}{$a$, $f$, $k$, $n$, $\mi{email}$, $\mi{inDomain}$, $s'$}
  \Let{$\mi{rpNonce}$}{$\nu_1$}
  \Let{$\mi{tagKey}$}{$\nu_2$}
  \Let{$\mi{iaKey}$}{$\nu_3$}
  \Let{$\mi{loginSessionToken}$}{$\nu_4$}
  \Let{$\mi{tag}$}{$\encs{\an{\mi{inDomain}, \mi{rpNonce}}}{\mi{tagKey}}$}
  \Let{$s'.\str{loginSessions}[\mi{loginSessionToken}]$}{$\an{\mi{email}, \mi{rpNonce}, \mi{iaKey}, \mi{tag}}$}
  \Let{$\mi{body}$}{$\an{\an{\str{tagKey}, \mi{tagKey}} \an{\str{loginSessionToken}, \mi{loginSessionToken}}, \an{\str{FWDDomain}, s'.\str{FWDDomain}}}$}
  \Let{$m'$}{$\encs{\an{\cHttpResp, n, 200, \an{}, \mi{body}}}{k}$}
  \Stop{\StopWithMPrime}
  \EndFunction
\end{algorithmic} \setlength{\parindent}{1em}

\captionof{algorithm}{\label{alg:rp-spresso} Relation of a Relying
  Party $R^r$}
\begin{algorithmic}[1]
\Statex[-1] \textbf{Input:} $\an{a,f,m},s$
  \If{$s'.\str{corrupt} \not\equiv \bot \vee m \equiv \corrupt$}
    \Let{$s'.\str{corrupt}$}{$\an{\an{a, f, m}, s'.\str{corrupt}}$}
    \LetND{$m'$}{$d_{V}(s')$}\label{line:usage-of-signkey-corrupt-spresso} %
    \LetND{$a'$}{$\addresses$} %
    \Stop{$\an{\an{a',a,m'}}$, $s'$}
  \EndIf
  \If{$\exists\, \an{\mi{reference}, \mi{request}, \mi{key}, f}$
      $\inPairing \comp{s'}{pendingRequests}$ \breakalgohook{0}
      \textbf{such that} $\proj{1}{\decs{m}{\mi{key}}} \equiv \cHttpResp$ } \label{line:rp-https-response} %
    \Comment{Encrypted HTTP response}
    \Let{$m'$}{$\decs{m}{\mi{key}}$}
    \If{$\comp{m'}{nonce} \not\equiv \comp{\mi{request}}{nonce}$}
      \Stop{\DefStop}
    \EndIf
    \State \textbf{remove} $\an{\mi{reference}, \mi{request}, \mi{key}, f}$ \textbf{from} $\comp{s'}{pendingRequests}$
    \LetST{$a'$, $f'$, $k$, $n$, $\mi{email}$, $\mi{inDomain}$}{$\an{a', f', k, n, \mi{email}, \mi{inDomain}} \equiv \mi{reference}$}{\textbf{stop} \DefStop} %
    \Let{$s'.\str{wkCache}[\mi{request}.\str{host}]$}{$m'.\str{body}$} \label{line:rp-populate-wkcache}\label{line:rp-accept}
    \State $\mathsf{SENDSTARTLOGINRESPONSE}$($a'$, $f'$, $k$, $n$, $\mi{email}$, $\mi{inDomain}$, $s'$) \label{line:call-ssl-second}
  \ElsIf{$m \in \dnsresponses$} \Comment{Successful DNS response}
      \If{$\comp{m}{nonce} \not\in \comp{s}{pendingDNS} \vee \comp{m}{result} \not\in \addresses \vee \comp{m}{domain} \not\equiv \comp{\proj{2}{\comp{s}{pendingDNS}}}{host}$}
      \Stop{\DefStop}
      \EndIf
      \Let{$\an{\mi{reference}, \mi{message}}$}{$\comp{s}{pendingDNS}[\comp{m}{nonce}]$}
      \AppendBreak{2}{$\langle\mi{reference}$, $\mi{message}$, $\nu_5$, $\comp{m}{result}\rangle$}{$\comp{s'}{pendingRequests}$} \label{line:move-reference-to-pending-request}
      \Let{$\mi{message}$}{$\enc{\an{\mi{message},\nu_5}}{\comp{s'}{keyMapping}\left[\comp{\mi{message}}{host}\right]}$} \label{line:select-enc-key-spresso}
      \Let{$\comp{s'}{pendingDNS}$}{$\comp{s'}{pendingDNS} - \comp{m}{nonce}$}
      \Stop{$\an{\an{\comp{m}{result}, a, \mi{message}}}$, $s'$} 
   \Else \Comment{Handle HTTP requests}
     \LetST{$m_{\text{dec}}$, $k$, $k'$, $\mi{inDomain}$}{\breakalgohook{0}$\an{m_{\text{dec}}, k} \equiv \dec{m}{k'} \wedge \an{inDomain,k'} \in s.\str{sslkeys}$\breakalgohook{0}}{\textbf{stop} \DefStop} \label{line:receive-https-request} %
     \LetST{$n$, $\mi{method}$, $\mi{path}$, $\mi{parameters}$, $\mi{headers}$, $\mi{body}$}{\breakalgohook{0}$\an{\cHttpReq, n, \mi{method}, \mi{inDomain}, \mi{path}, \mi{parameters}, \mi{headers}, \mi{body}} \equiv m_{\text{dec}}$\breakalgohook{0}}{\textbf{stop} \DefStop} %
     \If{$\mi{path} \equiv \str{/}$} \label{line:serve-rp-index}\Comment{Serve index page.}
       \Let{$m'$}{$\encs{\an{\cHttpResp, n, 200, \an{}, \an{\str{script\_rp}, \mi{initState_{rp}}}}}{k}$} \Comment{Initial state defined for $\mi{script\_rp}$ (below).} 
       \Stop{$\an{\an{f,a,m'}}$, $s'$}
     \ElsIf{$\mi{path} \equiv \str{/startLogin} \wedge \mi{method} \equiv \mPost$} \label{line:serve-rp-start-login}\Comment{Serve start login request.}
       \If{$\mi{body} \not\in \mi{ids}$}
         \Stop{\DefStop}
       \EndIf
       \Let{$\mi{domain}$}{$\mi{body}.\str{domain}$}
       \If{$\mi{domain} \in s.\str{wkCache}$} \label{line:branch-wk-cache}
         \State $\mathsf{SENDSTARTLOGINRESPONSE}$($a$, $f$, $k$, $n$, $\mi{body}$, $\mi{inDomain}$, $s'$)
       \Else 
         \Let{$\mi{message}$}{$\hreq{ nonce=\nu_6, method=\mGet,
             xhost=\mi{domain}, path=\str{/.well\mhyphen{}known/spresso\mhyphen{}info}, parameters=\an{}, headers=\an{}, xbody=\an{}}$} \label{line:rp-send-request}
         \Let{$\comp{s'}{pendingDNS}[\nu_6]$}{$\an{\an{a, f, k, n, \mi{body}, \mi{inDomain}}, \mi{message}}$}
         \Stop{$\an{\an{\comp{s'}{DNSaddress}, a, \an{\cDNSresolve, \mi{domain}, \nu_6}}}$, $s'$}
       \EndIf
     \ElsIf{$\mi{path} \equiv \str{/redir} \wedge \mi{method} \equiv \mGet$} \label{line:serve-rp-redir} \Comment{Serve redirection script.}
       \Let{$\mi{loginSession}$}{$s'.\str{loginSessions}[\mi{body}[\str{loginSessionToken}]]$}\label{line:rp-redir-start}
       \If{$\mi{loginSession} \equiv \an{}$}
         \Stop{\DefStop} \label{line:rp-redir-stop}
       \EndIf
       \Let{$\mi{domain}$}{$\mi{loginSession}.\str{email}.\str{domain}$}
       \Let{$\mi{params}$}{$\an{\an{\str{email}, \mi{loginSession}.\str{email}}, \an{\str{tag}, \mi{loginSession}.\str{tag}},$\breakalgohook{5}$ \an{\str{iaKey}, \mi{loginSession}.\str{iaKey}}, \an{\str{FWDDomain}, s'.\str{FWDDomain}}}$}
       \Let{$\mi{url}$}{$\an{\cUrl, \https, \mi{domain}, \str{/.well\mhyphen{}known/spresso\mhyphen{}login}, \mi{params}}$}
       \Let{$m'$}{$\encs{\an{\cHttpResp, n, 200, \an{}, \an{\str{script\_rp\_redir}, \mi{url}}}}{k}$} 
       \Stop{$\an{\an{f,a,m'}}$, $s'$} \label{line:rp-redir-response}
     \ElsIf{$\mi{path} \equiv \str{/login} \wedge \mi{method} \equiv \mPost$} \label{line:assemble-login-response} \label{line:check-path-method-rp} \Comment{Serve login request.}
       \If{$\mi{headers}[\str{Origin}] \not\equiv \an{\mi{inDomain}, \https} \vee \mi{body}[\str{loginSessionToken}] \equiv \an{}$}\label{line:check-origin-header}
         \Stop{\DefStop}
       \EndIf

       \Let{$\mi{loginSession}$}{$s'.\str{loginSessions}[\mi{body}[\str{loginSessionToken}]]$}
       \If{$\mi{loginSession} \equiv \an{}$}
         \Stop{\DefStop}
       \EndIf
       \Let{$s'.\str{loginSessions}$}{$s'.\str{loginSessions} - \mi{body}[\str{loginSessionToken}]$}
       \Let{$\mi{ia}$}{$\decs{\mi{body}[\str{eia}]}{\mi{loginSession}.\str{iaKey}}$} \label{line:check-service-token-request-contents}
       \Let{$e$}{$\an{\mi{loginSession}.\str{tag}, \mi{loginSession}.\str{email}, s'.\str{FWDDomain}}$}
       \If{$\checksigThree{e}{ia}{s'.\str{wkCache}[\mi{loginSession}.\str{email}.\str{domain}][\str{signkey}]} \equiv \bot$} \label{line:rp-check-ia}
         \Stop{\DefStop}
       \EndIf
       \Let{$\mi{serviceTokenNonce}$}{$\nu_7$} \label{line:choose-service-token}
       \Let{$\mi{serviceToken}$}{$\an{\mi{serviceTokenNonce}, \mi{loginSession}.\str{email}}$} \label{line:assemble-service-token}
       \Append{$\mi{serviceToken}$}{$s'.\str{serviceTokens}$} \label{line:store-service-token}
       \Let{$m'$}{$\encs{\an{\cHttpResp, n, 200, \an{}, \mi{serviceToken}}}{k}$} \label{line:send-service-token}
       \Stop{$\an{\an{f,a,m'}}$, $s'$}
     \EndIf
  \EndIf
  \Stop{\DefStop}
  
\end{algorithmic} \setlength{\parindent}{1em}

\subsection{Identity Providers}  \label{app:idps}

An identity provider $i \in \mathsf{IdPs}$ is a web server modeled as
an atomic process $(I^i, Z^i, R^i, s_0^i)$ with the addresses $I^i
:= \mapAddresstoAP(i)$. Its initial state $s^i_0$ contains
a list of its domains and (private) SSL keys, a
list of users and identites, and a private key for signing
UCs. Besides this, the full state of $i$ further contains a list of
used nonces, and information about active sessions.

IdPs react to three types of requests:

First, they provide the ``well-known document'', a machine-readable
document which contains the IdP's verification key. This document is served upon
a GET request to the path
$\str{/.well\mhyphen{}known/spresso\mhyphen{}info}$.

Second, upon a request to the LD path (i.e., $\str{/.well\mhyphen{}known/spresso\mhyphen{}login}$), an IdP serves the login dialog
script, i.e., $\str{script\_idp}$. Into the initial state of this
script, IdPs encode whether the browser is already logged in or not.
Further, IdP issues an XSRF token to the browser (in the same way RPs
do).

The login dialog will eventually send an XMLHTTPRequest to the path
$\str{loginxhr}$, where it retrieves the IA. This is also the last
type of requests IdPs answer to. Before serving the response to this
request, IdP checks whether the user is properly authenticated. It
then creates the IA and sends it to the browser.

\subsubsection{Formal description.} In the following, we will first
define the (initial) state of $i$ formally and afterwards present the
definition of the relation $R^i$.

To define the initial state, we will need a term that represents the
``user database'' of the IdP $i$. We will call this term
$\mi{userset}^i$. This database defines, which secret is valid for
which identity. It is encoded as a mapping of identities to secrets.
For example, if the secret $\mi{secret}_1$ is valid for the identites
$\mi{id}_1$and the secret $\mi{secret}_2$ is valid for the identity
$\mi{id}_2$, the $\mi{userset}^i$ looks as follows:
\begin{align*}
\mi{userset}^i = [\mi{id}_1{:}\mi{secret}_1, \mi{id}_2{:}\mi{secret}_2]
\end{align*}

We define $\mi{userset}^i$ as $\mi{userset}^i = \an{\{\an{u,
    \mapIDtoPLI(u)}\, |\, u \in \IDs^i\}}$.

\begin{definition}\label{def:initial-state-idp}
  A \emph{state $s\in Z^i$ of an IdP $i$} is a term of the form
  $\langle\mi{sslkeys}$, $\mi{users}$, $\mi{signkey}$,
  $\mi{sessions}$, $\mi{corrupt}\rangle$ where $\mi{sslkeys} =
  \mi{sslkeys}^i $, $\mi{users} = \mi{userset}^i$, $\mi{signkey} \in
  \nonces$ (the key used by the IdP $i$ to sign UCs),
  $\mi{sessions}\in\dict{\nonces}{\terms}$, $\mi{corrupt} \in \terms$.

  An \emph{initial state $s^i_0$ of $i$} is a state of the form $\an{
    \mi{sslkeys}^i, \mi{userset}^i, \mapSignKey(i), \an{},
    \bot}$.
\end{definition}

The relation $R^i$ that defines the behavior of the IdP $i$ is defined as follows:

\captionof{algorithm}{\label{alg:idp-spresso} Relation of IdP $R^i$}
\begin{algorithmic}[1]
\Statex[-1] \textbf{Input:} $\an{a,f,m},s$
  \Let{$s'$}{$s$}
  \If{$s'.\str{corrupt} \not\equiv \bot \vee m \equiv \corrupt$}
    \Let{$s'.\str{corrupt}$}{$\an{\an{a, f, m}, s'.\str{corrupt}}$}
    \LetND{$m'$}{$d_{V}(s')$}\label{line:usage-of-signkey-corrupt-spresso}
    \LetND{$a'$}{$\addresses$}
    \State \textbf{stop} $\an{\an{a',a,m'}}$, $s'$
  \EndIf
  \LetST{$m_{\text{dec}}$, $k$, $k'$, $\mi{inDomain}$}{\breakalgohook{0}$\an{m_{\text{dec}}, k} \equiv \dec{m}{k'} \wedge \an{inDomain,k'} \in s.\str{sslkeys}$\breakalgohook{0}}{\textbf{stop} \DefStop} %
  \LetST{$n$, $\mi{method}$, $\mi{path}$, $\mi{parameters}$, $\mi{headers}$, $\mi{body}$}{\breakalgohook{0}$\an{\cHttpReq, n, \mi{method}, \mi{inDomain}, \mi{path}, \mi{parameters}, \mi{headers}, \mi{body}} \equiv m_{\text{dec}}$\breakalgohook{0}}{\textbf{stop} \DefStop} %
  \If{$\mi{path} \equiv \str{/.well\mhyphen{}known/spresso\mhyphen{}info}$} \label{line:idp-response-wk} \Comment{Serve support document.}
    \Let{$\mi{wkDoc}$}{$\an{\an{\str{signkey},\pub(s'.\str{signkey})}}$}
    \Let{$m'$}{$\encs{\an{\cHttpResp, n, 200, \an{}, \mi{wkDoc}}}{k}$} 
    \Stop{\StopWithMPrime}
  \ElsIf{$\mi{path} \equiv \str{/.well\mhyphen{}known/spresso\mhyphen{}login}$} \Comment{Serve login dialog.}
    \Let{$\mi{sessionid}$}{$\mi{headers}[\str{Cookie}][\str{sessionid}]$}
    \Let{$\mi{email}$}{$s'.\str{sessions}[\mi{sessionid}]$}
    \Let{$m'$}{$\encs{\an{\cHttpResp, n, 200, \an{}, \an{\str{script\_idp}, \an{\str{start}, \mi{email}, \an{}}}}}{k}$} \Comment{Initial scriptstate of $\mi{script\_idp}$ (defined below).}
    \Stop{\StopWithMPrime}
  \ElsIf{$\mi{path} \equiv \str{/sign} \wedge \mi{method} \equiv \mPost$} \Comment{Serve signing request.}
    \Let{$\mi{sessionid}$}{$\mi{headers}[\str{Cookie}][\str{sessionid}]$}
    \Let{$\mi{loggedInAs}$}{$s'.\str{sessions}[\mi{sessionid}]$}
    \If{$\mi{body}[\str{email}] \not\equiv \mi{loggedInAs} \wedge \mi{body}[\str{password}] \not\equiv s'.\str{userset}[\mi{body}[\str{email}]]$} \label{line:spresso-idp-check-login-state}
      \Stop{\DefStop}
    \EndIf
    \Let{$\mi{ia}$}{$\sig{\an{\mi{body}[\str{tag}],\mi{body}[\str{email}],\mi{body}[\str{FWDDomain}]}}{s'.\str{signkey}}$} \label{line:sign-ia}
    \Let{$\mi{sessionid}$}{$\nu_8$}
    \Let{$s'.\str{sessions}[\mi{sessionid}]$}{$\mi{body}[\str{email}]$}
    \Let{$\mi{setCookie}$}{$\an{\cSetCookie, \an{\an{\str{sessionid}, \mi{sessionid}, \True, \True, \True}}}$}
    \Let{$m'$}{$\encs{\an{\cHttpResp, n, 200, \an{\mi{setCookie}}, \mi{ia}}}{k}$} 
    \Stop{\StopWithMPrime}
  \EndIf
  \Stop{\DefStop}
\end{algorithmic} \setlength{\parindent}{1em}

\subsection{Forwarders} \label{app:fwd-spresso} 

We define FWDs formally as atomic DY processes
$\mathit{fwd} = (I^\mathit{fwd}, Z^\mathit{fwd}, R^\mathit{fwd},
s^\mathit{fwd}_0)$.
As already mentioned, we define
$I^\mathit{fwd} = \mapAddresstoAP(\mathit{fwd})$ with the set of
states $Z^\mathit{fwd}$ being all terms of the form
$\an{\mi{sslkeys}, \mi{corrupt}}$ for $\mi{sslkeys}$,
$\mi{corrupt} \in \terms$. The initial state $s^\mathit{fwd}_0$ of an
FWD contains the private key of its domain and the corruption state:
$s^\mathit{fwd}_0 = \an{\mi{sslkeys}^\mathit{fwd}, \bot}$.

An FWD responds to any HTTPS request with $\str{script\_fwd}$ and its
initial state, which is empty. 

We now specify the relation $R^\mathit{fwd}$ of FWDs. We describe this
relation by a non-deterministic algorithm.

\captionof{algorithm}{\label{alg:spresso-fwd} Relation of an FWD
  $R^\mathit{fwd}$ }
\begin{algorithmic}[1]
\Statex[-1] \textbf{Input:} $\an{a,f,m},s$
  \If{$s.\str{corrupt} \not\equiv \bot \vee m \equiv \corrupt$}
    \Let{$s'.\str{corrupt}$}{$\an{\an{a, f, m}, s.\str{corrupt}}$}
    \LetND{$m'$}{$d_{V}(s')$}
    \LetND{$a'$}{$\addresses$}
    \State \textbf{stop} $\an{\an{a',a,m'}}$, $s'$
  \EndIf
  \LetST{$m_{\text{dec}}$, $k$, $k'$, $\mi{inDomain}$}{\breakalgohook{0}$\an{m_{\text{dec}}, k} \equiv \dec{m}{k'} \wedge \an{inDomain,k'} \in s$\breakalgohook{0}}{\textbf{stop} \DefStop} %
  \LetST{$n$, $\mi{method}$, $\mi{path}$, $\mi{parameters}$, $\mi{headers}$, $\mi{body}$}{\breakalgohook{0}$\an{\cHttpReq, n, \mi{method}, \mi{inDomain}, \mi{path}, \mi{parameters}, \mi{headers}, \mi{body}} \equiv m_{\text{dec}}$\breakalgohook{0}}{\textbf{stop} \DefStop} %
  \Let{$m'$}{$\encs{\an{\cHttpResp, n, 200, \an{}, \an{\str{script\_fwd, \an{}}}}}{k}$} \label{line:fwd-send-response}
  \Stop{$\an{\an{f,a,m'}},s$}
\end{algorithmic} \setlength{\parindent}{1em}

\subsection{DNS Servers}

As already outlined above, DNS servers are modeled as generic DNS
servers presented in Appendix~\ref{app:DNSservers}. Their (static)
state is set according to the allocation of domain names to IP
addresses\gs{das hier ist falsch und nicht unbedingt noetig
  hinzuschreiben: , i.e., their state is $\an{\dns}$}.  DNS servers
may not become corrupted.

\subsection{\spresso Scripts} \label{app:spresso-scripts} As
already mentioned in Appendix~\ref{app:outlinespressomodel}, the set
$\scriptset$ of the web system
$\spressowebsystem=(\bidsystem, \scriptset,
\mathsf{script}, E^0)$ consists of the scripts $\Rasp$,
$\mi{script\_rp}$, $\mi{script\_idp}$, and
$\mi{script\_fwd}$, with their string representations being
$\str{att\_script}$, $\str{script\_rp}$,
$\str{script\_idp}$, and $\str{script\_fwd}$ (defined by
$\mathsf{script}$). 

In what follows, the scripts $\mi{script\_rp}$,
$\mi{script\_idp}$, and $\mi{script\_fwd}$ are
defined formally. First, we introduce some notation and
helper functions.

\subsubsection{Notations and Helper Functions.}
In the formal description of the scripts we use an abbreviation for
URLs. We write $\mathsf{URL}^d_\mi{path}$ to
describe the following URL term: $\an{\tUrl, \https, d,
  \mi{path}, \an{}}$.  

In order to simplify the description of the scripts, several helper
functions are used.

\paragraph{CHOOSEINPUT.}

The state of a document contains a term $\mi{scriptinputs}$ which
records the input this document has obtained so far (via \xhrs and
\pms, append-only). If the script of the document is activated, it
will typically need to pick one input message from the sequence
$\mi{scriptinputs}$ and record which input it has already processed.
For this purpose, the function
$\mathsf{CHOOSEINPUT}(\mi{scriptinputs},\mi{pattern})$ is used. If called, it chooses
the first message in $\mi{scriptinputs}$ that matches $\mi{pattern}$
and returns it.

\captionof{algorithm}{\label{alg:chooseinput} Choose an unhandled input message for a script}
\begin{algorithmic}[1]
  \Function{$\mathsf{CHOOSEINPUT}$}{$\mi{scriptinputs},\mi{pattern}$}
  \LetST{$i$}{$i = \min\{ j : \proj{j}{\mi{scriptinputs}} \sim \mi{pattern}\}$}{\Return$\bot$}
  \State \Return$\proj{i}{\mi{scriptinputs}}$
  \EndFunction
\end{algorithmic} \setlength{\parindent}{1em}

\paragraph{PARENTWINDOW.} To determine the nonce referencing the
active document in the parent window in the browser, the function
$\mathsf{PARENTWINDOW}(\mi{tree}, \mi{docnonce})$ is used. It takes
the term $\mi{tree}$, which is the (partly cleaned) tree of browser
windows the script is able to see and the document nonce
$\mi{docnonce}$, which is the nonce referencing the current document
the script is running in, as input. It outputs the nonce referencing
the active document in the window which directly contains in its
subwindows the window of the document referenced by $\mi{docnonce}$.
If there is no such window (which is the case if the script runs in a
document of a top-level window) or no active document,
$\mathsf{PARENTWINDOW}$ returns $\mi{docnonce}$.

\paragraph{SUBWINDOWS.} This function takes a term
$\mi{tree}$ and a document nonce $\mi{docnonce}$ as input
just as the function above. If $\mi{docnonce}$ is not a
reference to a document contained in $\mi{tree}$, then
$\mathsf{SUBWINDOWS}(\mi{tree},\mi{docnonce})$ returns
$\an{}$. Otherwise, let $\an{\mi{docnonce}$, $\mi{origin}$,
  $\mi{script}$, $\mi{scriptstate}$, $\mi{scriptinputs}$,
  $\mi{subwindows}$, $\mi{active}}$ denote the subterm of
$\mi{tree}$ corresponding to the document referred to by
$\mi{docnonce}$. Then,
$\mathsf{SUBWINDOWS}(\mi{tree},\mi{docnonce})$ returns
$\mi{subwindows}$.

\paragraph{AUXWINDOW.} This function takes a term $\mi{tree}$ and a
document nonce $\mi{docnonce}$ as input as above. From all window
terms in $\mi{tree}$ that have the window containing the document
identified by $\mi{docnonce}$ as their opener, it selects one
non-deterministically and returns its active document's nonce. If
there is no such window or no active document, it returns
$\mi{docnonce}$.

\paragraph{OPENERWINDOW.} This function takes a
term $\mi{tree}$ and a document nonce $\mi{docnonce}$ as
input as above. It returns the window nonce of the opener
window of the window that contains the document identified
by $\mi{docnonce}$. Recall that the nonce identifying the
opener of each window is stored inside the window term. If
no document with nonce $\mi{docnonce}$ is found in the tree
$\mi{tree}$, $\notdef$ is returned.

\paragraph{GETWINDOW.} This function takes a term
$\mi{tree}$ and a document nonce $\mi{docnonce}$ as input
as above. It returns the nonce of the window containing $\mi{docnonce}$.

\paragraph{GETORIGIN.} To extract the origin of a
document, the function
$\mathsf{GETORIGIN}(\mi{tree},\mi{docnonce})$ is used. This
function searches for the document with the identifier
$\mi{docnonce}$ in the (cleaned) tree $\mi{tree}$ of the
browser's windows and documents. It returns the origin $o$
of the document. If no document with nonce $\mi{docnonce}$
is found in the tree $\mi{tree}$, $\notdef$ is returned.

\paragraph{GETPARAMETERS.} Works exactly as GETORIGIN, but returns the
document's parameters instead.

\subsubsection{Relying Party Index Page (script\_rp).}\label{app:spresso-script-rp}
As defined in Appendix~\ref{app:websystem}, a script is a relation that
takes as input a term and outputs a new term. As specified in
Appendix~\ref{app:deta-descr-brows} (Triggering the Script of a Document
(\textbf{\hlExp{$m = \trigger$}, \hlExp{$\mi{action} = 1$}})) and
formally specified in Algorithm~\ref{alg:runscript}, the input term is
provided by the browser. It contains the current internal state of the
script (which we call \emph{scriptstate} in what follows) and
additional information containing all browser state information the
script has access to, such as the input the script has obtained so far
via \xhrs and \pms, information about windows, etc. The browser
expects the output term to have a specific form, as also specified in
Appendix~\ref{app:deta-descr-brows} and Algorithm~\ref{alg:runscript}. The
output term contains, among other information, the new internal
scriptstate.

We first describe the structure of the internal scriptstate
of the script $\mi{script\_rp}$.

\begin{definition} \label{def:scriptstaterp} A \emph{scriptstate $s$
    of $\mi{script\_rp}$} is a term of the form $\langle q$,
  $\mi{loginSessionToken}$, $\mi{refXHR}$, 
  $\mi{tagKey}$, $\mi{FWDDomain} \rangle$ where $q \in \mathbb{S}$, $\mi{loginSessionToken}$,
  $\mi{refXHR}$, $\mi{tagKey} \in \nonces \cup
  \{\bot\}$, 
  $\mi{FWDDomain} \in \terms$. 

  The
  \emph{initial scriptstate $\mi{initState_{rp}}$} of
  $\mi{script\_rp}$ is
  $\an{\str{start},\bot,\bot,\bot,\bot}$.
\end{definition}

We now specify the relation $\mi{script\_rp}$ formally. We describe this relation
by a non-deterministic algorithm.

Just like all scripts, as explained in
Appendix~\ref{app:deta-descr-brows} (see also
Algorithm~\ref{alg:runscript} for the formal
specification), the input term this script obtains from the
browser contains the cleaned tree of the browser's windows
and documents $\mi{tree}$, the nonce of the current
document $\mi{docnonce}$, its own scriptstate
$\mi{scriptstate}$ (as defined in
Definition~\ref{def:scriptstaterp}), a sequence of all
inputs $\mi{scriptinputs}$ (also containing already handled
inputs), a dictionary $\mi{cookies}$ of all accessible
cookies of the document's domain, the \ls
$\mi{localStorage}$ belonging to the document's origin, the
secrets $\mi{secret}$ of the document's origin, and a set
$\mi{nonces}$ of fresh nonces as input. The script returns
a new scriptstate $s'$, a new set of cookies
$\mi{cookies'}$, a new \ls $\mi{localStorage'}$, and a term
$\mi{command}$ denoting a command to the browser.
\captionof{algorithm}{\label{alg:spresso-script-rp} Relation of $\mi{script\_rp}$ }
\begin{algorithmic}[1]
\Statex[-1] \textbf{Input:} $\langle\mi{tree}$, $\mi{docnonce}$, $\mi{scriptstate}$, $\mi{scriptinputs}$, $\mi{cookies}$, $\mi{localStorage}$, $\mi{sessionStorage}$,\breakalgohook{-1}$\mi{ids}$, $\mi{secret}\rangle$
\Let{$s'$}{$\mi{scriptstate}$}
\Let{$\mi{command}$}{$\an{}$}
\Let{$\mi{origin}$}{$\mathsf{GETORIGIN}(\mi{tree},\mi{docnonce})$} \label{line:determine-origin-rp-script}
\Switch{$s'.\str{q}$}
 \Case{$\str{start}$}
  \LetND{$s'.\str{email}$}{$\mi{ids}$} \label{line:select-email-address}
  \Let{$s'.\str{refXHR}$}{$\lambda_1$}
  \Let{$\mi{command}$}{$\an{\tXMLHTTPRequest,\textsf{URL}^{\mi{origin}.\str{domain}}_\str{/startLogin},\mPost,s'.\str{email},s'.\str{refXHR}}$} \label{line:send-start-login}
  \Let{$s'.\str{q}$}{$\str{expectStartLoginResponse}$}
 \EndCase

 \Case{$\str{expectStartLoginResponse}$}
  \Let{$\mi{pattern}$}{$\an{\tXMLHTTPRequest,*,s'.\str{refXHR}}$}
  \Let{$\mi{input}$}{\textsf{CHOOSEINPUT}($\mi{scriptinputs},\mi{pattern}$)}
  \If{$\mi{input} \not\equiv \bot$}
   \Let{$s'.\str{loginSessionToken}$}{$\proj{2}{\mi{input}}[\str{loginSessionToken}]$}  \label{line:set-login-session-token}
   \Let{$s'.\str{tagKey}$}{$\proj{2}{\mi{input}}[\str{tagKey}]$}
   \Let{$s'.\str{FWDDomain}$}{$\proj{2}{\mi{input}}[\str{FWDDomain}]$}
   \Let{$\mi{command}$}{\breakalgohook{3}$\an{\tHref,\an{\tUrl, \https, \mi{origin}.\str{domain},
  \str{/redir}, \an{\an{\str{loginSessionToken}, s'.\str{loginSessionToken}}}},\wBlank,\an{}}$} 
   \Let{$s'.\str{q}$}{$\str{expectFWDReady}$}
  \EndIf
 \EndCase

 \Case{$\str{expectFWDReady}$}
  \Let{$\mi{fwdWindowNonce}$}{\textsf{SUBWINDOWS}($\mi{tree}$, \textsf{AUXWINDOW}($\mi{tree}$, $\mi{docnonce}$))$.1.\str{nonce}$}
  \Let{$\mi{pattern}$}{$\an{\tPostMessage,\mi{fwdWindowNonce},\an{s'.\str{FWDDomain}, \https},\str{ready}}$}
  \Let{$\mi{input}$}{\textsf{CHOOSEINPUT}($\mi{scriptinputs},\mi{pattern}$)}
  \If{$\mi{input} \not\equiv \bot$}
   \Let{$\mi{command}$}{$\langle\tPostMessage$, $\mi{fwdWindowNonce}$, $\an{\str{tagKey},\mi{tagKey}}$, $\an{s'.\str{FWDDomain}, \https}\rangle$}
  \Let{$s'.\str{q}$}{$\str{expectEIA}$}
  \EndIf
 \EndCase

 \Case{$\str{expectEIA}$}
  \Let{$\mi{fwdWindowNonce}$}{\textsf{SUBWINDOWS}($\mi{tree}$, \textsf{AUXWINDOW}($\mi{tree}$, $\mi{docnonce}$))$.1.\str{nonce}$}
  \Let{$\mi{pattern}$}{$\an{\tPostMessage,\mi{fwdWindowNonce},\an{s'.\str{FWDDomain}, \https},\an{\str{eia},*}}$}
  \Let{$\mi{input}$}{\textsf{CHOOSEINPUT}($\mi{scriptinputs},\mi{pattern}$)}
  \If{$\mi{input} \not\equiv \bot$}
   \Let{$\mi{eia}$}{$\proj{2}{\proj{4}{\mi{input}}}$}
   \Let{$s'.\str{refXHR}$}{$\lambda_1$}
    \Let{$\mi{body}$}{$\an{\an{\str{eia}, \mi{eia}}, \an{\str{loginSessionToken}, s'.\str{loginSessionToken}}}$} \label{line:spresso-rp-script-assemble-login-body}
    \Let{$\mi{command}$}{$\an{\tXMLHTTPRequest,\textsf{URL}^{\mi{origin}.\str{domain}}_\str{/login},\mPost,\mi{body},s'.\str{refXHR}}$} \label{line:send-ia-to-r}
    \Let{$s'.\str{q}$}{$\str{expectServiceToken}$}
  \EndIf
 \EndCase
\EndSwitch
\State \textbf{stop} $\an{s',\mi{cookies},\mi{localStorage},\mi{sessionStorage},\mi{command}}$

\end{algorithmic} \setlength{\parindent}{1em}

\subsubsection{Relying Party Redirection Page (script\_rp\_redir).}\label{app:spresso-script-rp-redir}
This simple script (which is loaded from RP in a regular run) is used
to redirect the login dialog window to the actual login dialog
(IdPdoc) loaded from IdP. It expects the URL of the page to which the
browser should be redirected in its initial (and only) state.

\captionof{algorithm}{\label{alg:spresso-script-rp-redir} Relation of $\mi{script\_rp\_redir}$ }
\begin{algorithmic}[1]
\Statex[-1] \textbf{Input:} $\langle\mi{tree}$, $\mi{docnonce}$, $\mi{scriptstate}$, $\mi{scriptinputs}$, $\mi{cookies}$, $\mi{localStorage}$, $\mi{sessionStorage}$,\breakalgohook{-1}$\mi{ids}$, $\mi{secret}\rangle$
\Let{$\mi{command}$}{$\an{\tHref,\mi{scriptstate},\bot,\True}$} \label{line:redir-to-idpdoc}
\State \textbf{stop} $\an{\mi{scriptstate},\mi{cookies},\mi{localStorage},\mi{sessionStorage},\mi{command}}$
\end{algorithmic} \setlength{\parindent}{1em}

\subsubsection{Login Dialog Script (script\_idp).}\label{app:spresso-script-idp}
This script models the contents of the login dialog.

\begin{definition}\label{def:scriptstateidp}
  A \emph{scriptstate $s$ of $\mi{script\_idp}$} is a term of the form
  $\langle q$, $\mi{email} \rangle$ with $q \in
  \mathbb{S}$, $\mi{email} \in \IDs \cup \{\an{}\} \in \gterms$. We call the
  scriptstate $s$ an \emph{initial scriptstate} of $\mi{script\_idp}$
  iff $s \sim \an{\str{start},*}$.
\end{definition}

We now formally specify the relation $\mi{script\_idp}$ of the
LD's scripting process. 

\captionof{algorithm}{\label{alg:spresso-script-idp} Relation of $\mi{script\_idp}$ }
\begin{algorithmic}[1]
\Statex[-1] \textbf{Input:} $\langle\mi{tree}$, $\mi{docnonce}$, $\mi{scriptstate}$, $\mi{scriptinputs}$, $\mi{cookies}$, $\mi{localStorage}$, $\mi{sessionStorage}$,\breakalgohook{-1}$\mi{ids}$, $\mi{secret}\rangle$
\Let{$s'$}{$\mi{scriptstate}$}
\Let{$\mi{command}$}{$\an{}$}
\Let{$\mi{origin}$}{$\mathsf{GETORIGIN}(\mi{tree},\mi{docnonce})$}

\Switch{$s'.\str{q}$}
 \Case{$\str{start}$}
  \Let{$\mi{email}$}{\textsf{GETPARAMETERS}$(\mi{tree}, \mi{docnonce})[\str{email}]$}
  \Let{$\mi{tag}$}{\textsf{GETPARAMETERS}$(\mi{tree}, \mi{docnonce})[\str{tag}]$}
  \Let{$\mi{FWDDomain}$}{\textsf{GETPARAMETERS}$(\mi{tree}, \mi{docnonce})[\str{FWDDomain}]$}
  \Let{$\mi{body}$}{$\langle\an{\str{email}, \mi{email}}, \an{\str{password}, \mi{secret}},  \an{\str{tag}, \mi{tag}}, \an{\str{FWDDomain}, \mi{FWDDomain}}\rangle$}
  \Let{$\mi{command}$}{$\an{\tXMLHTTPRequest,\textsf{URL}^{\mi{origin}.\str{domain}}_\str{/sign},\mPost,\mi{body}, \bot}$} \label{line:send-to-idp}
  \Let{$s'.\str{q}$}{$\str{expectIA}$}
 \EndCase

 \Case{$\str{expectIA}$}
  \Let{$\mi{pattern}$}{$\an{\tXMLHTTPRequest,*,*}$}
  \Let{$\mi{input}$}{\textsf{CHOOSEINPUT}($\mi{scriptinputs},\mi{pattern}$)}
  \If{$\mi{input} \not\equiv \bot$}
  \Let{$\mi{iaKey}$}{\textsf{GETPARAMETERS}$(\mi{tree}, \mi{docnonce})[\str{iaKey}]$}
  \Let{$\mi{FWDDomain}$}{\textsf{GETPARAMETERS}$(\mi{tree}, \mi{docnonce})[\str{FWDDomain}]$}
  \Let{$\mi{tag}$}{\textsf{GETPARAMETERS}$(\mi{tree}, \mi{docnonce})[\str{tag}]$}
   \Let{$\mi{eia}$}{$\encs{\proj{2}{\mi{input}}}{\mi{iaKey}}$}
   \Let{$\mi{url}$}{$\an{\tUrl, \https, \mi{FWDDomain}, \str{/}, \an{\an{\str{tag}, \mi{tag}}, \an{\str{eia}, \mi{eia}}}}$} \label{line:send-to-fwd}
   \Let{$\mi{command}$}{$\an{\tIframe,\mi{url},\wSelf}$}
   \Let{$s'.\str{q}$}{$\str{stop}$}
  \EndIf
 \EndCase
\EndSwitch
\State \textbf{stop} $\an{s',\mi{cookies},\mi{localStorage},\mi{sessionStorage},\mi{command}}$
\end{algorithmic} \setlength{\parindent}{1em}

\subsubsection{Forwarder Script (script\_fwd).}\label{app:spresso-script-fwd}

\begin{definition}\label{def:scriptstatetp}
  A \emph{scriptstate $s$ of $\mi{script\_fwd}$} is a term
  of the form $ q $ with $q \in \mathbb{S}$. 
  We call $s$ the \emph{initial scriptstate of
  $\mi{script\_fwd}$} iff $s \equiv
  \str{start}$.
\end{definition}

We now formally specify the relation $\mi{script\_rp\_index}$ of the
FWD's scripting process. 

\captionof{algorithm}{\label{alg:spresso-script-fwd} Relation of $\mi{script\_fwd}$}
\begin{algorithmic}[1]
\Statex[-1] \textbf{Input:} $\langle\mi{tree}$, $\mi{docnonce}$, $\mi{scriptstate}$, $\mi{scriptinputs}$, $\mi{cookies}$, $\mi{localStorage}$, $\mi{sessionStorage}$,\breakalgohook{-1}$\mi{ids}$, $\mi{secret}\rangle$
\Let{$s'$}{$\mi{scriptstate}$}
\Let{$\mi{command}$}{$\an{}$}
\Let{$\mi{target}$}{\textsf{OPENERWINDOW}$(\mi{tree}, $\textsf{PARENTWINDOW}$(\mi{tree}, \mi{docnonce}))$}

\Switch{$s'.\str{q}$}
 \Case{$\str{start}$}
  \Let{$\mi{command}$}{$\langle\tPostMessage$, $\mi{target}$, $\str{ready}$, $\bot\rangle$}
  \Let{$s'.\str{q}$}{$\str{expectTagKey}$}
 \EndCase
 \Case{$\str{expectTagKey}$}
  \Let{$\mi{pattern}$}{$\an{\tPostMessage,\mi{target},*,\an{\str{tagKey},*}}$}
  \Let{$\mi{input}$}{\textsf{CHOOSEINPUT}($\mi{scriptinputs},\mi{pattern}$)}
   \If{$\mi{input} \not\equiv \bot$}
    \Let{$\mi{tagKey}$}{$\proj{2}{\proj{4}{\mi{input}}}$}
    \Let{$\mi{tag}$}{\textsf{GETPARAMETERS}$(\mi{tree}, \mi{docnonce})[\str{tag}]$}
    \Let{$\mi{eia}$}{\textsf{GETPARAMETERS}$(\mi{tree}, \mi{docnonce})[\str{eia}]$}
    \Let{$\mi{rpOrigin}$}{$\an{\decs{\mi{tag}}{\mi{tagKey}}.1, \https}$}
    \Let{$\mi{command}$}{$\an{\tPostMessage, \mi{target}, \an{\str{eia},\mi{eia}}, \mi{rpOrigin}}$}
    \Let{$s'.\str{q}$}{$\str{stop}$}
   \EndIf
 \EndCase
\EndSwitch

\State \textbf{stop} $\an{s',\mi{cookies},\mi{localStorage},\mi{sessionStorage},\mi{command}}$
\end{algorithmic} \setlength{\parindent}{1em}

\section{Formal Security Properties Regarding
  Authentication}\label{app:form-secur-prop}
 To state the security properties for \spresso, we first
define an \emph{\spresso web system for authentication analysis}. This
web system is based on the \spresso web system and only considers one
network attacker (which subsumes all web attackers and further network
attackers).

\begin{definition}
  Let
  $\spressoauthwebsystem = (\bidsystem, \scriptset, \mathsf{script},
  E^0)$
  an \spresso web system. We call $\spressoauthwebsystem$ an
  \emph{\spresso web system for authentication analysis} iff
  $\bidsystem$ contains only one network attacker process
  $\fAP{attacker}$ and no other attacker processes (i.e.,
  $\mathsf{Net} = \{\fAP{attacker}\}$, $\mathsf{Web} = \emptyset$).
  Further, $\bidsystem$ contains no DNS servers. DNS servers are
  assumed to be dishonest, and hence, are subsumed by
  $\fAP{attacker}$. In the initial state $s_0^b$ of each browser $b$
  in $\bidsystem$, the DNS address is
  $\mapAddresstoAP(\fAP{attacker})$. Also, in the initial state
  $s_0^r$ of each relying party $r$, the DNS address is
  $\mapAddresstoAP(\fAP{attacker})$.
\end{definition}

The security properties for \spresso are formally defined as follows.
First note that every RP service token $\an{n,i}$ recorded in RP was
created by RP as the result of an HTTPS $\mPost$ request $m$. We refer
to $m$ as the \emph{request corresponding to $\an{n,i}$}.

\begin{definition}\label{def:spresso-security-property} Let $\spressoauthwebsystem$ be an \spresso web
  system for authentication analysis. We say that
  \emph{$\spressoauthwebsystem$ is secure} if for every run $\rho$ of
  $\spressoauthwebsystem$, every state $(S^j, E^j, N^j)$ in $\rho$,
  every $r\in \fAP{RP}$ that is honest in $S^j$ with
  $s_0(r).\str{FWDDomain}$ being a domain of an FWD that is honest in
  $S^j$, every RP service token of the form $\an{n,i}$ recorded in
  $S^j(r).\str{serviceTokens}$, the following two conditions are
  satisfied:

  \textbf{(A)} If $\an{n,i}$ is derivable from the attackers knowledge
  in $S^j$ (i.e., $\an{n,i} \in d_{\emptyset}(S^j(\fAP{attacker}))$),
  then it follows that the browser $b$ owning $i$ is fully corrupted
  in $S^j$ (i.e., the value of $\mi{isCorrupted}$ is $\fullcorrupt$)
  or $\mapGovernor(i)$ is not an honest IdP (in $S^j$).

  \textbf{(B)} If the request corresponding to $\an{n,i}$ was sent by
  some $b\in \fAP{B}$ which is honest in $S^j$, then $b$ owns $i$.
\end{definition}

\section{Proof of Theorem~\ref{thm:authentication}}
\label{app:proof-spresso}

Before we prove Theorem~\ref{thm:authentication}, we show some general properties of the $\spressoauthwebsystem$.

\subsection{Properties of $\spressoauthwebsystem$}

Let $\spressoauthwebsystem = (\bidsystem, \scriptset, \mathsf{script}, E^0)$
be a web system. In the following, we write $s_x = (S^x,E^x,N^x)$ for the
states of a web system.\gs{hint: the following generic properties also apply for the generic \spresso model.}

\begin{definition}\label{def:emitting}
  In what follows, given an atomic process $p$ and a message $m$, we
  say that \emph{$p$ emits $m$} in a run $\rho=(s_0,s_1,\ldots)$ if
  there is a processing step  of the form
  \[ s_{u-1} \xrightarrow[p \rightarrow E]{} s_{u}\] for some $u \in
  \mathbb{N}$, a set of events $E$ and some addresses $x$, $y$ with
  $\an{x,y,m} \in E$.
\end{definition}

\begin{definition}\label{def:contains}
  We say that a term $t$ \emph{is derivably contained in (a term) $t'$
    for (a set of DY processes) $P$ (in a processing step $s_i
    \rightarrow s_{i+1}$ of a run $\rho=(s_0,s_1,\ldots)$)} if $t$ is
  derivable from $t'$ with the knowledge available to $P$, i.e.,
  \begin{align*}
    t \in d_{\emptyset}(\{t'\} \cup \bigcup_{p\in P}S^{i+1}(p))
  \end{align*}

\end{definition}

\begin{definition}\label{def:leak}
  We say that \emph{a set of processes $P$ leaks a term $t$ (in a
    processing step $s_i \rightarrow s_{i+1}$) to a set of processes
    $P'$} if there exists a message $m$ that is emitted (in $s_i
  \rightarrow s_{i+1}$) by some $p \in P$ and $t$ is derivably
  contained in $m$ for $P'$ in the processing step $s_i \rightarrow
  s_{i+1}$. If we omit $P'$, we define $P' := \bidsystem \setminus
  P$. If $P$ is a set with a single element, we omit the set notation.
\end{definition}

\begin{definition}\label{def:creating}
  We say that an DY process $p$ \emph{created} a message $m$ (at
  some point) in a run if $m$ is derivably contained in a message
  emitted by $p$ in some processing step and if there is no earlier
  processing step where $m$ is derivably contained in a message
  emitted by some DY process $p'$.
\end{definition}

\begin{definition}\label{def:accepting}
  We say that \emph{a browser $b$ accepted} a message (as a response
  to some request) if the browser decrypted the message (if it was an
  HTTPS message) and called the function $\mathsf{PROCESSRESPONSE}$,
  passing the message and the request (see
  Algorithm~\ref{alg:processresponse}).
\end{definition}

\begin{definition}\label{def:rp-accepting}
  In a similar fashion, we say that \emph{an RP $r$ accepted} a
  message (as a response to some request) if the RP decrypted the
  message (RPs can only accept HTTPS messages) and added the message's
  body to the $\str{wkCache}$ in its state (i.e.,
  Line~\ref{line:rp-accept} of Algorithm~\ref{alg:rp-spresso} was called).
\end{definition}

\begin{definition}\label{def:knowing}
  We say that an atomic DY process \emph{$p$ knows a term $t$} in some
  state $s=(S,E,N)$ of a run if it can derive the term from its
  knowledge, i.e., $t \in d_{\emptyset}(S(p))$.
\end{definition}

\begin{definition}\label{def:initiating}
  We say that a \emph{script initiated a request $r$} if a browser
  triggered the script (in Line~\ref{line:trigger-script} of
  Algorithm~\ref{alg:runscript}) and the first component of the
  $\mi{command}$ output of the script relation is either $\tHref$, 
  $\tIframe$, $\tForm$, or $\tXMLHTTPRequest$ such that the browser
  issues the request $r$ in the same step as a result.
\end{definition}

For a run $\rho = (s_0, s_1,\dots)$ of $\spressoauthwebsystem$, we state the
following lemmas:

\begin{lemma}\label{lemma:k-does-not-leak-from-honest-rp} 
  If in the processing step $s_i \rightarrow s_{i+1}$ of a run $\rho$
  of $\spressoauthwebsystem$ an honest relying party $r$ (I) emits an HTTPS
  request of the form

  \[ m = \ehreqWithVariable{\mi{req}}{k}{\pub(k')} \]
  (where $\mi{req}$ is an HTTP request, $k$ is a nonce (symmetric
  key), and $k'$ is the private key of some other DY process $u$), and (II) in the
  initial state $s_0$ the private key $k'$ is only known to $u$, and
  (III) $u$ never leaks $k'$, then all of the following
  statements are true:
  \begin{enumerate}
  \item There is no state of $\spressoauthwebsystem$ where any party except
    for $u$ knows $k'$, thus no one except for $u$ can
    decrypt $\mi{req}$.
    \label{prop:attacker-cannot-decrypt-spresso}
  \item If there is a processing step $s_j \rightarrow s_{j+1}$ where
    the RP $r$ leaks $k$ to $\bidsystem \setminus \{u, r\}$ there
    is a processing step $s_h \rightarrow s_{h+1}$ with $h < j$
    where $u$ leaks the symmetric key $k$ to $\bidsystem \setminus
    \{u,r\}$ or $r$ is corrupted in
    $s_j$. \label{prop:k-doesnt-leak-spresso}
  \item The value of the host header in $\mi{req}$ is the domain that
    is assigned the public key $\pub(k')$ in RP's keymapping
    $s_0.\str{keyMapping}$ (in its initial
    state). \label{prop:host-header-matches-spresso}
  \item If $r$ accepts a response (say, $m'$) to $m$ in a processing step $s_j
    \rightarrow s_{j+1}$ and $r$ is honest in $s_j$ and $u$ did not
    leak the symmetric key $k$ to $\bidsystem \setminus \{u,r\}$ prior
    to $s_j$, then $u$ created the HTTPS response $m'$ to the HTTPS
    request $m$, i.e., the nonce of the HTTP request $\mi{req}$ is not known to
    any atomic process $p$, except for the atomic DY processes $r$ and
    $u$.\label{prop:only-owner-answers-spresso}
  \end{enumerate}
\end{lemma}

\begin{proof} 

  \textbf{(\ref{prop:attacker-cannot-decrypt-spresso})} follows
  immediately from the condition. If $k'$ is initially only known to
  $u$ and $u$ never leaks $k'$, i.e., even with the knowledge of all
  nonces (except for those of $u$), $k'$ can never be derived from any network
  output of $u$, $k'$ cannot be known to any other party. Thus, nobody
  except for $u$ can derive $\mi{req}$ from $m$.

  \textbf{(\ref{prop:k-doesnt-leak-spresso})} 
  We assume that $r$ leaks $k$ to $\bidsystem \setminus \{u,r\}$ in
  the processing step $s_j \rightarrow s_{j+1}$ without $u$ prior
  leaking the key $k$ to anyone except for $u$ and $r$ and that the
  RP is not fully corrupted in $s_j$, and lead this to a contradiction.

  The RP is honest in $s_i$. From the definition of the RP, we see
  that the key $k$ is always a fresh nonce that is not used anywhere
  else. Further, the key is stored in $\mi{pendingRequests}$. The
  information from $\mi{pendingRequests}$ is not extracted or used
  anywhere else, except when handling the received messages, where it
  is only checked against. Hence, $r$ does not leak $k$ to any other
  party in $s_j$ (except for $u$ and $r$). This proves
  (\ref{prop:k-doesnt-leak-spresso}).

  \textbf{(\ref{prop:host-header-matches-spresso})} Per the definition
  of RPs (Algorithm~\ref{alg:rp-spresso}), a host header is always
  contained in HTTP requests by RPs. From
  Line~\ref{line:select-enc-key-spresso} of Algorithm~\ref{alg:rp-spresso}
  we can see that the encryption key for the request $\mi{req}$ was
  chosen using the host header of the message. It is chosen from the
  $\mi{keyMapping}$ in RP's state, which is never changed during
  $\rho$. This proves (\ref{prop:host-header-matches-spresso}).

  \textbf{(\ref{prop:only-owner-answers-spresso})} An HTTPS response
  $m'$ that is accepted by $r$ as a response to $m$ has to be
  encrypted with $k$. The nonce $k$ is stored by the RP in the
  $\mi{pendingRequests}$ state information. The RP only stores freshly
  chosen nonces there (i.e., the nonces are not used twice, or for
  other purposes than sending one specific request). The information
  cannot be altered afterwards (only deleted) and cannot be read
  except when the browser checks incoming messages. The nonce $k$ is
  only known to $u$ (which did not leak it to any other party prior to
  $s_j$) and $r$ (which did not leak it either, as $u$ did not leak it
  and $r$ is honest, see (\ref{prop:k-doesnt-leak-spresso})). The RP $r$
  cannot send responses that are encrypted by symmetric encryption
  keys used for outgoing HTTPS requests (all encryption keys used for
  encrypting responses are taken from the matching HTTPS requests and
  never from $\mi{pendingRequests}$). This proves
  (\ref{prop:only-owner-answers-spresso}). \qed
\end{proof}

\begin{lemma}\label{lemma:wkcache-never-lies}
  For every honest relying party $r \in \fAP{RP}$, every $s \in \rho$, every
  $\an{\mi{host}, \mi{wkDoc}} \inPairing S(r).\str{wkCache}$ it holds
  that $\mi{wkDoc}[\str{signkey}] \equiv
  \pub(\mathsf{signkey}(\mapDomain^{-1}(\mi{host})))$ if
  $\mapDomain^{-1}(\mi{host})$ is an honest IdP.
\end{lemma}

\begin{proof}
  First, we can see that (in an honest RP) $S(r).\str{wkCache}$ can
  only be populated in Line~\ref{line:rp-populate-wkcache} (of
  Algorithm~\ref{alg:rp-spresso}). There, the body of a received message
  $m'$ is written to $S(r).\str{wkCache}$. From
  Line~\ref{line:rp-https-response} we can see that $m'$ is the
  response to a HTTPS message that was sent by $r$. Only in
  Lines~\ref{line:rp-send-request}ff., $r$ can assemble
  (and later sent) such requests.

  All such requests are sent to the path
  $\str{/.well\mhyphen{}known/spresso\mhyphen{}info}$. As the original request was stored
  in $\str{pendingRequests}$, in Line~\ref{line:rp-populate-wkcache},
  we know that $\mi{request}.\str{host}$ is the domain the original
  request was encrypted for and finally sent to.

  With the condition of this lemma we see that
  $\mapDomain^{-1}(\mi{request}.\str{host})$ is an honest IdP, say, $p$.
  Lemma~\ref{lemma:k-does-not-leak-from-honest-rp} applies here and we
  can see that $p$ created the HTTPS response, and it was not altered
  by any other party.
  In Algorithm~\ref{alg:idp-spresso} we can see that an honest IdP
  responds to requests to the path $\str{/.well\mhyphen{}known/spresso\mhyphen{}info}$ in
  Line~\ref{line:idp-response-wk}ff. Here, $p$ constructs a document
  $\mi{wkDoc}$ and sends this document in the body of the HTTPS
  response.
  This document is of the following form: $\an{\an{\str{signkey}, \pub(s'.\str{signkey})}}$. The
  term $s'.\str{signkey}$ is defined in
  Definition~\ref{def:initial-state-idp} to be $\mathsf{signkey}(p)$
  and is never changed in Algorithm~\ref{alg:idp-spresso}.

  Therefore, a pairing of the form
  $\an{\mi{request}.\str{host}, x}$ with $x[\str{signkey}]\equiv
  \pub(\mathsf{signkey}(\mapDomain^{-1}(\mi{request}.\str{host})))$ is
  stored in $S(r).\str{wkCache}$. As this applies to all pairings in $S(r).\str{wkCache}$,
  this proves the lemma.

 \qed
\end{proof}

\begin{definition}
  For every service token $\an{n,i}$ we define a \emph{service token
    response for $\an{n,i}$} to be an HTTPS response where the value
  $n$ is contained in the body of the message. A \emph{service token
    request for $\an{n,i}$} is an HTTPS request that triggered the
  service token response for $\an{n,i}$.
\end{definition}

\begin{lemma}\label{lemma:spresso-request-exists}
  In a run $\rho$ of $\spressoauthwebsystem$, for every state $s_j \in
  \rho$, every RP $r \in \fAP{RP}$ that is honest in $s_j$, every
  $\an{n,i} \inPairing S^j(r).\str{serviceTokens}$, the following
  properties hold:

  \begin{enumerate}
  \item There exists exactly one $l' < j$ such that there exists a
    processing step in $\rho$ of the form
    \[ s_{l'} \xrightarrow[r \rightarrow \an{\an{a',f',m'}}]{e'
      \rightarrow r} s_{l'+1}\]
    with $e'$ being some events, $a'$ and $f'$
    being addresses and $m'$ being a service token response for $\an{n,i}$.

  \item There exists exactly one $l < j$ such that there exists a
    processing step in $\rho$ of the form 
    \[ s_{l} \xrightarrow[r \rightarrow e]{\an{a,f,m} \rightarrow r}
    s_{l+1} \] with $e$ being some events, $a$ and $f$ being
    addresses and $m$ being a service token request for $\an{n,i}$.

  \item The processing steps from (1) and (2) are the same, i.e., $l = l'$.

  \item The service token request for $\an{n,i}$, $m$ in (2), is an HTTPS message of the following form:
    \[ \mathsf{enc}_\mathsf{a}(\langle \hreq{ 
      nonce=n_\text{req}, 
      method=\mPost,
      xhost=d_r, 
      path=\str{/login}, 
      parameters=x, 
      headers=h,
      xbody=b}, k\rangle, \pub(\mapSSLKey(d_r))) \]
    for $d_r \in \mapDomain(r)$, some terms $x$, $h$, $n_\text{req}$, and a dictionary $b$ such that 
    \[ b[\str{eia}] \equiv \encs{\sig{\an{\mi{tag}, i, S(r).\str{FWDDomain}}}{k_\text{sign}}}{\mi{iaKey}} \]
    with 
    \[ \mi{tag} \equiv \encs{\an{d_r, n_\text{rp}}}{\mi{tagKey}}, \]
    \[ i \equiv S^l(r).\str{loginSessions}[b[\str{loginSessionToken}]].\str{email}, \]
    \[ \mi{tag} \equiv S^l(r).\str{loginSessions}[b[\str{loginSessionToken}]].\str{tag}, \]
    \[ \mi{iaKey} \equiv S^l(r).\str{loginSessions}[b[\str{loginSessionToken}]].\str{iaKey} \]
    for some nonces $n_\text{rp}$, and $k_\text{sign}$.
  \item If the governor of $i$ is an honest IdP, we have that $k_\text{sign} = \mathsf{signkey}(\mapGovernor(i))$.
  \end{enumerate}
\end{lemma}

\begin{proof}
  \textbf{(1).} The service token nonce $n$ of service tokens
  $\an{n,i} \inPairing S^j(r).\str{serviceTokens}$ can only be
  contained in a response that is assembled in
  Lines~\ref{line:assemble-login-response}ff of
  Algorithm~\ref{alg:rp-spresso}. The $n$ is freshly chosen in
  Line~\ref{line:choose-service-token}, stored (along with the
  identity $i$) to $S^j(r).\str{serviceTokens}$ (actually to
  $S^q(r).\str{serviceTokens}$ for some $q \leq j$) in
  Line~\ref{line:store-service-token} and sent out in the service
  token response in Line~\ref{line:send-service-token}f. The service
  tokens stored in $S^j(r).\str{serviceTokens}$ are not used or
  altered anywhere else. Therefore, each service token nonce is sent
  in exactly one (service token) response.

  \noindent
  \textbf{(2).} From Line~\ref{line:assemble-login-response} of
  Algorithm~\ref{alg:rp-spresso} it is
  easy to see that each service token response is triggered by exactly
  one request.

  \noindent
  \textbf{(3).} Follows immediately from (2). 

  \noindent
  \textbf{(4).} The basic form of the encrypted HTTPS request, the
  host header, and the usage of the correct encryption key are
  enforced by Lines~\ref{line:receive-https-request}f. The $\mi{path}$
  component is checked to be $\str{/login}$ and the $\mi{method}$
  component is checked to be $\mPost$ in
  Line~\ref{line:check-path-method-rp}. The values of $b[\str{eia}]$, $\mi{i}$,
  $\mi{tag}$, and $\mi{iaKey}$ are checked in
  Lines~\ref{line:check-service-token-request-contents}ff.

  \noindent
  \textbf{(5).} In Line~\ref{line:rp-check-ia}, the term $\mi{ia}$ is
  checked to be signed with the signature key stored in
  $S^q(r).\str{wkCache}$ indexed under the domain of the email address
  $i$ (for some $q \leq j$). With
  Lemma~\ref{lemma:wkcache-never-lies}, we can see that for the domain
  of the email address $i$ this signature key is
  $\mathsf{signkey}(\mapDomain^{-1}(i.\str{domain}))$. With
  $\mapDomain^{-1}(i.\str{domain}) = \mapGovernor(i)$ we can see that
  $\mi{ia}$ must have been signed with the signature key of the honest
  IdP that governs the email address $i$. Further, in the same line,
  the contents of the signature, including the tag, are checked. 

  \qed
\end{proof}

\subsection{Property A}

As stated above, the Property A is defined as follows: 
\begin{definition}\label{def:spresso-security-property} Let $\spressoauthwebsystem$ be an \spresso web
  system for authentication analysis. We say that \emph{$\spressoauthwebsystem$ is secure (with respect to Property A)} if for every
  run $\rho$ of $\spressoauthwebsystem$, every state $(S^j, E^j, N^j)$ in $\rho$,
  every $r\in \fAP{RP}$ that is honest in $S^j$ with
  $S^0(r).\str{FWDDomain}$ being a domain of an FWD that is honest in
  $S^j$, every RP service token of the form $\an{n,i}$ recorded in
  $S^j(r).\str{serviceTokens}$ and derivable from the attackers
  knowledge in $S^j$ (i.e., $\an{n,i} \in
  d_{\emptyset}(S^j(\fAP{attacker}))$), it follows that
  the browser $b$ owning $i$ is fully corrupted in $S^j$ (i.e., the
  value of $\mi{isCorrupted}$ is $\fullcorrupt$) or $\mapGovernor(i)$
  is not an honest IdP (in $S^j$). 
\end{definition}

We want to show that every \spresso web system is secure with regard to
Property A and therefore assume that there exists an \spresso web system
that is not secure. We will lead this to a contradication and thereby
show that all \spresso web systems are secure (with regard to Property A).

In detail, we assume: \emph{There exists an \spresso web system
$\spressoauthwebsystem$, a run $\rho$ of $\spressoauthwebsystem$, a
state $s_j = (S^j, E^j, N^j)$ in $\rho$, a RP $r\in \fAP{RP}$ that is
honest in $S^j$ with $S^0(r).\str{FWDDomain}$ being a domain of an FWD
that is honest in $S^j$, an RP service token of the form $\an{n,i}$
recorded in $S^j(r).\str{serviceTokens}$ and derivable from the
attackers knowledge in $S^j$ (i.e., $\an{n,i} \in
d_{\emptyset}(S^j(\fAP{attacker}))$), and the browser $b$ owning $i$
is not fully corrupted and $\mapGovernor(i)$ is an honest IdP (in
$S^j$).}

We now proceed to to proof that this is a contradiction. First, we can
see that for $\an{n,i}$ and $s_j$, the conditions in
Lemma~\ref{lemma:spresso-request-exists} are fulfilled, i.e., a
service token request $m$ and a service token response $m'$ to/from
$r$ exist, and $m'$ is of form shown in
Lemma~\ref{lemma:spresso-request-exists} (4). Let $I :=
\mapGovernor(i)$. We know that $I$ is an honest IdP. As such, it never
leaks its signing key (see Algorithm~\ref{alg:idp-spresso}).
Therefore, the signed subterm $\mi{ia} := \sig{\an{\mi{tag}, i,
    S(r).\str{FWDDomain}}}{\mathsf{signkey}(I)}$ had to be created by
the IdP $I$.
An (honest) IdP creates signatures only in Line~\ref{line:sign-ia} of
Algorithm~\ref{alg:idp-spresso}.

\begin{lemma}
  Under the assumption above, only the browser $b$ can issue a request
  (say, $m_\text{cert}$) that triggers the IdP $I$ to create the
  signed term $\mi{ia}$. The request $m_\text{cert}$ was sent by $b$
  over HTTPS using $I$'s public HTTPS key.
\end{lemma}
\begin{proof}
  We have to consider two cases for the request $m_\text{cert}$:

  \textbf{(A).} First, if the user is not logged in with the identity $i$ at $I$
  (i.e., the browser $b$ has no session cookie that carries a nonce
  which is a session id at $I$ for which the identitiy $i$ is marked
  as being logged in, compare
  Line~\ref{line:spresso-idp-check-login-state} of
  Algorithm~\ref{alg:idp-spresso}), then the request has to carry (in
  the request body) the password matching the identity $i$
  ($\mapIDtoPLI(i)$). This secret is only known to $b$ initially.
  Depending on the corruption status of $b$, we can now have two
  cases:
  \begin{enumerate}
  \item[a)] If $b$ is honest in $s_j$, it has not sent the secret to
    any party except over HTTPS to $I$ (as defined in the definition
    of browsers). 
  \item[b)] If $b$ is close-corrupted, it has not sent it to any other
    party while it was honest (case a). When becoming close-corrupted,
    it discarded the secret.
  \end{enumerate}
  I.e., the secret has been sent only to $I$ over HTTPS or to nobody
  at all. The IdP $I$ cannot send it to any other party. Therefore we
  know that only the browser $b$ can send the request $m_\text{cert}$
  in this case.

  \textbf{(B).} Second, if the user is logged in for the identity $i$
  at $I$, the browser provides a session id to $I$ that refers to a
  logged in session at $I$. This session id can only be retrieved from
  $I$ by logging in, i.e., case (A) applies, in particular, $b$ has to
  provide the proper secret, which only itself and $I$ know (see
  above). The session id is sent to $b$ in the form of a cookie, which
  is set to secure (i.e., it is only sent back to $I$ over HTTPS, and
  therefore not derivable by the attacker, see prior work \df{todo})
  and httpOnly (i.e., it is not accessible by any scripts). The
  browser $b$ sends the cookie only to $I$. The IdP $I$ never sends
  the session id to any other party than $b$. The session id therefore
  only leaks to $b$ and $I$, and never to the attacker. Hence, the
  browser $b$ is the only atomic DY process which can send the request
  $m_\text{cert}$ in this case.

  We can see that in both cases, the request was sent by $b$ using
  HTTPS and $I$'s public key: If the browser would intend to sent the
  request without encryption, the request would not contain the
  password in case (A) or the cookie in case (B). The browser always
  uses the ``correct'' encryption key for any domain (as defined in
  $\spressoauthwebsystem$).\qed
\end{proof}

As the request $m_\text{cert}$ is sent over HTTPS, it cannot be
altered or read by any other party. In particular, it is easy to see
that at the point in the run where $m_\text{cert}$ was sent, $b$ was honest (otherwise,
it would have had no knowledge of the secret anymore).

\begin{lemma}
  In the browser $b$, the request $m_\text{cert}$ was triggered by
  $\mi{script\_idp}$ loaded from the origin $\an{d,\https}$ for some
  $d \in \mapDomain(I)$. 
\end{lemma}

\begin{proof}
  First, $\an{d,\https}$ for some $d \in \mapDomain(I)$ is the only
  origin that has access to the secret $\mapIDtoPLI(i)$ for the
  identity $i$ (as defined in Appendix~\ref{app:browsers-spresso}).

  With the general properties defined in~\cite{FettKuestersSchmitz-TR-BrowserID-Primary-2015} and the
  definition of Identity Providers in Appendix~\ref{app:idps}, in
  particular their property that they only send out one script,
  $\mi{script\_idp}$, we can see that this is the only script that can
  trigger a request containing the secret.\qed
\end{proof}

\begin{lemma}\label{lemma:idp-to-script-idp}
  In the browser $b$, the script $\mi{script\_idp}$ receives the
  response to the request $m_\text{cert}$ (and no other script), and
  at this point, the browser is still honest.
\end{lemma}

\begin{proof}
  From the definition of browser corruption, we can see that the
  browser $b$ discards any information about pending requests in its
  state when it becomes close-corrupted, in particular any SSL keys.
  It can therefore not decrypt the response if it becomes
  close-corrupted before receiving the response.

  The rest follows from the general properties defined
  in~\cite{FettKuestersSchmitz-TR-BrowserID-Primary-2015}.\qed
\end{proof}

We now know that only the script $\mi{script\_idp}$ received the
response containing the IA. For the following lemmas, we will assume
that the browser $b$ is honest. In the other case (the browser is
close-corrupted), the IA $\mi{ia}$ and any information about pending
HTTPS requests (in particular, any decryption keys) would be discarded
from the browser's state (as seen in the proof for
Lemma~\ref{lemma:idp-to-script-idp}). This would be a contradiction to
the assumption (which requires that the IA arrived at the RP).

\begin{lemma}\label{lemma:script-idp-to-script-fwd}
  After receiving $\mi{ia}$, $\mi{script\_idp}$ forwards the $\mi{ia}$
  only to an FWD that is honest (in $s_j$, and therefore, also at any
  earlier point in the run) and a document $\mi{script\_fwd}$ that was
  loaded from this FWD over HTTPS.
\end{lemma}

\begin{proof}
  We know that the browser $b$ is either close-corrupted (in which
  case the $\mi{ia}$ would be discarded as it is only stored in the
  window structure, or, more precisely, the script states inside the
  window structure of the browser, which are removed when the browser
  becomes close-corrupted) or it is honest. In the latter case,
  $\mi{script\_idp}$ (defined in Algorithm~\ref{alg:spresso-script-idp})
  opens an iframe from the FWDDomain that was given to it by RP. It
  always uses HTTPS for this request.
  
  We can see that $\mi{script\_idp}$ forwards the $\mi{ia}$ to the
  domain stored in the variable $\mi{FWDDomain}$
  (Line~\ref{line:send-to-fwd} of
  Algorithm~\ref{alg:spresso-script-idp}). This variable is set five
  lines earlier with the value taken from the parameters of the
  current document. While we cannot know the actual value of the
  parameter $\str{FWDDomain}$ yet, we know that this parameter does not
  change (in the browser definition, it is only set once, when the
  document is loaded). We can also see that the very same parameter
  was sent to $I$ in Line~\ref{line:send-to-idp} as the value for the
  FWD domain that was then signed by $I$ in the $\mi{ia}$. As we know
  the value of the FWD origin in the $\mi{ia}$ (it is
  $S(r).\str{FWDDomain}$), we know that the domain to which the
  $\mi{ia}$ is forwarded is the same.

  From our assumption, we know that $S(r).\str{FWDDomain}$ is the
  origin of an honest FWD in $s_j$. It is contacted over HTTPS, so the
  general properties defined
  in~\cite{FettKuestersSchmitz-TR-BrowserID-Primary-2015}
  apply. According to the definition of forwarders
  (Algorithm~\ref{alg:spresso-fwd}), they only respond with
  $\mi{script\_fwd}$. The $\mi{ia}$ is therefore only forwarded to the
  FWD and its script $\mi{script\_fwd}$.\qed
\end{proof}

\begin{lemma}\label{lemma:script-fwd-to-script-rp}
  The script $\mi{script\_fwd}$ forwards the $\mi{ia}$ only to the
  script $\mi{script\_rp}$ loaded from the origin $\an{d_r, \https}$.
\end{lemma}

\begin{proof}
  The script $\mi{script\_idp}$ that runs in the honest browser $b$
  forwards the (then encrypted) IA along with the tag to
  $\mi{script\_fwd}$. From the definition of the IdP script
  (Algorithm~\ref{alg:spresso-script-idp}) it is clear that the tag that
  is forwarded along with the encrypted IA is the same that was signed
  by the IdP.

  This script (Algorithm~\ref{alg:spresso-fwd}) tries to decrypt the tag
  (once it receives a matching key) and sends a postMessage containing
  the encrypted IA to the domain contained in the tag, which is $d_r$.

  The protocol part of the origin is HTTPS. The only document that $r$
  delivers and which receives postMessages is $\mi{script\_rp}$, and
  this therefore is the only script that can receive this postMessage. \qed
\end{proof}

\begin{lemma}\label{lemma:script-rp-to-rp}
  From the RP document, the EIA is only sent to the RP $r$ and over
  HTTPS.
\end{lemma}

\begin{proof}
  This follows immediately from the definition of $\mi{script\_rp}$
  (see Algorithm~\ref{alg:spresso-script-rp}, in particular
  Line~\ref{line:send-ia-to-r} in conjunction with
  Line~\ref{line:determine-origin-rp-script}) and the fact that the RP
  document must have been loaded from the origin $\an{d_r, \https}$
  (as shown above).
\end{proof}

With Lemmas~\ref{lemma:idp-to-script-idp}--\ref{lemma:script-rp-to-rp} we
see that the $\mi{ia}$, once it was signed by $I$, was transferred
only to $r$, the browser $b$, and to an honest forwarder. It
cannot be known to the attacker or any corrupted party, as none of
the listed parties leak it to any corrupted party or the attacker.

Now, for $\an{n,i}$ to be created and recorded in $S^j(r)$, a message
$m$ as shown above has to be created and sent. This can only be done
with knowledge of $\mi{eia}$. From their definitions, we can see that
neither $I$, $r$ nor any forwarder create such a message, with the
only option left being $b$. If $b$ sends such a request, it is the
only party able to read the response (see general security properties
in~\cite{FettKuestersSchmitz-TR-BrowserID-Primary-2015}) and it will not do anything with the contents
of the response (see Algorithm~\ref{alg:spresso-script-rp}), in
particular not leak it to the attacker or any corrupted party.

This is a contradication to the assumption, where we assumed that
$\an{n,i} \in d_{\emptyset}(S^j(\fAP{attacker}))$. This shows
every $\spressoauthwebsystem$ is secure in the sense of Property~A.

\QED

\subsection{Property B}

As stated above, Property B is defined as follows: 
\begin{definition}
  Let $\spressoauthwebsystem$ be an \spresso web system. We say that
  \emph{$\spressoauthwebsystem$ is secure (with respect to Property B)} if
  for every run $\rho$ of $\spressoauthwebsystem$, every state $(S^j, E^j, N^j)$
  in $\rho$, every $r\in \fAP{RP}$ that is honest in $S^j$ with
  $S^0(r).\str{FWDDomain}$ being a domain of an FWD that is honest in
  $S^j$, every RP service token of the form $\an{n,i}$ recorded in
  $S^j(r).\str{serviceTokens}$, with the request corresponding to
  $\an{n,i}$ sent by some $b\in \fAP{B}$ which is honest in $S^j$,
  $b$ owns $i$.
\end{definition}

Applying Lemma~\ref{lemma:spresso-request-exists} (1--4), we call the
request corresponding to $\an{n,i}$ (or service token request) $m$ and
its response $m'$, and (as in Lemma~\ref{lemma:spresso-request-exists}
(2)) we refer to the state of $\spressoauthwebsystem$ in the run $\rho$
where $r$ processes $m$ by $s_l$.

\begin{lemma}\label{lemma:request-m-is-from-script-rp}
  The request $m$ was sent by $\mi{script\_rp}$ loaded from the origin
  $\an{d_r, \https}$ where $d_r$ is some domain of $r$.
\end{lemma}

\begin{proof}
  The request $m$ is XSRF protected. In
  Algorithm~\ref{alg:rp-spresso}, Line~\ref{line:check-origin-header},
  RP checks the presence of the Origin header and its value. If the
  request $m$ was initiated by a document from a different origin than
  $\an{d_r, \https}$, the (honest!) browser $b$ would have added an
  Origin header that would not pass this test (or no Origin header at
  all), according to the browser definition. The script
  $\mi{script\_rp}$ is the only script that the honest party $r$ sends
  as a response and that sends a request to $r$. \qed
\end{proof}

\begin{lemma}\label{lemma:m-contains-lst-for-i}
  The request $m$ contains a nonce $\mi{loginSessionToken}$ such that
  \[S^l(r).\str{loginSessions}[\mi{loginSessionToken}].\str{email} \equiv i'\] and
  $b$ owns $i'$, i.e., $\mapIDtoOwner(i') = b$.
\end{lemma}

\begin{proof}
  With Lemma~\ref{lemma:request-m-is-from-script-rp} we know that the
  request was sent by $\mi{script\_rp}$. In
  Algorithm~\ref{alg:spresso-script-rp} defining $\mi{script\_rp}$, in
  Line~\ref{line:spresso-rp-script-assemble-login-body}, the body of the
  request $m$ is assembled (and this is the only line where this
  script sends a request that contains the same path as $m$). The
  login session token is taken from the script's state
  ($\mi{loginSessionToken}$). This part of the state is initially set
  to $\bot$ and is only changed in
  Line~\ref{line:set-login-session-token}. There, it is taken from the
  response to the start login XHR issued in
  Line~\ref{line:send-start-login} (the request and response are
  coupled using $\mi{refXHR}$ which is tracked in the script's state).
  In Line~\ref{line:select-email-address}, the script selects one of
  the browser's identities (which are the identities that the browser
  owns, by the definition of browsers in
  Appendix~\ref{app:browsers-spresso}). This identity is then used in the
  start login XHR.

  When receiving this request (which is an HTTPS message, and
  therefore, cannot be altered nor read by the attacker),
  ultimatively, the function $\mi{SENDSTARTLOGINRESPONSE}$
  (Algorithm~\ref{alg:sendstartloginresponse}) is called. There are
  two cases how this function can be called (see
  Line~\ref{line:branch-wk-cache} of Algorithm~\ref{alg:rp-spresso}):

  \begin{itemize}
  \item If the well-know cache of $r$ already contains an entry for
    the host contained in the email address,
    $\mi{SENDSTARTLOGINRESPONSE}$ is called immediately with the email
    address contained in the request's body.
  \item Else, the email address in the request's body is stored,
    together with the request's HTTP nonce, the HTTPS encryption key
    and other data, in the subterm $\str{pendingDNS}$ of $r$'s state.
    From there, it is later moved to $\str{pendingRequests}$
    (Line~\ref{line:move-reference-to-pending-request}). Finally, in
    Line~\ref{line:call-ssl-second}, $\mi{SENDSTARTLOGINRESPONSE}$ is
    called. 
  \end{itemize}

  We will come back to these two cases further down.

  After $\mi{SENDSTARTLOGINRESPONSE}$ is called, a new
  $\mi{loginSessionToken}$ is chosen and in the dictionary
  $S^x(r).\str{loginSessions}[\mi{loginSessionToken}]$ the email
  address (along with other data) is stored (for some $x$).

  The $\mi{loginSessionToken}$ is then sent as a response to $m$, in
  particular, it is encrypted with the symmetric key $k$ contained in
  the request. In the first case listed above, the $k$ is immediately
  retrieved from the request. Otherwise, the relationship between $k$
  and the email address is preserved in any case: If the receiver can
  decrypt the response to $m$, it sent the email address $i'$ in the
  request.

  As explained above, $\mi{script\_rp}$ takes the
  $\mi{loginSessionToken}$ from the response body and stores it in its
  state to later use it in the request $m$. Therefore the start login
  XHR described above must have taken place before $m$, i.e., $x < l$.

  The entries in the dictionary $\str{loginSessions}$ can not be
  altered and only be removed when a service token request with the
  corresponding value of $\mi{loginSessionToken}$ is processed. As
  each $\mi{loginSessionToken}$ is not leaked
  to any other party except $r$, we know that
  $S^l(r).\str{loginSessions}[\mi{loginSessionToken}].\str{email}
  \equiv i'$. As shown above, due to the way $i'$ is selected by the
  script, $b$ owns $i'$. \qed
\end{proof}

With Lemma~\ref{lemma:m-contains-lst-for-i}, we can now show that $i =
i'$: In Line~\ref{line:assemble-service-token} of
Algorithm~\ref{alg:rp-spresso}, the service token is assembled. In
particular, $i$ is chosen to be
$S^l(r).\str{loginSessions}[\mi{loginSessionToken}].\str{email}$, and
therefore $i = i'$ and $b$ owns $i$.

\QED

\section{Indistinguishability of Web Systems}

\gs{removed description/intro text as the paper contains a way better
  version of it}

\begin{definition}[Web System Command and Schedule]
  We call a term $\zeta$ a \emph{web system command} (or simply,
  command) if $\zeta$ is of the form
  \[
  \an{i,j,\tau_\text{process},\mi{cmd}_\text{switch},\mi{cmd}_\text{window},\tau_\text{script},\mi{url}}
  \]
  
  The components are defined as follows:
  \begin{itemize}
  \item $i \in \mathbb{N}$,
  \item $j \in \mathbb{N}$,
  \item $\mi{cmd}_\text{switch} \in \{1,2,3\}$,
  \item $\mi{cmd}_\text{window} \in \mathbb{N}$,
  \item
    $\tau_\text{script} \in \gterms_\emptyset(V_\text{script} \cup
    \{x\})$
    with $x$ being a variable and $V_\text{script}$ the set of
    placeholders for scripting processess (see
    Definition~\ref{def:placeholder-sp}).
  \item
    $\tau_{\text{process}} \in \gterms_\emptyset(V_\text{process} \cup
    \{x\})$
    with $x$ being a variable and $V_\text{process}$ the set of
    placeholders (see
    Definition~\ref{def:groundterms-messages-placeholders-protomessages}).
  \item $\mi{url} \in \urls$ with $\urls$ being the set of all valid
    URLs (see Definition~\ref{def:url}).
  \end{itemize}

  We call a (finite) sequence $\sigma = \an{\zeta_1,\ldots,\zeta_n}$,
  with $\zeta_1,\ldots,\zeta_n$ being web system commands, a \emph{web
    system schedule} (or simply, schedule).
\end{definition}

\begin{definition}[Induced Processing Step]
  Let
  $\completewebsystem = (\websystem,\scriptset,\mathsf{script},E^0)$
  be a web system and
  \[(S,E,N) \xrightarrow[p \rightarrow E_{\text{out}}]{\an{a,f,m}
    \rightarrow p} (S', E', N')\]
  be a processing step of $\websystem$ (as in
  Definition~\ref{def:processing-step}) with $E = (e_1, e_2, \ldots)$
  and
  \[\zeta =
  \an{i,j,\tau_\text{process},\mi{cmd}_\text{switch},\mi{cmd}_\text{window},\tau_\text{script},\mi{url}}\]
  a web system command. We say that this processing step is
  \emph{induced by $\zeta$} iff

  \begin{enumerate}
  \item $e_i = \an{a,f,m}$.
  \item Under a lexicographic ordering of $\websystem$, $p$ is the
    $j$-th process in $\websystem$ with $a \in I^p$.
  \item
    $E' = E_\text{out} \cdot (e_1, \ldots, e_{i-1}, e_{i+1}, \ldots)$.
  \item If $p$ is a (web) attacker process or $p$ is a corrupted
    browser (i.e., $S(p).\str{isCorrupted} \not\equiv \bot$), then
    $E_\text{out} = \an{ e_\text{out}}$ with
    $\an{S'(p),e_\text{out}} =
    \tau_\text{process}[\an{e_i,s}\!/\!x]\nf$.
  \item If $p$ is an honest browser (i.e.,
    $S(p).\str{isCorrupted} \equiv \bot$) and $m\equiv \str{TRIGGER}$,
    the browser relation behaves as follows and $E_\text{out}$ and
    $S'(p)$ are obtained accordingly:
    \begin{enumerate}
    \item If $\mi{cmd}_\text{switch} = 1$, the browser relation
      chooses $\mi{switch}=1$ in Line~\ref{line:browser-switch} of
      Algorithm~\ref{alg:browsermain} and $\overline{w}$ in
      Line~\ref{line:browser-trigger-window} of
      Algorithm~\ref{alg:browsermain} such that $\overline{w}$ is the
      $\mi{cmd}_\text{window}$-th window in the tree of browser's
      state $S(p).\str{windows}$. If this script is not the attacker
      script, the browser (deterministically) executes the script in
      this window. Otherwise, in Line~\ref{line:trigger-script} of
      Algorithm~\ref{alg:runscript}, the browser relation chooses the
      output of the script (of this window) as
      $\mi{out}^\lambda = \tau_\text{script}[\mi{in}\!/\!x]\nf$ with
      the variable $\mi{in}$ (deterministically) chosen in
      Line~\ref{line:browser-scriptinputs} of
      Algorithm~\ref{alg:runscript}.

    \item If $\mi{cmd}_\text{switch} = 2$, the browser relation
      chooses $\mi{switch}=2$ in Line~\ref{line:browser-switch} of
      Algorithm~\ref{alg:browsermain} and
      $\mi{protocol}$, $\mi{host}$, $\mi{domain}$, $\mi{path}$, $\mi{parameters}$
      in Line~\ref{line:browser-choose-url}ff. of
      Algorithm~\ref{alg:browsermain} such that
      $\mi{url} = \an{\tUrl, \mi{protocol},\mi{host}, \mi{path},
        \mi{parameters}}$.

    \item If $\mi{cmd}_\text{switch} = 3$, the browser relation
      chooses $\mi{switch}=3$ in Line~\ref{line:browser-switch} of
      Algorithm~\ref{alg:browsermain} and $\overline{w}$ in
      Line~\ref{line:browser-reload-window} of
      Algorithm~\ref{alg:browsermain} such that $\overline{w}$ is the
      $\mi{cmd}_\text{window}$-th window in the tree of browser's
      state $S(p).\str{windows}$. (The browser then starts to reload
      the document in this window.)
    \end{enumerate}
  \end{enumerate}

  We write
  \[
  (S,E,N) \xrightarrow{\zeta} (S',E',N') \ .
  \]
\end{definition}

\begin{corollary}\label{corr:processingstep-not-induced}
  In some cases a command
  $\sigma =
  \an{i,j,\tau_\text{process},\mi{cmd}_\text{switch},\mi{cmd}_\text{window},\tau_\text{script},\mi{url}}$
  does not induce a processing step under the configuration $(S,E,N)$
  in a web system: If $i > |E|$, a processing step cannot
  be induced. The same applies if $j$ does not refer to an existing
  process. Also, if the command schedules a $\str{TRIGGER}$ message to
  be delivered to a browser $p$, $\mi{cmd}_\text{switch} \in \{1,3\}$,
  and $\mi{cmd}_\text{window} > |\mathsf{Subwindows}(S(p))|$ (i.e.,
  the command chooses a window of the browser $p$, which does not
  exist), then no processing step can be induced.
\end{corollary}

\begin{definition}[Induced Run]
  Let
  $\completewebsystem = (\websystem,\scriptset,\mathsf{script},E^0)$
  be a web system, $\sigma = \an{\zeta_1, \ldots, \zeta_n}$ be a
  finite web system schedule, and $N^0$ be an infinite sequence of
  pairwise disjoint nonces. We say that a finite run
  $\rho = ((S^0,E^0,N^0),\ldots,(S^n,E^n,N^n))$ of the system
  $\websystem$ is \emph{induced by $\sigma$ under nonces $N^0$} iff
  for all $1 \le i \le n$, $\zeta_i$ induces the processing step
  \[(S^{i-1},E^{i-1},N^{i-1}) \xrightarrow{\zeta_i} (S^i, E^i, N^i) \
  .\]
  We denote the set of runs induced by $\sigma$ under all infinite
  sequences of pairwise disjoint nonces $N^0$ by
  $\sigma(\completewebsystem)$.
\end{definition}

To define the notion of indistinguishability for web systems, we need
to define the notion of static equivalence of terms in our model. This
definition follows the notion of static equivalence by Abadi and
Fournet~\cite{AbadiFournet-POPL-2001}.

\begin{definition}[Static Equivalence]
  Let $t_1$, $t_2 \in \terms(V)$ be two terms with $V$ a set of
  variables. We say that $t_1$ and $t_2$ are \emph{statically
    equivalent}, written $t_1 \approx t_2$, iff for all terms $M$,
  $N \in \gterms_\emptyset(\{x\})$ with $x$ a variable and
  $x \not\in V$, it holds true that
  \[
  M[t_1\!/\!x] \equiv N[t_1\!/\!x] \quad \Leftrightarrow \quad
  M[t_2\!/\!x] \equiv N[t_2\!/\!x].
  \]
\end{definition}

\begin{definition}[Web System with Distinguished Attacker]\label{def:ws-d-att}
  Let
  $\completewebsystem = (\websystem,\scriptset,\mathsf{script},E^0)$
  be a web system with $\websystem$ partitioned into $\fAP{Hon}$,
  $\fAP{Web}$, and $\fAP{Net}$ (as in Definition~\ref{def:websystem}).
  Let $\fAP{attacker} \in \websystem$ be an attacker process (out of
  $\fAP{Web} \cup \fAP{Net}$). We call
  $\hat{\completewebsystem} =
  (\websystem,\scriptset,\mathsf{script},E^0,\fAP{attacker})$
  a \emph{web system with the distinguished attacker
    $\fAP{attacker}$}.
\end{definition}\gs{$\downarrow$ should we restrict this notion just to attackers?
  we could also say: a web system with distinguished process}

\begin{definition}[Indistinguishability]\label{def:indistinguishability}
  Let
  $\completewebsystem_0 =
  (\websystem_0,\scriptset_0,\mathsf{script}_0,E_0^0,p_0)$,
  $\completewebsystem_1 =
  (\websystem_1,\scriptset_1,\mathsf{script}_1,E_1^{0},p_1)$
  be web systems with a distinguished attacker. We call
  $\completewebsystem_0$ and $\completewebsystem_1$
  \emph{indistinguishable under the schedule $\sigma$} iff for every
  schedule $\sigma$ and every $i \in \{0,1\}$, we have that for every
  run $\rho \in \sigma(\completewebsystem_i)$ there exists a run
  $\rho' \in \sigma(\completewebsystem_{1-i})$ such that
  $\rho(p_i) \approx \rho'(p_{1-i})$.

  We call $\completewebsystem_0$ and $\completewebsystem_1$
  \emph{indistinguishable} iff they are indististinguishable under all
  schedules $\sigma$.
\end{definition}

\section{Formal Proof of Privacy}
\label{app:formal-proof-privacy}

We will here first describe the precise model that we use for privacy.
After that, we define an equivalence relation between configurations,
which we will then use in the proof of privacy.

\subsection{Formal Model of \spresso for Privacy Analysis}
\label{app:model-spresso-priv}

\begin{definition}[Challenge Browser]
  Let $\mi{dr}$ some domain and $b(\mi{dr})$ a DY process. We call $b(\mi{dr})$ a \emph{challenge browser} iff $b$
  is defined exactly the same as a browser (as described in
  Appendix~\ref{app:deta-descr-brows}) with two exceptions: (1) the
  state contains one more property, namely $\mi{challenge}$, which  initially contains the term $\top$. (2)
  Algorithm~\ref{alg:send} is extended by the following at its very
  beginning: It is checked if a message $m$ is addressed to the domain
  $\str{CHALLENGE}$ (which we call the challenger domain). If $m$ is
  addressed to this domain and no other message $m'$ was addressed to
  this domain before (i.e., $\mi{challenge} \not\equiv \bot$), then
  $m$ is changed to be addressed to the domain $\mi{dr}$ and $\mi{challenge}$ is set to $\bot$ to recorded
  that a message was addressed to $\str{CHALLENGE}$.
\end{definition}

\begin{definition}[Deterministic DY Process]
 We call a DY process $p = (I^p,Z^p,R^p,s_0^p)$ \emph{deterministic} iff the relation $R^p$ is a (partial) function.

We call a script $R_\text{script}$ \emph{deterministic} iff the relation $R_\text{script}$ is a (partial) function.
\end{definition}

\begin{definition}[\spresso Web System for Privacy Analysis]\label{def:spresso-ws-priv}
  Let $\spressowebsystem = (\bidsystem, \scriptset, \mathsf{script},
  E^0)$ be an \spresso web system with $\bidsystem = \mathsf{Hon} \cup
  \mathsf{Web} \cup \mathsf{Net}$, $\mathsf{Hon} = \fAP{B} \cup
  \fAP{RP} \cup \fAP{IDP} \cup \fAP{FWD} \cup \fAP{DNS}$ (as described
  in Appendix~\ref{app:outlinespressomodel}), $\fAP{RP} =
  \{r_1,r_2\}$, $\fAP{FWD} = \{\fAP{fwd}\}$, $\fAP{DNS} =
  \{\fAP{dns}\}$, $r_1$ and $r_2$ two (honest) relying parties,
  $\fAP{fwd}$ an honest forwarder, $\fAP{dns}$ an honest DNS
  server. Let $\fAP{attacker} \in \mathsf{Web}$ be some
  web attacker. Let
  $\mi{dr}$ be a domain of $r_1$ or $r_2$ and $b(\mi{dr})$ a challenge
  browser. Let $\mathsf{Hon}' := \{ b(\mi{dr}) \} \cup \fAP{RP} \cup
  \fAP{FWD} \cup \fAP{DNS}$, $\mathsf{Web}' := \mathsf{Web}$, and
  $\mathsf{Net}' := \emptyset$ (i.e., there is no network attacker).
  Let $\bidsystem' := \mathsf{Hon}' \cup \mathsf{Web}' \cup
  \mathsf{Net}'$.  Let $\scriptset' := \scriptset \setminus
  \{\str{script\_idp}\}$ and $\mathsf{script}'$ be accordingly. We
  call $\spressoprivwebsystem(\mi{dr}) = (\bidsystem', \scriptset',
  \mathsf{script}', E^0, \fAP{attacker})$ an \emph{\spresso web system
    for privacy analysis} iff %
  \gs{$\nf$ hier vielleicht noch nach wichtigkeit sortieren} the
  domain $\str{fwddomain}$ is the only domain assigned to $\fAP{fwd}$,
  the domain $\mi{dr}_1$ the only domain assigned to $r_1$, and
  $\mi{dr}_2$ the only domain assigned to $r_2$. Both, $r_1$ and $r_2$
  are configured to use the forwarder $\fAP{fwd}$, i.e., in their
  state $\str{FWDdomain}$ is set to $\str{fwddomain}$. The browser
  $b(\mi{dr})$ owns exactly one email address and this email address
  is governed by some attacker.  All honest parties (in
  $\mathsf{Hon}$) are not corruptible, i.e., they ignore any
  $\str{CORRUPT}$ message. Identity providers are assumed to be
  dishonest, and hence, are subsumed by the web attackers (which
  govern all identities). In the initial state $s_0^b$ of the (only)
  browser in $\bidsystem'$ and in the initial states $s_0^{r_1}$,
  $s_0^{r_2}$ of both relying parties, the DNS address is
  $\mapAddresstoAP(\fAP{dns})$. Further, $\mi{wkCache}$ in the initial
  states $s_0^{r_1}$, $s_0^{r_2}$ is equal and contains a public key
  for each domain registered in the DNS server (i.e., the relying
  parties already know some public key to verify \spresso identity
  assertions from all domains known in the system\gs{$\nf$ genauer
    formulieren} and they do not have to fetch them from IdP).
\end{definition}

As all parties in an \spresso web system for privacy analysis are
either web attackers, browsers, or deterministic processes and all
scripting processes are either the attacker script or deterministic,
it is easy to see that in \spresso web systems for privacy analysis with configuration $(S,E,N)$ a command $\zeta$ induces at most one processing step. We further note that, under a given infinite sequence of nonces $N^0$, all schedules $\sigma$ induce at most one run $\rho = ((S^0,E^0,N^0),\dots,(S^i,E^i,N^i),\dots,(S^{|\sigma|},E^{|\sigma|},N^{|\sigma|}))$ as all of its commands induce at most one processing step for the $i$-th configuration.

We will now define our privacy
property for \spresso:

\begin{definition}[IdP-Privacy]\label{def:idp-privacy}\gs{$\downarrow$ muessen wir hier $\mi{dr}_1$ einfueheren?}
  Let 
  \begin{align*}
    \spressoprivwebsystem_1 := \spressoprivwebsystem(\mi{dr}_1) =
    (\bidsystem_1, \scriptset, \mathsf{script}, E^0, \fAP{attacker}_1)&\text{ and}\\
    \spressoprivwebsystem_2 := \spressoprivwebsystem(\mi{dr}_2) =
    (\bidsystem_2, \scriptset, \mathsf{script}, E^0, \fAP{attacker}_2)&
  \end{align*}
  be \spresso web systems for privacy analysis.  Further, we require
  $\fAP{attacker}_1 = \fAP{attacker}_2 =: \fAP{attacker}$ and for $b_1
  := b(\mi{dr}_1)$, $b_2 := b(\mi{dr}_2)$ we require $S(b_1) = S(b_2)$
  and $\bidsystem_1 \setminus \{b_1\} = \bidsystem_2 \setminus
  \{b_2\}$ (i.e., the web systems are the same up to the parameter of
  the challenge browsers).  We say that $\spressoprivwebsystem$ is
  \emph{IdP-private} iff $\spressoprivwebsystem_1$ and
  $\spressoprivwebsystem_2$ are indistinguishable.
\end{definition}

\subsection{Definition of Equivalent Configurations}
\label{app:defin-equiv-stat}

Let $\spressoprivwebsystem_1 = (\bidsystem_1, \scriptset,
\mathsf{script}, E^0, \fAP{attacker})$ and $\spressoprivwebsystem_2 =
(\bidsystem_2, \scriptset, \mathsf{script}, E^0, \fAP{attacker})$ be
\spresso web systems for privacy analysis. Let $(S_1,E_1,N_1)$ be a
configuration of $\spressoprivwebsystem_1$ and $(S_2,E_2,N_2)$ be a
configuration of $\spressoprivwebsystem_2$.

\begin{definition}[Proto-Tags]
  We call a term of the form $\encs{\an{y,n}}{k}$ with the variable
  $y$ as a placeholder for a domain, and $n$ and $k$ some nonces a
  \emph{proto-tag}.
\end{definition}

\begin{definition}[Term Equivalence up to Proto-Tags]
  Let $\theta = \{a_1, \ldots, a_l \}$ be a finite set of proto-tags.
  Let $t$ and $t'$ be terms. We call $t_1$ and $t_2$
  \emph{term-equivalent under a set of proto-tags $\theta$} iff there
  exists a term $\tau \in \terms(\{x_1,\dots,x_l\})$ such that
  $t_1 = (\tau [ a_1 / x_1 , \dots , a_l / x_l ])[ \mi{dr}_1 / y ]$
  and
  $t_2 = (\tau [ a_1 / x_1 , \dots , a_l / x_l ])[ \mi{dr}_2 / y ]$.
  We write $t_1 \prototagequiv{\theta} t_2$.

  We say that two finite sets of terms $D$ and $D'$ are
  \emph{term-equivalent under a set of proto-tags $\theta$} iff
  $|D| = |D'|$ and, given a lexicographic ordering of the elements in
  $D$ of the form $(d_1,\dots,d_{|D|})$ and the elements in $D'$ of
  the form $(d'_1,\dots,d_{|D'|})$, we have that for all
  $i \in \{1,\dots,|D|\}$: $d_i \prototagequiv{\theta} d'_i$. We then
  write $D \prototagequiv{\theta} D'$.
\end{definition}

\begin{definition}[Equivalence of HTTP Requests]
  Let $m_1$ and $m_2$ be (potentially encrypted) HTTP requests and
  $\theta = \{a_1, \ldots, a_l \}$ be a finite set of proto-tags. We
  call $m_1$ and $m_2$ \emph{$\delta$-equivalent under a set of
    proto-tags $\theta$} iff $m_1 \prototagequiv{\theta} m_2$ or all
  subterms are equal with the following exceptions:
  \begin{enumerate}
  \item the Host value and the Origin/Referer headers in both requests
    are the same except that the domain $\mi{dr}_1$ in $m_1$ can be
    replaced by $\mi{dr}_2$ in $m_2$,
  \item the HTTP body $g_1$ of $m_1$ and the HTTP body $g_2$ of $m_2$
    are (I) term-equivalent under $\theta$, (II) for $j\in \{1,2\}$ if
    $g_j[\str{eia}] \sim \encs{\sig{ \an{\encs{\an{\mi{dr}_j,*}}{*},
          *, \str{fwddomain}} }{*}}{*} $
    and the origin (HTTP header) of HTTP message in $m_j$ is
    $\an{\mi{dr}_j,\https}$ then the receiver of this message is
    $r_j$, and (III) if $g_1$ contains a dictionary key
    $\str{loginSessionToken}$ then there exists an $l' \in L$ such
    that $g_1[\str{loginSessionToken}] \equiv l'$, and
  \item if $m_1$ is an encrypted HTTP request then and only then $m_2$
    is an encrypted HTTP request and the keys used to encrypt the
    requests have to be the correct keys for $\mi{dr}_1$ and
    $\mi{dr}_2$ respectively.
  \end{enumerate}
  We write $m_1 \httptagequiv{\theta} m_2$.

\end{definition}

\gs{$\nf$ check if we have $\fAP{FWD}$ for forwarder defined}

\gs{$\nf$ wir verwenden im Folgenden $\fAP{dns}$, $b$, etc.. sind
  diese hinreichend definiert?}

\begin{definition}[Extracting Entries from Login Sessions]\gs{$\nf$ todo: find name}
  Let $t_1$, $t_2$ be dictionaries over $\nonces$ and $\terms$,
  $\theta$ be a finite set of proto-tags, and $d$ a domain. We call
  $t_1$ and $t_2$ \emph{$\eta$-equivalent} iff $t_2$ can be
  constructed from $t_1$ as follows: For every proto-tag
  $a \in \theta$, we remove the entry identified by the dictionary key
  $i$ for which it holds that $\proj{4}{t_1[i]} \equiv a[ d / y]$, if
  any. We denote the set of removed entries by $D$. We write
  $\logsessminus{t_1}{t_2}{\theta}{d}{D}$.
\end{definition}

\begin{definition}
  Let $a$ be a proto-tag, $S_1$ and $S_2$ be states of \spresso web
  systems for privacy analysis, and $l$ a nonce. We call $l$ a login
  session token for the proto-tag $a$, written
  $l \in \mathsf{loginSessionTokens}(a,S_1,S_2)$ iff for any
  $i \in \{1,2\}$ and any $j \in \{1,2\}$ we have that
  $\proj{4}{S_i(r_j).\str{loginSessions}[l]} = a[\mi{dr}_j/y]$.
\end{definition}

\begin{definition}[Equivalence of States]\label{def:eq-of-states}
  Let $\theta$ be a set of proto-tags and $H$ be a set of nonces. Let
  $K := \{ k \mid \exists\,n: \encs{\an{y,n}}{k} \in \theta\}$. We
  call $S_1$ and $S_2$ \emph{$\gamma$-equivalent under
    $(\theta, H)$}\gs{schreibweise uebernehmen} iff the following
  conditions are met:
  \begin{enumerate}
  \item\label{eqs:fwd} $S_1(\fAP{fwd}) = S_2(\fAP{fwd})$, and
  \item\label{eqs:dns} $S_1(\fAP{dns}) = S_2(\fAP{dns})$, and
  \item\label{eqs:r1} $S_1(\fAP{r_1})$ equals $S_2(\fAP{r_1})$ except
    for the subterms $\str{pendingDNS}$, $\str{loginSessions}$ and
    $\str{serviceTokens}$, and
  \item\label{eqs:r2} $S_1(\fAP{r_2})$ equals $S_2(\fAP{r_2})$ except
    for the subterms $\str{pendingDNS}$, $\str{loginSessions}$ and
    $\str{serviceTokens}$, and
  \item\label{eqs:logsess} for two sets of terms $D$ and $D'$:
    $\logsessminus{S_1(\fAP{r_1}).\str{loginSessions}}{S_2(\fAP{r_1}).\str{loginSessions}}{\theta}{\mi{dr}_1}{D}$,
    $\logsessminus{S_2(\fAP{r_2}).\str{loginSessions}}{S_1(\fAP{r_2}).\str{loginSessions}}{\theta}{\mi{dr}_2}{D'}$,
    and $D \prototagequiv{\theta} D'$, and
  \item\label{eqs:rp-invalid-dns} for all entries $x$ in the subterms
    $\str{pendingDNS}$ of $S_1(\fAP{r_1})$, $S_1(\fAP{r_1})$,
    $S_1(\fAP{r_1})$, and $S_1(\fAP{r_1})$ it holds true that
    $\pi_2{x}.\str{host}$ is not a domain name known to the DNS
    server, and
  \item\label{eqs:rp-no-pendingRequests} the subterms
    $\str{pendingRequest}$ of $S_1(\fAP{r_1})$, $S_1(\fAP{r_2})$,
    $S_2(\fAP{r_1})$, and $S_2(\fAP{r_2})$ are $\an{}$, and
  \item\label{eqs:rp-wkCache} the subterm $\str{wkCache}$ of
    $S_1(\fAP{r_1})$, $S_1(\fAP{r_2})$, $S_2(\fAP{r_1})$, and
    $S_2(\fAP{r_2})$ are equal and contain a public key for each
    domain registered in the DNS server, and
  \item\label{eqs:att-not-k} $\forall k \in K$:
    $k \not\in d_\emptyset(\bigcup_{i \in \{1,2\},\ A\, \in\,
      \mathsf{Web}\, \cup \,\mathsf{Net}\, \cup\, \{\mathsf{dns}, \mathsf{fwd}\}}S_i(A))$
  \item\label{eqs:att} for each attacker $A$:
    $S_1(A) \prototagequiv{\theta} S_2(A)$, and
  \item\label{eqs:att-not-l} for all $a\in\theta$ and all attackers $A$ we have that
    $\nexists\ l \in \mathsf{loginSessionTokens}(a,S_1,S_2)$ such that
    $l$ is a subterm of $S_1(A)$ or $S_2(A)$.
  \item\label{eqs:b} $S_1(b_1)$ equals $S_2(b_2)$ except for for the
    subterms $\str{challenge}$, $\str{pendingDNS}$,
    $\str{pendingRequests}$, $\str{windows}$ and we have that

    \begin{enumerate}
    \item \label{eqs:b:challenge}
      $S_1(b_1).\str{challenge} = \mi{dr}_1 \wedge
      S_2(b_2).\str{challenge} = \mi{dr}_2$
      or $S_1(b_1).\str{challenge} = S_2(b_2).\str{challenge} = \bot$,
      and
    \item\label{eqs:b:pendingDNS}
      $|S_1(b_1).\str{pendingDNS}| = |S_2(b_2).\str{pendingDNS}| =:
      j$,
      for all $i \in \{1, \dots, j\}$,
      $q_1 := \proj{i}{S_1(b_1).\str{pendingDNS}}$,
      $q_2 := \proj{i}{S_2(b_2).\str{pendingDNS}}$ we have that
      $\proj{1}{q_1} = \proj{1}{q_2} \in \nonces$ and for
      $v_1 := \proj{2}{q_1}$ and $v_2 := \proj{2}{q_2}$:
      \begin{enumerate}
      \item $\proj{1}{v_1} = \proj{1}{v_2}$, and
      \item $\proj{3}{v_1} = \proj{3}{v_2}$, and
      \item $\proj{1}{v_1}$ is either a nonce ($\in \nonces$) or a
        term of the form $\an{x,y}$ with $x \in \nonces$ a nonce and
        $y \in \nonces \cup \{\bot\}$ a nonce or $\bot$, and 
      \item if
        $\proj{2}{v_1}.\str{host} = \mi{dr}_1 \wedge
        \proj{2}{v_2}.\str{host} = \mi{dr}_2$,\\then
        $\proj{2}{v_1} \httptagequiv{\theta}
        \proj{2}{v_2}$ $\wedge$ $\proj{2}{v_1}.\str{nonce} \in H$,\\
        else $\proj{2}{v_1} \prototagequiv{\theta} \proj{2}{v_2}$
        $\wedge$ $\proj{2}{v_1}.\str{nonce} \not\in H$ $\wedge$ $\nexists\, l \in L$ such that $l$ is a subterm of $\proj{2}{v_1}$,
      \end{enumerate}
and
\item\label{eqs:b:pendingRequests}
  $|S_1(b_1).\str{pendingRequests}| = |S_2(b_2).\str{pendingRequests}|
  =: j$,
  for all $i \in \{1, \dots, j\}$,
  $v_1 := \proj{i}{S_1(b_1).\str{pendingRequests}}$,
  $v_2 := \proj{i}{S_2(b_2).\str{pendingRequests}}$ we have that

  \begin{enumerate}
  \item $\proj{1}{v_1} = \proj{1}{v_2}$, and
  \item $\proj{3}{v_1} = \proj{3}{v_2}$, and
  \item $\proj{1}{v_1}$ is either a nonce ($\in \nonces$) or a term of
    the form $\an{x,y}$ with $x \in \nonces$ a nonce and
    $y \in \nonces \cup \{\bot\}$ a nonce or $\bot$, and 
  \item if
    $\proj{2}{v_1}.\str{host} = \mi{dr}_1 \wedge
    \proj{2}{v_2}.\str{host} = \mi{dr}_2$,\\then
    $\proj{2}{v_1} \httptagequiv{\theta}
    \proj{2}{v_2}$ $\wedge$ $\proj{2}{v_1}.\str{nonce} \in H$ $\wedge$ $\proj{4}{v_1} \in \mapAddresstoAP(r_1)$ $\wedge$ $\proj{4}{v_2} \in \mapAddresstoAP(r_2)$,\\
    else $\proj{2}{v_1} \prototagequiv{\theta} \proj{2}{v_2}$ $\wedge$
    $\proj{2}{v_1}.\str{nonce} \not\in H$ $\wedge$
    $\proj{4}{v_1} = \proj{4}{v_2}$ $\wedge$ $\nexists\, l \in L$ such that $l$ is a subterm of $\proj{2}{v_1}$, 
  \end{enumerate}
  and
\item\label{eqs:b:no-k-in-pending-something} there is no $k \in K$
  such that
  \begin{align*}
    k \in d_{\nonces \setminus \{k\}}(\{&S_1(b_1).\str{pendingRequests}, S_2(b_2).\str{pendingRequests},\\ & S_1(b_1).\str{pendingDNS}, S_2(b_2).\str{pendingDNS}\})
  \end{align*}
  (i.e., $k$ cannot be derived from these terms by any party unless it knows $k$), and
\item $S_1(b_1).\str{windows}$ equals $S_2(b_2).\str{windows}$ with
  the exception of the subterms $\str{location}$, $\str{referrer}$,
  $\str{scriptstate}$, and $\str{scriptinputs}$ of some document terms
  pointed to by
  $\mathsf{Docs}^+(S_1(b_1)) = \mathsf{Docs}^+(S_2(b_2)) =: J$. For
  all $j \in J$ we have that: \label{eqs:b:w}\gs{$\nf$ do we need
    normal form in the following?}
  \begin{enumerate}
  \item there is no $k \in K$ such that
    \begin{align*}
    k \in d_{\nonces \setminus \{k\}}(\{&S_1(b_1).j.\str{location}
    ,  S_2(b_2).j.\str{location},\\ & S_1(b_1).j.\str{referrer} , 
    S_2(b_2).j.\str{referrer}\})
    \end{align*}

  \item if
    $S_1(b_1).j.\str{origin} \in \{\an{\mi{dr}_1,
      \https},\an{\mi{dr}_2, \https}\}$
    then
    $S_1(b_1).j.\str{script} \in \{\str{script\_rp},
    \str{script\_rp\_redir}\}$, and
  \item if
    $S_1(b_1).j.\str{origin} \equiv \an{\str{fwddomain}, \https}$ then
    $S_1(b_1).j.\str{script} \equiv \str{script\_fwd}$, and
  \item\label{eqs:b:w:script_rp} if
    $S_1(b_1).j.\str{origin} \in \{\an{\mi{dr}_1,
      \https},\an{\mi{dr}_2, \https}\}$
    and $S_1(b_1).j.\str{script} \equiv \str{script\_rp}$ then \
    \begin{enumerate}
    \item $S_1(b_1).j.\str{location}$ and $S_2(b_2).j.\str{location}$
      are term-equivalent under $\theta$ except for the host part,
      which is either equal or $\mi{dr}_1$ in $b_1$ and $\mi{dr}_2$ in
      $b_2$, and
    \item $S_1(b_1).j.\str{referrer}$ and $S_2(b_2).j.\str{referrer}$
      are term-equivalent under $\theta$ except for the host part,
      which is either equal or $\mi{dr}_1$ in $b_1$ and $\mi{dr}_2$ in
      $b_2$, and
    \item
      $S_1(b_1).j.\str{scriptstate} \prototagequiv{\theta}
      S_2(b_2).j.\str{scriptstate}$ and if $\exists\, l \in L$ such that $l$ is a subterm of $S_1(b_1).j.\str{scriptstate}$, then $S_1(b_1).j.\str{location}.\str{host} \equiv \mi{dr}_1$ and $S_2(b_2).j.\str{location}.\str{host} \equiv \mi{dr}_2$, and
    \item for $p \in \{$
      \begin{align*}
        & \an{\tXMLHTTPRequest,*,*},\\
        & \an{\tPostMessage,*,\an{\str{fwddomain}, \https},\str{ready}},\\
        & \an{\tPostMessage,*,\an{\str{fwddomain},
          \https},\an{\str{eia},*}}
      \end{align*}
      $\}$ we have
      $S_1(b_1).j.\str{scriptinputs} |\, p \prototagequiv{\theta}
      S_2(b_2).j.\str{scriptinputs} |\, p$, and
    \item if $\exists\, l \in L$ such that $l$ is a subterm of
      $S_1(b_1).j.\str{scriptinputs}$, then
      $S_1(b_1).j.\str{location}.\str{host} \equiv \mi{dr}_1$ and
      $S_2(b_2).j.\str{location}.\str{host} \equiv \mi{dr}_2$, and
    \item $\forall k \in K$: $k$ is not contained in any subterm of
      $S_1(b_1).j.\str{scriptstate}$ except for
      $S_1(b_1).j.\str{scriptstate}.\str{tagKey}$, and
      \begin{itemize}
      \item
        $S_1(b_1).j.\str{origin} \not\equiv
        \an{\mi{dr}_1,\https}$\\$\implies
        k \not\equiv S_1(b_1).j.\str{scriptstate}.\str{tagKey}$, and
      \item $S_1(b_1).j.\str{origin}
        \not\equiv
        \an{\mi{dr}_1,\https}$\\$\implies k \not\in
        d_\emptyset(S_1(b_1).j.\str{scriptinputs})$, and
      \item
        $S_2(b_2).j.\str{origin} \not\equiv
        \an{\mi{dr}_2,\https}$\\$\implies
        k \not\equiv S_2(b_2).j.\str{scriptstate}.\str{tagKey}$, and
      \item $S_2(b_2).j.\str{origin}
        \not\equiv
        \an{\mi{dr}_2,\https}$\\$\implies k \not\in
        d_\emptyset(S_2(b_2).j.\str{scriptinputs})$,
        and \end{itemize}\end{enumerate}
  \item\label{eqs:b:w:script_rp_redir} if
    $S_1(b_1).j.\str{origin} \in \{\an{\mi{dr}_1,
      \https},\an{\mi{dr}_2, \https}\}$
    and
    $S_1(b_1).j.\str{script} \not\equiv \str{script\_rp}$\footnote{It
      immediately follows that
      $S_1(b_1).j.\str{script} \equiv \str{script\_rp\_redir}$ in this
      case.} then
    \begin{enumerate}
    \item $S_1(b_1).j.\str{location}$ and $S_2(b_2).j.\str{location}$
      are term-equivalent under $\theta$ except for the host part,
      which is either equal or $\mi{dr}_1$ in $b_1$ and $\mi{dr}_2$ in
      $b_2$, and
    \item $S_1(b_1).j.\str{referrer}$ and $S_2(b_2).j.\str{referrer}$
      are term-equivalent under $\theta$ except for the host part,
      which is either equal or $\mi{dr}_1$ in $b_1$ and $\mi{dr}_2$ in
      $b_2$, and
    \item
      $S_1(b_1).j.\str{scriptstate} \prototagequiv{\theta}
      S_2(b_2).j.\str{scriptstate}$ and if $\exists\, l \in L$ such that $l$ is a subterm of $S_1(b_1).j.\str{scriptstate}$, then $S_1(b_1).j.\str{location}.\str{host} \equiv \mi{dr}_1$ and $S_2(b_2).j.\str{location}.\str{host} \equiv \mi{dr}_2$, and
    \item there is no $k \in K$ such that
      $k \in d_{\nonces \setminus
        \{k\}}(\{S_1(b_1).j.\str{scriptstate}\})$
    \end{enumerate}
  \item\label{eqs:b:w:script_fwd} if
    $S_1(b_1).j.\str{origin} = \an{ \str{fwddomain} , \https}$ then
    \begin{enumerate}
    \item
      $S_1(b_1).j.\str{location} \prototagequiv{\theta}
      S_2(b_2).j.\str{location}$, and
    \item
      $S_1(b_1).j.\str{scriptstate} \prototagequiv{\theta}
      S_2(b_2).j.\str{scriptstate}$, and
    \item for $p = \an{\tPostMessage,*,*,\an{\str{tagKey},*}}$ and
      $x_1 = S_1(b_1).j.\str{scriptinputs}|\,p$ and
      $x_2 = S_2(b_2).j.\str{scriptinputs} |\, p$ we have that for all
      $i \in \{1,\dots,|x|\}$:
      \begin{itemize}
      \item
        $\proj{2}{\proj{i}{x_1}} \prototagequiv{\theta}
        \proj{2}{\proj{i}{x_2}}$, and
      \item $\proj{1}{\proj{3}{\proj{i}{x_1}}} \prototagequiv{\theta}
        \proj{1}{\proj{3}{\proj{i}{x_2}}}$ or\\
        $\proj{1}{\proj{3}{\proj{i}{x_1}}} = \mi{dr}_1 \wedge
        \proj{1}{\proj{3}{\proj{i}{x_2}}} = \mi{dr}_2$, and
      \item
        $\proj{2}{\proj{3}{\proj{i}{x_1}}} \prototagequiv{\theta}
        \proj{2}{\proj{3}{\proj{i}{x_2}}}$, and
      \item
        $\proj{4}{\proj{i}{x_1}} \prototagequiv{\theta}
        \proj{4}{\proj{i}{x_2}}$, and
      \end{itemize}
    \end{enumerate}
  \item\label{eqs:b:w:att_script} if
    $S_1(b_1).j.\str{origin} \not\in
    \{\an{\mi{dr}_1,\https},\an{\mi{dr}_2,\https},\an{\str{fwddomain},\https}\}$
    then
    \begin{enumerate}
    \item
      $S_1(b_1).j.\str{location} \prototagequiv{\theta}
      S_2(b_2).j.\str{location}$, and
    \item
      $S_1(b_1).j.\str{referrer} \prototagequiv{\theta}
      S_2(b_2).j.\str{referrer}$, and
    \item
      $S_1(b_1).j.\str{scriptstate} \prototagequiv{\theta}
      S_2(b_2).j.\str{scriptstate}$, and
    \item
      $S_1(b_1).j.\str{scriptinputs} \prototagequiv{\theta}
      S_2(b_2).j.\str{scriptinputs}$, and
    \item there is no $k \in K$ such that $k \in d_{\nonces \setminus
        \{k\}}(\{S_1(b_1).j.\str{scriptstate},
      S_1(b_1).j.\str{scriptinputs}\})$, and
    \item $\nexists\, l \in L$ such that $l$ is a subterm of
      $S_1(b_1).j.\str{scriptstate}$ or of
      $S_1(b_1).j.\str{scriptinputs}$, and
    \end{enumerate}
  \end{enumerate}

\item\label{eqs:b:misc} for
  $x \in
  \{\str{cookies},\str{localStorage},\str{sessionStorage},\str{sts}\}$
  we have that $S_1(b_1).x \prototagequiv{\theta} S_2(b_2).x$. For the
  domains $\mi{dr}_1$ and $\mi{dr}_2$ there are no entries in the
  subterms $x$.

\end{enumerate}

\end{enumerate}
\end{definition}

\begin{definition}[Equivalence of Events]\label{def:eq-of-events}
  Let $\theta$ be a set of proto-tags, $L$ be a set of login session
  tokens, $H$ be a set of nonces, and
  $K := \{ k \mid \exists\,n: \encs{\an{y,n}}{k} \in \theta\}$. We
  call $E_1 = (e_1^{(1)}, e_2^{(1)}\dots)$ and
  $E_2= (e_1^{(2)}, e_2^{(2)} \dots)$ \emph{$\beta$-equivalent under
    $(\theta, L, H)$}\gs{$\downarrow$ in alpha equiv
    erklaeren, was das genau aussagt?} iff all of the following conditions are
  satisfied for every $i \in \mathbb{N}$:

  \begin{enumerate}
  \item\label{eqe:distinction} One of the following conditions holds
    true:
    \begin{enumerate}
    \item\label{eqe:prototagequiv}
      $e_i^{(1)} \prototagequiv{\theta} e_i^{(2)}$ and if $e_i^{(1)}$
      contains an HTTP(S) message (i.e., HTTP(S) request or HTTP(S)
      response), then the HTTP nonce of this HTTP(S) message is not
      contained in $H$,\gs{$\nf$ kann man noch besser ausformulieren}
      or
    \item\label{eqe:dns-req} $e_i^{(1)}$ is a DNS request from $b_1$
      to $\fAP{dns}$ for $\mi{dr}_1$ and $e_i^{(2)}$ is a DNS request
      from $b_2$ to $\fAP{dns}$ for $\mi{dr}_2$, or
    \item\label{eqe:invalid-dns-req} $e_i^{(1)}$ and $e_i^{(2)}$ are
      both DNS requests from any party except $\fAP{dns}$ addressed to
      $\fAP{dns}$ for a domain unknown to the DNS server, or
    \item\label{eqe:dns-res} $e_i^{(1)}$ is a DNS response from
      $\fAP{dns}$ to $b_1$ for a DNS request for $\mi{dr}_1$ and
      $e_i^{(2)}$ is a DNS response from $\fAP{dns}$ to $b_2$ for a
      DNS request for $\mi{dr}_2$, or
    \item\label{eqe:http-req} $e_i^{(1)}$ is an HTTP request $m_1$
      from $b_1$ to $r_1$ and $e_i^{(2)}$ is an HTTP request $m_2$
      from $b_2$ to $r_2$, $m_1 \httptagequiv{\theta} m_2$, and both
      requests are unencrypted or encrypted (i.e., $m_1$ and $m_2$ are
      the content of the encryption) and $m_1.\str{nonce} \in H$, or
    \item\label{eqe:http-res} $e_i^{(1)}$ is an HTTP(S) response from
      $r_1$ to $b_1$ and $e_i^{(2)}$ is an HTTP(S) response from $r_2$
      to $b_2$, and their HTTP messages $m_1$ (contained in
      $e_i^{(1)}$) and $m_2$ (contained in $e_i^{(1)}$) are the same
      except for the HTTP body $g_1 := m_1.\str{body}$ and the HTTP
      body $g_2 := m_2.\str{body}$ which have to be
      $g_1 \prototagequiv{\theta} g_2$ and $m_1.\str{nonce} \in H$ and
      if $g_1$ contains a dictionary key $\str{loginSessionToken}$
      then there exists an $l' \in L$ such that
      $g_1[\str{loginSessionToken}] \equiv l'$.
    \end{enumerate}
  \item\label{eqe:pre:l} If there exists $l \in L$ such that $l$ is a
    subterm of $e_i^{(1)}$ or $e_i^{(2)}$ then we have that
    $e_i^{(1)}$ is a message from $b_1$ to $r_1$ and $e_i^{(2)}$ is a
    message from $b_2$ to $r_2$ or we have that $e_i^{(1)}$ is a
    message from $r_1$ to $b_1$ and $e_i^{(2)}$ is a message from
    $r_2$ to $b_2$.
  \item\label{eqe:pre:k} If there exists $k \in K$ such that
    $k \in d_{\nonces \setminus \{k\}}(\{e_i^{(1)}, e_i^{(2)}\})$ then
    $e_i^{(1)}$ is an HTTP(S) response from $r_1$ to $b_1$ and
    $e_i^{q(2)}$ is an HTTP(S) response from $r_2$ to $b_2$ and the
    bodies of both HTTP messages are of the form
    $\an{\an{\str{tagKey}, k}, *, *}$.
  \item\label{eqe:pre:fwd-script} If $e_i^{(1)}$ or $e_i^{(2)}$ is an
    encrypted HTTP response with body $g$ from $\fAP{fwd}$, then
    $\proj{1}{g}$ is $\str{script\_fwd}$.
  \item\label{eqe:pre:rp-scripts} If $e_i^{(1)}$ or $e_i^{(2)}$ is an
    HTTP(S) response with body $g$ from a relying party, then it does
    not contain any $\str{Location}$, $\cSTS$ or $\cSetCookie$ header
    and if $\proj{1}{g}$ is a string representing a script, then
    $\proj{1}{g}$ is either $\str{script\_rp}$ or
    $\str{script\_rp\_redir}$.
  \item\label{eqe:pre:no-invalid-dns-res} Neither $e_i^{(1)}$ nor
    $e_i^{(2)}$ are DNS responses from $\fAP{dns}$ for domains unknown
    to the DNS server.
  \item\label{eqe:pre:unencrypted-http} If $e_i^{(1)}$ or $e_i^{(2)}$
    is an unencrypted HTTP response, then the message was sent by some
    attacker.
  \end{enumerate}
\end{definition}

\begin{definition}[Equivalence of Configurations]
  We call $(S_1,E_1,N_1)$ and $(S_2,E_2,N_2)$
  \emph{$\alpha$-equivalent} iff there exists a set of proto-tags
  $\theta$ and a set of nonces $H$ such that $S_1$ and $S_2$ are
  $\gamma$-equivalent under $(\theta,H)$, $E_1$ and $E_2$ are
  $\beta$-equivalent under $(\theta,L,H)$ for
  $L := \bigcup_{a\in\theta} \mathsf{loginSessionTokens}(a,S_1,S_2)$,
  and $N_1 = N_2$.
\end{definition}

\subsection{Privacy Proof}

\begin{theorem}\label{thm:privacy-main}
  Every \spresso web system for privacy analysis is IdP-private.
\end{theorem}

Let
$\spressoprivwebsystem = (\bidsystem, \scriptset, \mathsf{script},
E^0, \fAP{attacker})$
be an \spresso web system for privacy analysis.

To prove Theorem~\ref{thm:privacy-main}, we have to show that the
\spresso web systems $\spressoprivwebsystem_1$ and
$\spressoprivwebsystem_2$ are indistinguishable (according to
Definition~\ref{def:idp-privacy}), where $\spressoprivwebsystem_1$ and
$\spressoprivwebsystem_2$ are defined as follows: Let
$\spressoprivwebsystem_1 = (\bidsystem_1, \scriptset, \mathsf{script},
E^0, \fAP{attacker})$
and
$\spressoprivwebsystem_2 = (\bidsystem_2, \scriptset, \mathsf{script},
E^0, \fAP{attacker})$
with $b_1 := b(\mi{dr}_1) \in \bidsystem_1$ and
$b_2 := b(\mi{dr_2}) \in \bidsystem_2$ challenge browsers. Further, we
require
$\bidsystem \setminus \{b\} = \bidsystem_1 \setminus \{b_1\} =
\bidsystem_2 \setminus \{b_2\}$.
We denote the following processes in $\spressoprivwebsystem$ as in
Definition~\ref{def:spresso-ws-priv}:
\begin{enumerate}[leftmargin=3em]
\item[$\fAP{dns}$] denotes the honest DNS server,
\item[$\fAP{fwd}$] denotes the honest forwarder with domain $\str{fwddomain}$,
\item[$\fAP{r_1}$] denotes the honest relying party with domain $\mi{dr_1}$, and
\item[$\fAP{r_2}$] denotes the honest relying party with domain $\mi{dr_2}$.
\end{enumerate}

Following Definition~\ref{def:indistinguishability}, to show the
indistinguishability of $\spressoprivwebsystem_1$ and
$\spressoprivwebsystem_2$ we show that they are indistinguishable
under all schedules $\sigma$. For this, we first note that for all
$\sigma$, there is only one run induced by each $\sigma$ (as our web
system, when scheduled, is deterministic). We now proceed to
show that for all schedules $\sigma = (\zeta_1,\zeta_2,\dots)$, iff $\sigma$
induces a run $\sigma(\spressoprivwebsystem_1)$ there exists a run
$\sigma(\spressoprivwebsystem_2)$ such that
$\sigma(\spressoprivwebsystem_1) \approx
\sigma(\spressoprivwebsystem_2)$.

We now show that if two configurations are $\alpha$-equivalent, then
the view of the attacker is statically equivalent.

\begin{lemma}
  Let $(S_1,E_1,N_1)$ and $(S_2,E_2,N_2)$ be two $\alpha$-equivalent
  configurations. Then
  $S_1(\fAP{attacker}) \approx S_2(\fAP{attacker})$.
\end{lemma}

\begin{proof}
  From the $\alpha$-equivalence of $(S_1,E_1,N_1)$ and $(S_2,E_2,N_2)$
  it follows that
  $S_1(\fAP{attacker}) \prototagequiv{\theta} S_2(\fAP{attacker})$.
  From Condition~\ref{eqs:att-not-k} for $\gamma$-equivalence it
  follows that
  $\{ k \mid \exists\,n: \encs{\an{y,n}}{k} \in \theta\} \cap
  d_\emptyset(\bigcup_{i \in \{1,2\},\ A\, \in\, \mathsf{Web}\, \cup
    \,\mathsf{Net}\}}S_i(A))$
  (i.e., the attacker does not know any keys for the tags contained in
  its view), and therefore it is easy to see that the views are
  statically equivalent.\qed
\end{proof}

We now show that
$\sigma(\spressoprivwebsystem_1) \approx
\sigma(\spressoprivwebsystem_2)$
by induction over the length of $\sigma$. We first, in Lemma~\ref{lemma:initial-config-private}, show that
$\alpha$-equivalence (and therefore, indistinguishability of the views
of $\fAP{attacker}$) holds for the initial configurations of
$\spressoprivwebsystem_1$ and $\spressoprivwebsystem_2$. We then, in Lemma~\ref{lemma:step-config-private}, show
that for each configuration induced by a processing step in $\zeta$,
$\alpha$-equivalence still holds true.

\begin{lemma}\label{lemma:initial-config-private}
  The initial configurations $(S_1^0,E^0,N^0)$ of
  $\spressoprivwebsystem_1$ and $(S_2^0,E^0,N^0)$ of
  $\spressoprivwebsystem_2$ are $\alpha$-equivalent.
\end{lemma}

\begin{proof}
  We now have to show that there exists a set of proto-tags $\theta$ and a set of nonces $H$
  such that $S_1^0$ and $S_2^0$ are $\gamma$-equivalent under
  $(\theta,H)$, $E_1^0 = E^0$ and $E_2^0 = E^0$ are $\beta$-equivalent
  under $(\theta,L,H)$ with $L := \bigcup_{a\in\theta} \mathsf{loginSessionTokens}(a,S_1,S_2)$, and $N_1^0 = N_2^0 = N^0$.

  Let $\theta = H = L = \emptyset$. Obviously, both latter conditions are
  true. For all parties $p \in \bidsystem_1 \setminus \{b_1\}$, it is
  clear that $S_1^0(p) = S_2^0(p)$. Also the states $S_1^0(b_1)$ and
  $S_2^0(b_2)$ are equal. Therefore, all conditions
  of Definition~\ref{def:eq-of-states} are fulfilled. Hence, the
  initial configurations are $\alpha$-equivalent. \qed
\end{proof}

\begin{lemma}\label{lemma:step-config-private}
  Let $(S_1,E_1,N_1)$ and $(S_2,E_2,N_2)$ be two $\alpha$-equivalent
  configurations of $\spressoprivwebsystem_1$ and
  $\spressoprivwebsystem_2$, respectively. Let
  $\zeta =
  \an{\mi{ci},\mi{cp},\tau_\text{process},\mi{cmd}_\text{switch},\mi{cmd}_\text{window},\tau_\text{script},\mi{url}}$
  be a web system command. Then, $\zeta$ induces a processing step in
  either both configurations or in none. In the latter case, let
  $(S_1',E_1',N_1')$ and $(S_2',E_2',N_2')$ be configurations induced
  by $\zeta$ such that
  \[
  (S_1,E_1,N_1) \xrightarrow{\zeta} (S_1',E_1',N_1') \quad \text{and}
  \quad (S_2,E_2,N_2) \xrightarrow{\zeta} (S_2',E_2',N_2') \ .
  \]
  Then, $(S_1',E_1',N_1')$ and $(S_2',E_2',N_2')$ are
  $\alpha$-equivalent.

\end{lemma}

\begin{proof}
  Let $\theta$ be a set of proto-tags and $H$ be a set of nonces for
  which $\alpha$-equivalence of $(S_1,E_1,N_1)$ and $(S_2,E_2,N_2)$
  holds and let
  $L := \bigcup_{a\in\theta} \mathsf{loginSessionTokens}(a,S_1,S_2)$,
  $K := \{ k \mid \exists\,n: \encs{\an{y,n}}{k} \in \theta\}$.

  To induce a processing step, the $\mi{ci}$-th message from $E_1$ or
  $E_2$, respectively, is selected. Following
  Definition~\ref{def:eq-of-events}, we denote these messages by
  $e_i^{(1)}$ or $e_i^{(2)}$, respectively. We now differentiate
  between the receivers of the messages.

  We first note that due to the $\alpha$-equivalence, $\zeta$ either
  induces a processing step in both configurations or in none. We have
  to analyze the conditions stated in
  Corollary~\ref{corr:processingstep-not-induced}: The number of
  waiting events is the same in both configurations, and therefore,
  $(\mi{ci} > |E_1|) \iff (\mi{ci} > |E_2|)$. Further, the set of
  processes is the same (except for the browsers, which are exchanged
  but have the same IP addresses). Also, we have no processes that
  share IP addresses within each system. Therefore, if $\mi{cp} \neq
  1$ then it refers to no process in both runs (and no processing step
  can be induced).  As we show below, if $e_i^{(1)}$ is delivered to
  $b_1$ then and only then $e_i^{(2)}$ is delivered to
  $b_2$. Additionally, the window structure in both browsers is the
  same, and therefore, $\mi{cmd}_\text{window}$ either refers to a
  window that exists in both configurations or in none. There are no
  other cases that induce no processing step in either system.

  We denote the induced processing steps by
  \begin{align*}
    (S_1,E_1,N_1) \xrightarrow[p_1 \rightarrow E^{(1)}_{\text{out}}]{\an{a_1,f_1,m_1} \rightarrow p_1} (S_1', E_1', N_1') &\ \text{and} \\
    (S_2,E_2,N_2) \xrightarrow[p_2 \rightarrow E^{(2)}_{\text{out}}]{\an{a_2,f_2,m_2} \rightarrow p_2} (S_2', E_2', N_2') &\ .
  \end{align*}

  \paragraph{\underline{Case $p_1 = \fAP{fwd}$:}}
  We know that one of the cases of Case~\ref{eqe:distinction} of
  Definition~\ref{def:eq-of-events} must apply for $e_i^{(1)}$ and
  $e_i^{(2)}$. Out of these cases only Case~\ref{eqe:prototagequiv}
  applies. Hence, $p_2 = \fAP{fwd}$.

  In the forwarder relation (Algorithm~\ref{alg:spresso-fwd}), either
  Lines~\ref{line:fwd-send-response}f. are executed in both processing
  steps or in none. It is easy to see that
  $E^{(1)}_\text{out} \prototagequiv{\theta} E^{(2)}_\text{out}$
  (containing at most one event). For this new event all cases of
  Definition~\ref{def:eq-of-events} except for Cases~\ref{eqe:pre:l}
  and~\ref{eqe:distinction} hold trivially true. 

  (*): As both events are static except for IP addresses, the HTTP
  nonce, and the HTTPS key, there is no $k$ contained in the input
  messages or in the state of $\fAP{fwd}$ (except potentially in tags,
  from where it cannot be extracted), and the output messages are sent
  to $f_1$ or $f_2$, respectively, they cannot contain any $l \in L$
  or $k \in K$. Hence, Case~\ref{eqe:pre:l} of
  Definition~\ref{def:eq-of-events} holds true.

  Both output events are constructed exactly the same out of their
  respective input events and Case~\ref{eqe:prototagequiv} applies for
  the output events.

  Therefore, $E_1'$ and $E_2'$ are $\beta$-equivalent under
  $(\theta,H,L)$.  As there are no changes to any state, we have that
  $S_1'$ and $S_2'$ are $\gamma$-equivalent under $(\theta,H)$. No new
  nonces are chosen, hence, $N_1 = N'_1 = N_2 = N'_2$.

  \paragraph{\underline{Case $p_1 = \fAP{dns}$:}}
  In this case, only Cases~\ref{eqe:prototagequiv}, \ref{eqe:dns-req}
  and~\ref{eqe:invalid-dns-req} of Definition~\ref{def:eq-of-events}
  can apply. Hence, $p_2 = \fAP{dns}$. We note that (*) applies
  analogously in all cases.

  In the first case, it is easy to see that $E^{(1)}_\text{out}
  \prototagequiv{\theta} E^{(2)}_\text{out}$.  In the second case, it
  is easy to see that the DNS server only outputs empty events in both
  processing steps.  In the third case, $E^{(1)}_\text{out}$ and
  $E^{(2)}_\text{out}$ are such that Case~\ref{eqe:dns-res} of
  Definition~\ref{def:eq-of-events} applies.

  Therefore, $E_1'$ and $E_2'$ are $\beta$-equivalent under
  $(\theta,H,L)$ in all three cases.  As there are no changes to any
  state in all cases, we have that $S_1'$ and $S_2'$ are
  $\gamma$-equivalent under $(\theta,H)$. No new nonces are chosen,
  hence, $N_1 = N'_1 = N_2 = N'_2$.

  \paragraph{\underline{Case $p_1 = \fAP{r_1}$:}} First, we consider cases that
  can never happen or are ignored in both processing steps. After
  this, we distinct several cases of HTTPS requests.

  If $e^{(1)}$ is a DNS response, we know that $e_i^{(1)} \prototagequiv{\theta}
  e_i^{(2)}$, which implies $p_2 = \fAP{r_1}$. Only DNS responses from
  $\fAP{dns}$ are processed by a relying party, other DNS responses
  are dropped without any state change. As the state of a relying
  party fulfills Condition~\ref{eqs:rp-invalid-dns} of
  Definition~\ref{def:eq-of-states} (RPs only query domains unknown to
  $\fAP{dns}$) and both $e_i^{(1)}$ and $e_i^{(2)}$ fulfill
  Condition~\ref{eqe:pre:no-invalid-dns-res} of
  Definition~\ref{def:eq-of-events} (there are no DNS responses from
  $\fAP{dns}$ about domains unknown to $\fAP{dns}$), we have a
  contradiction. Hence, $e_i^{(1)}$ cannot be a DNS response.

  If $e^{(1)}$ is an HTTP response, we know that $e_i^{(1)}
  \prototagequiv{\theta} e_i^{(2)}$, which implies $p_2 = \fAP{r_1}$. From
  Condition~\ref{eqs:rp-no-pendingRequests} of
  Definition~\ref{def:eq-of-states}, we know that relying parties
  always drop HTTP responses (without any state change).

  If $e_i^{(1)}$ is any other message that is not a (properly) encrypted
  HTTP request, we have that $e_i^{(1)} \prototagequiv{\theta} e_i^{(2)}$, which
  implies $p_2 = \fAP{r_1}$. The relying party drops such messages in
  both processing steps (without any state change).

  For the following, we note that a relying party never sends
  unencrypted HTTP responses.

  There are four possible types of HTTP requests that are accepted by
  $r_1$ in Algorithm~\ref{alg:rp-spresso}:
  \begin{itemize}
  \item $\mi{path} = \str{/}$ (index page),
    Line~\ref{line:serve-rp-index},
  \item $\mi{path} = \str{/startLogin}$ (start a login),
    Line~\ref{line:serve-rp-start-login},
  \item $\mi{path} = \str{/redir}$ (redirect to IdP),
    Line~\ref{line:serve-rp-redir}, and
  \item $\mi{path} = \str{/login}$ (login),
    Line~\ref{line:assemble-login-response}.
    \end{itemize}

  From the cases in Definition~\ref{def:eq-of-events}, only two can
  possibly apply here: Case~\ref{eqe:prototagequiv} and
  Case~\ref{eqe:http-req}. For both cases, we will now analyze each of
  the HTTP requests listed above separately.

  \noindent \emph{Definition~\ref{def:eq-of-events},
    Case~\ref{eqe:prototagequiv}:} $e_i^{(1)} \prototagequiv{\theta}
  e_i^{(2)}$. This case implies $p_2 = \fAP{r_1} = p_1$. As we see
  below, for the output events $E^{(1)}_\text{out}$ and
  $E^{(2)}_\text{out}$ (if any) only Case~\ref{eqe:prototagequiv} of
  Definition~\ref{def:eq-of-events} applies. This implies that the
  output events may not contain any HTTP nonce contained in $H$. As we
  know that the HTTP nonce of the incoming HTTP requests is not
  contained in $H$ and the output HTTP responses (if any) of the RP
  reuses the same HTTP nonce, the nonce of the HTTP responses cannot
  be in $H$.

  \begin{itemize}
  \item $\mi{path} = \str{/}$. In this case, the same output event is
    produced, i.e. $E^{(1)}_\text{out} = E^{(2)}_\text{out}$, and
    Condition~\ref{eqe:pre:rp-scripts} of
    Definition~\ref{def:eq-of-events} holds true. Also, (*)
    applies. The remaining conditions are trivially fulfilled and
    $E_1'$ and $E_2'$ are $\beta$-equivalent under $(\theta,H,L)$. As
    there are no changes to any state, we have that $S_1'$ and $S_2'$
    are $\gamma$-equivalent under $(\theta,H)$. No new nonces are
    chosen, hence, $N_1 = N'_1 = N_2 = N'_2$.
  \item $\mi{path} = \str{/startLogin}$. The domains of the email
    addresses in both message bodies are either equivalent and
    registered to the DNS server (and hence, $\str{wkCache}$ contains
    a public key for this domain), or they are not contained in
    $\str{wkCache}$ (in both, $S_1(r_1)$ and $S_2(r_1)$). If they are
    unknown (i.e., not contained in $\str{wkCache}$), they are not
    registered in the DNS server. Nonetheless, in this case a DNS
    request is sent to $\fAP{dns}$. Then, the terms
    $E^{(1)}_\text{out}$ and $E^{(2)}_\text{out}$ contain a request
    matching Case~\ref{eqe:prototagequiv} of
    Definition~\ref{def:eq-of-events}. As $E^{(1)}_\text{out}$ and
    $E^{(2)}_\text{out}$ are constructed such that besides IP
    addresses, a string, and a nonce, they only contain a term derived
    from the input events. In particular, they contain no $k \in K$ or
    $l \in L$ (**): As Condition~\ref{eqe:pre:l} of
    Definition~\ref{def:eq-of-events} applies for the input events,
    this condition also applies for the output events. Thus, $E_1'$
    and $E_2'$ are $\beta$-equivalent under $(\theta,H,L)$. The states
    $S_1'(r_1)$ is equal to $S_1(r_1)$ up to the subterm
    $\str{pendingDNS}$, and $S_2'(r_1)$ is equal to $S_2(r_1)$ up to
    the subterm $\str{pendingDNS}$. The subterm $\str{pendingDNS}$
    only contains a new entry for a domain unknown to the DNS server.
    Hence, Condition~\ref{eqs:rp-invalid-dns} of
    Definition~\ref{def:eq-of-states} holds. Thus, we have that $S_1'$
    and $S_2'$ are $\gamma$-equivalent under $(\theta,H)$. Exactly one
    nonce is chosen in both processing steps, and therefore
    $N_1' = N_2'$.

    If the domains of the email addresses are valid and registered in
    $\str{wkCache}$, then $\mathsf{SENDSTARTLOGINRESPONSE}$ is called.
    In both processing steps, a tag is constructed exactly the same.
    The same HTTP response (which does not contain a $k \in K$ or a
    $l \in L$) is put in both $E^{(1)}_\text{out}$ and
    $E^{(2)}_\text{out}$. The first element of the response's body is
    not a string and therefore Condition~\ref{eqe:pre:rp-scripts}
    holds true. The tag is only created on $r_1$ in both runs and
    hence, $\theta$ does not have to be altered. Analogously to (**)
    we have that $E_1'$ and $E_2'$ are $\beta$-equivalent under
    $(\theta,H,L)$. The subterm $\str{loginSessions}$ of the state of
    $r_1$ is extended exactly the same. Thus, we have that $S_1'$ and
    $S_2'$ are $\gamma$-equivalent under $(\theta,H)$. In both
    processing steps exactly four nonces are chosen, and we have that
    $N_1' = N_2'$.

  \item $\mi{path} = \str{/redir}$. First, we note that there is no
    $l \in L$ contained in either $m_1$ or $m_2$ (by the Defintion of
    $\beta$-equivalence). We further note that a relying party that
    receives an (encrypted) HTTP request for the path $\str{/redir}$
    either (I) stops in Line~\ref{line:rp-redir-stop} of
    Algorithm~\ref{alg:rp-spresso} with an empty output or (II) emits
    an HTTP response in Line~\ref{line:rp-redir-response}. The state
    of the relying party is not changed in either case.

    We now show that in both processing steps always the same cases
    apply. For $i \in \{1,2\}$ and $\mi{body}_i$ the body of the HTTPS
    request $m_i$, Case~(II) applies for $r_1$
    iff
    \[S_i(r_1).\str{loginSessions}[\mi{body}_i[\str{loginSessionToken}]]
    \not\equiv \an{}\]
    (see Lines~\ref{line:rp-redir-start}ff. of
    Algorithm~\ref{alg:rp-spresso}).

    From Condition~\ref{eqs:logsess} of
    Definition~\ref{def:eq-of-states}, we know that
    $S_2(r_1).\str{loginSessions}$ can be constructed from
    $S_1(r_1).\str{loginSessions}$ without removing the entry with the
    dictionary key $\mi{body}_1[\str{loginSessionToken}]$ (as this key
    is not in $L$). Thus, both dictionaries either contain the same
    entry for the dictionary key
    $\mi{body}_1[\str{loginSessionToken}$] or they both contain no
    such entry. Hence, Case~(II) applies in both processing steps or
    in none.

    In both cases, exactly the same outputs are emitted (without
    containing any $l\in L$ or $k \in K$) and no state is changed and no new nonces
    are chosen. In both cases, the first element of the response body
    (if any) is $\str{script\_rp\_redir}$. We therefore trivially have
    $\alpha$-equivalence of the new configurations.

  \item $\mi{path} = \str{/login}$. This case can be handled
    analogously to the previous case with two exceptions:

    (A) First, there are two additional checks, the first in
    Line~\ref{line:check-origin-header} of
    Algorithm~\ref{alg:rp-spresso} and the second in
    Line~\ref{line:rp-check-ia}. We have to show that both checks each
    either simultaneously succeed or fail in both cases.

    For the first check, it is easy to see that this follows from
    $m_1 \prototagequiv{\theta} m_2$.

    As we have that $m_1 \prototagequiv{\theta} m_2$, and in
    particular $\mi{eia}_1 := \mi{body}_1[\str{eia}]
    \prototagequiv{\theta} \mi{body}_2[\str{eia}] =: \mi{eia}_2$,
    both, $\mi{eia}_1$ and $\mi{eia}_2$ have the same structure. If
    this structure does not match the expected structure (see
    Line~\ref{line:check-service-token-request-contents}f.), the
    checks in both processing steps fail.

    \gs{``Assuming email does not contain a tag with $\exists a \in \theta$
    such that $a[\mi{dr_1}/y]$ is this tag, i.e. the email addresses
    are equal.'' --- what if not? - the equivalence does not hold for
      email addresses that contain a ``challenged'' tag.} If $r_1$
    accepts the identity assertion, then we have that the tag, the
    email address and the forwarder domain must be equal in $m_1$ and
    $m_2$
    as
    \begin{align*}
    &S_1(r_1).\str{loginSessions}[\mi{body}_1[\str{loginSessionToken}]]\\
    =\,\,
    &S_2(r_1).\str{loginSessions}[\mi{body}_2[\str{loginSessionToken}]]
    \ .
    \end{align*}

    Hence, $r_1$ either accepts in both processing steps or in none.
 
    (B) If $r_1$ accepts, i.e., it does not stop with an empty
    message, then $r_2$ accepts. A nonce is chosen exactly the same in both processing
    steps. Hence, we have that $N_1' = N_2'$.
  \end{itemize}

  \noindent \emph{Definition~\ref{def:eq-of-events},
    Case~\ref{eqe:http-req}:} $e_i^{(1)}$ is an HTTP(S) request 
  from $b_1$ to $r_1$ and $e_i^{(2)}$ is an HTTP(S) request from
  $b_2$ to $r_2$. This case implies
  $p_2 = \fAP{r_2}$.

  We note that Condition~\ref{eqe:pre:rp-scripts} of
  Definition~\ref{def:eq-of-events} holds for the same reasons as in
  the previous case. As the response is always addressed to the IP
  address of $b_1$ or $b_2$, respectively,
  Condition~\ref{eqe:pre:rp-scripts} of
  Definition~\ref{def:eq-of-events} is fulfilled. 

  As we see below, for the output events $E^{(1)}_\text{out}$ and
  $E^{(2)}_\text{out}$ (if any) only Case~\ref{eqe:http-res} of
  Definition~\ref{def:eq-of-events} applies. This implies that the
  output events must contain an HTTP nonce contained in $H$. As we
  know that the HTTP nonce of the incoming HTTP requests is contained
  in $H$ and the output HTTP responses (if any) of the RP reuses the
  same HTTP nonce, the nonce of the HTTP responses is in $H$.

  \begin{itemize}
  \item $\mi{path} = \str{/}$. In this case, the output events
    produced (containing no $l\in L$ or $k\in K$ result in $E_1'$ and $E_2'$ being $\beta$-equivalent
    under $(\theta,H,L)$ according to
    Definition~\ref{def:eq-of-events}, Case~\ref{eqe:http-res}. As
    there are no changes to any state, we have that $S_1'$ and $S_2'$
    are $\gamma$-equivalent under $(\theta,H)$. No new nonces are
    chosen, hence, $N_1 = N'_1 = N_2 = N'_2$.
  \item $\mi{path} = \str{/startLogin}$. As above, both email
    addresses in the input events either equivalent and their domain
    is known to the relying parties, or both email address domains are
    unknown. The latter case is analogue to above.

    Otherwise, $\str{wkCache}$, then $\mathsf{SENDSTARTLOGINRESPONSE}$
    is called.  In both processing steps, a tag is constructed the
    same up to the RP domain $\mi{dr}_1$ or $\mi{dr}_2$, respectively.

    In both processing steps, an HTTP response is created. We denote
    the HTTP response generated by $r_1$ as $m_1'$ and the one
    generated by $r_2$ as $m_2'$. We then have that
    \begin{align*}
      m_1' = \encs{\an{\cHttpResp,n,200,\an{},g_1}}{k} \\
      m_2' = \encs{\an{\cHttpResp,n,200,\an{},g_2}}{k}
    \end{align*}
    with
    \begin{align*}
      g_1 = \an{\an{\str{tagKey},\nu_2},\an{\str{loginSessionToken},\nu_4},\an{\str{FWDDomain},S_1(r_1).\str{FWDDomain}}} \\
      g_2 =
      \an{\an{\str{tagKey},\nu_2},\an{\str{loginSessionToken},\nu_4},\an{\str{FWDDomain},S_2(r_2).\str{FWDDomain}}}
    \end{align*}

    Obviously, $m_1'$ equals $m_2'$. For
    $N_1 = N_2 = (n_1, n_2, \dots)$, we set
    $\theta' = \theta \cup \{ \encs{\an{y,n_1}}{n_2} \}$,
    $N_1' = N_2' = (n_5, \dots)$ (as exactly four nonces are chosen in
    both processing steps), and $L' = L \cup \{n_4\}$. The receiver of
    both messages is the browser $b_1$ or $b_2$, respectively.
    Obviously, it holds that
    $L' = \bigcup_{a\in\theta'}
    \mathsf{loginSessionTokens}(a,S_1',S_2')$
    and there exists an $l' \in L'$ such that
    $g_1[\str{loginSessionToken}] \equiv l'$. As
    Conditions~\ref{eqe:http-res} and~\ref{eqe:pre:k} of
    Definition~\ref{def:eq-of-events} hold, $E_1'$ and $E_2'$ are
    $\beta$-equivalent under $(\theta',H,L')$. The subterm
    $\str{loginSessions}$ of $S_1(r_1)$ is extended exactly the same
    as the subterm $\str{loginSessions}$ of $S_2(r_2)$. \gs{evtl.
      genauer} Thus, we have that $S_1'$ and $S_2'$ are
    $\gamma$-equivalent under $(\theta',H)$. (As mentioned above, in
    both processing steps exactly four nonces are chosen, and we have
    that $N_1' = N_2'$.)

  \item $\mi{path} = \str{/redir}$. By the definition of
    $\beta$-equivalence, the login session token $l$ is the same. By
    the definition of $\gamma$-equivalence, we have that
    Algorithm~\ref{alg:rp-spresso} either (I) stops in both processing
    steps in Line~\ref{line:rp-redir-stop} with an empty output or
    (II) (if $l \in L$) emits an HTTP response in both processing
    steps in Line~\ref{line:rp-redir-response}.

    From Condition~\ref{eqs:logsess} of
    Definition~\ref{def:eq-of-states} we know that for $\mi{ls}_1 :=
    S_1(r_1).\str{loginSessions}[l]$ and $\mi{ls}_2 :=
    S_2(r_2).\str{loginSessions}[l]$ we have that $\mi{ls}_1
    \prototagequiv{\theta} \mi{ls}_2$.

    We denote the HTTP response generated by $r_1$ as $m_1'$ and the
    one generated by $r_2$ as $m_2'$. We then have that
    \begin{align*}
      m_1' = \encs{\an{\cHttpResp,n,200,\an{},g_1}}{k} \\
      m_2' = \encs{\an{\cHttpResp,n,200,\an{},g_2}}{k}
    \end{align*}
    with
    \begin{align*}
      g_1 = \an{\str{script\_rp\_redir}, \an{\cUrl, \https, \mi{domain}_1, \str{/.well\mhyphen{}known/spresso\mhyphen{}login}, \mi{params}_1}} \\
      g_2 = \an{\str{script\_rp\_redir}, \an{\cUrl, \https, \mi{domain}_2, \str{/.well\mhyphen{}known/spresso\mhyphen{}login}, \mi{params}_2}}
    \end{align*}
    and $\mi{domain}_1$, $\mi{domain}_2$, $\mi{params}_1$,
    $\mi{params}_2$ derived exactly the same from $\mi{ls}_1$ and
    $\mi{ls}_2$, respectively (and neither $\mi{params}_1$ nor
    $\mi{params}_2$ contains a key $\str{loginSessionToken}$). No keys $k \in K$ are
    contained in the output. Thus, the output events fulfill
    Condition~\ref{eqe:http-res} of Definition~\ref{def:eq-of-events}.

      No state is changed and no new nonces are chosen. We therefore
      have $\alpha$-equivalence of the new configurations.

  \item $\mi{path} = \str{/login}$. \gs{über diesen fall und den analogen oben nochmal drueberschauen:}

    This case can be handled analogously to the previous case with two
    exceptions:

    (A) First, there are two additional checks, the first in
    Line~\ref{line:check-origin-header} of
    Algorithm~\ref{alg:rp-spresso} checks the origin header and the
    second in Line~\ref{line:rp-check-ia} checks the identity
    assertion. 

    As we know that $m_1 \httptagequiv{\theta} m_2$, we have that if
    the first check fails in $r_1$ then and only then it fails in
    $r_2$. The same holds true for the second check.

    \gs{``Assuming email does not contain a tag with $\exists a \in \theta$
    such that $a[\mi{dr_1}/y]$ is this tag, i.e. the email addresses
    are equal.'' --- what if not? - the equivalence does not hold for
      email addresses that contain a ``challenged'' tag.} If $r_1$
    accepts the identity assertion, then we have that the email
    address must be equal in $m_1$ and $m_2$
    as
    \begin{align*}
    &S_1(r_1).\str{loginSessions}[\mi{body}_1[\str{loginSessionToken}]]\\
    =
    &S_2(r_1).\str{loginSessions}[\mi{body}_2[\str{loginSessionToken}]]
    \ .
    \end{align*}
and we have that the identity assertion in $g_1$ is valid
    for $r_1$, i.e., signed correctly and contains a tag for
    $\mi{dr}_1$, and thus, the identity assertion in $g_2$ is valid
    for $r_2$.  Hence, $r_1$ and $r_2$ either accept in both
    processing steps or in none.\gs{$\nf$ schoener formulieren}

    (B) If $r_1$ accepts, i.e., it does not stop with an empty
    message, we know that $r_2$ accepts. A nonce is chosen exactly the
    same in both processing steps. Hence, we have that $N_1' = N_2'$.
  \end{itemize}

  \paragraph{\underline{Case $p_1 = \fAP{r_2}$:}} This case is
  analogue to the case $p_1 = \fAP{r_1}$ above. Note that the
  Case~\ref{eqe:http-req} of Definition~\ref{def:eq-of-events} cannot
  occur by definition.

  \paragraph{\underline{Case $p_1 = \fAP{b_1}$:}} $\implies p_2 = \fAP{b_2}$ \gs{$\nf$ erklaeren, warum diese implikation gilt}

  We now do a case distinction over the types of messages a browser
  can receive.

  \begin{description}
  \item[DNS response] For the input events either
    Condition~\ref{eqe:prototagequiv} of
    Definition~\ref{def:eq-of-events} or Condition~\ref{eqe:dns-res}
    apply. Therefore, the DNS request/response nonces in both events
    are equivalent up to RP domains under a set of proto-tags
    $\theta$. From Condition~\ref{eqs:b:pendingDNS} of
    Definition~\ref{def:eq-of-states}, we know that for a given nonce,
    there is either an entry in the dictionary $\mi{pendingDNS}$ in
    both browsers or in none. There are no entries under keys that are
    not nonces. Hence, both browsers either continue processing the
    incoming DNS response or stop with no state change and no output
    events in Line~\ref{line:browser-dns-response-stop} of
    Algorithm~\ref{alg:browsermain}. Further, we note that the
    resolved address contained in the DNS response has to be an IP
    address.

    From Condition~\ref{eqs:b:pendingDNS} of
    Definition~\ref{def:eq-of-states}, we know that the protocol in
    both stored HTTP requests is the same. Therefore, the browsers
    either both choose a nonce (for HTTPS request) or none.

    There can now be two cases: (I) The IP addresses in both DNS
    responses are the same or (II) the IP address in $m_1$ is an IP
    address of $r_1$ and the IP address in $m_2$ is an IP address of
    $r_2$. In both cases, the pending requests of the respective
    browsers are amended in such a way that they fulfill
    Condition~\ref{eqs:b:pendingRequests} (as they fulfilled
    Condition~\ref{eqs:b:pendingDNS}, which is essentially the same
    for HTTP(s) requests).

    For $E ^{(1)}_\text{out}$ and $E^{(2)}_\text{out}$, we have that
    in Case~(I)
    $E ^{(1)}_\text{out} \prototagequiv{\theta} E^{(2)}_\text{out}$
    and hence Condition~\ref{eqe:prototagequiv} of
    Definition~\ref{def:eq-of-events} is fulfilled. In Case~(II), the
    output messages fulfill Condition~\ref{eqe:http-req}.

    From Condition~\ref{eqe:pre:l} of
    Definition~\ref{def:eq-of-events}, we know that no $l\in L$ is
    contained in the DNS responses. 
    Further, we know from Condition~\ref{eqe:dns-res} of
    Definition~\ref{def:eq-of-events} that if the IP addresses in the
    DNS responses differ, then they are responses for $\mi{dr}_1$ and
    $\mi{dr}_2$, respectively. From Condition~\ref{eqs:b:pendingDNS}
    of Definition~\ref{def:eq-of-states}, we know that only requests
    (prepared) for $\mi{dr}_1$ and $\mi{dr}_2$, respectively, may
    contain a subterm $l \in L$. Hence, Condition~\ref{eqe:pre:l} of
    Definition~\ref{def:eq-of-states} holds true. 

    We also have that no $k \in K$ is contained in the response (with
    Condition~\ref{eqe:pre:k} of Definition~\ref{def:eq-of-states}).
    There is also no $k \in K$ contained in the browser's pending HTTP
    requests, and therefore, there is none in the output events.

    We have that $S_1'$ and $S_2'$ are $\gamma$-equivalent under
    $(\theta,H)$, $E_1'$ and $E_2'$ are $\beta$-equivalent under
    $(\theta,H,L)$, $N_1' = N_2'$, and thus, the new configurations
    are $\alpha$-equivalent.

  \item[HTTP response] In this case, it is clear that
    the HTTP(s) response nonce, which has to match the nonce in the
    browser's $\str{pendingRequests}$, is either the same in both
    messages $m_1$ and $m_2$ or it contains a tag. If it contains a
    tag (with Condition~\ref{eqs:b:pendingRequests} of
    Definition~\ref{def:eq-of-states}) or if it contains a nonce that
    is not in $\str{pendingRequests}$ (which contains the same nonces
    for both browsers), both browsers stop and do not output anything
    or change their state.

    We can now distinguish between two cases: In both browsers, \ref{browser-http-response-normal}
    the $\mi{reference}$ that is stored along with the HTTP nonce is a
    window reference (in this case, the request was a ``normal''
    HTTP(S) request), or \ref{browser-http-response-xhr} this reference is a pairing of a
    document nonce and an XHR reference chosen by the script that sent
    the request, which is an XHR. From
    Condition~\ref{eqs:b:pendingRequests} of
    Definition~\ref{def:eq-of-states} it is easy to see that no other
    cases are possible (in particular, the $\mi{reference}$ in both
    browsers is the same).

    \begin{enumerate}[label={(\Roman*)}]
    \item\label{browser-http-response-normal} In Case~(I), we can distinguish between the following two
      cases:
      \begin{enumerate}
      \item The HTTP nonce in $m_1$ is in $H$: In this case, only
        Case~\ref{eqe:http-res} of Definition~\ref{def:eq-of-events}
        can apply. We therefore have that there is no Location,
        Set-Cookie or Strict-Transport-Security header in the
        response, and that the responses $m_1$ and $m_2$ are equal up
        to proto-tags in $\theta$. From
        Case~\ref{eqs:b:pendingRequests} of
        Definition~\ref{def:eq-of-states} we have that in both
        browsers $b_1$ and $b_2$ the encryption keys stored in
        $\str{pendingRequests}$ are the same, that the expected sender
        in $e_i^{(1)}$ is $r_1$ and in $e_i^{(2)}$ is $r_2$.

        With this, we observe that both browsers either accept and
        successfully decrypt the messages and call the function
        $\mathsf{PROCESSRESPONSE}$, or both browsers stop with not
        state change and no output event (in which case the
        $\alpha$-equivalence is given trivially). In particular we
        note that the expected sender in both cases matches precisely
        the sender the message has (compare Case~\ref{eqe:http-res} of
        Definition~\ref{def:eq-of-events}).

        In $\mathsf{PROCESSRESPONSE}$, we see that no changes in the
        browsers' cookies are performed (as no cookies are in the
        response), the $\str{sts}$ subterm is not changed, and no
        redirection is performed (as no Location header is present).

        Now, new documents are created in each browser. These have the
        form
        \[ \an{\nu_7, \mi{location}, \mi{referrer}, \mi{script},
          \mi{scriptstate}, \an{}, \an{}, \True} \] with
        \[ \mi{location} = \an{\cUrl, \mi{protocol}, \mi{host},
          \mi{path}, \mi{parameters}}\ .\]

        Here, $\mi{script}$, $\mi{scriptstate}$ are the same and
        $\mi{protocol}$, $\mi{path}$, $\mi{parameters}$ are taken from
        the requests, which means that these subterms are equal or
        term-equivalent up to proto-tags $\theta$ according to
        Case~\ref{eqs:b:pendingRequests} of
        Definition~\ref{def:eq-of-states}. The host and the referrer
        are the same in both states up to exchange of domains, which
        can be $\mi{dr}_1$ in $b_1$ and $\mi{dr}_2$ in $b_2$.

        We note that if $k \in K$, then the request will not be of the
        correct form to be parsed into a document in the browser, and
        both browsers stop with an empty output and no state change.

        The browser now attaches these newly created documents to its
        window tree, and we have to check that the
        Condition~\ref{eqs:b:w} of Definition~\ref{def:eq-of-states}
        is satisfied.

        As we have that both incoming messages were encrypted messages
        (see Case~\ref{eqe:pre:unencrypted-http} of
        Definition~\ref{def:eq-of-events}) and both messages come from
        $r_1$ and $r_2$, respectively, and therefore $\mi{script}$ is
        either $\str{script\_rp}$ or $\str{script\_rp\_redir}$ (see
        Case~\ref{eqe:pre:rp-scripts} of
        Definition~\ref{def:eq-of-events}) we have to check
        Conditions~\ref{eqs:b:w:script_rp}
        and~\ref{eqs:b:w:script_rp_redir} of
        Definition~\ref{def:eq-of-states} in particular.

        The scriptstate is initially equal and may contain a subterm $l \in
        L$ (as we know from HTTP nonce in $m_1$ being in $H$ that the
        host of this document is $\mi{dr}_1$ in $b_1$ and $\mi{dr}_2$
        in $b_2$), and the script inputs are empty. The document's
        referer is constructed from the referer header of the request,
        which is equal in both cases or has the host $\mi{dr}_1$ in
        $b_1$ and $\mi{dr}_2$ in $b_2$.

        To sum up, $\gamma$-equivalence under $(\theta, H)$ is
        preserved. $\alpha$-equivalence is preserved as no output
        event is generated and the exact same number of nonces are
        chosen.

      \item The HTTP nonce in $m_1$ is not in $H$: In this case we
        have that $e_i^{(1)} \prototagequiv{\theta} e_i^{(2)}$
        (Case~\ref{eqe:prototagequiv} of
        Definition~\ref{def:eq-of-events}), and that the HTTP nonces,
        senders, encryption keys (if any) and original requests in the
        pending requests of both browsers are either equal or
        equivalent up to proto-tags $\theta$. There can be no
        $k \in K$ as a subterm (except in tags) of the input.

        With this, we observe that both browsers either accept and
        successfully decrypt the messages and call the function
        $\mathsf{PROCESSRESPONSE}$, or both browsers stop with no
        state change and no output event (in which case the
        $\alpha$-equivalence is given trivially). In particular we
        note that the expected sender in both cases matches precisely
        the sender of the message (as it is equal).

        If there is a Set-Cookie header in one of the responses, a new
        entry in the cookies of each browsers is created (which
        obviously is term-equivalent up to $\theta$, and therefore is
        in compliance with the requirements for $\gamma$-equivalence).
        The same holds true for any Strict-Transport-Security headers.

        Now, if there is a Location header in $m_1$ (and therefore
        also in $m_2$), a new request is generated and stored under
        the pending DNS requests, and a DNS request is sent out. The
        new HTTP(S) requests contains the method, body, and Origin
        header of the original request (which were equivalent up to
        proto-tags $\theta$), where the Origin header is amended by
        the host and protocol of the original request.

        Also, we know from
        $e_i^{(1)} \prototagequiv{\theta} e_i^{(2)}$ that neither
        event may contain a subterm $l\in L$ or $k \in K$. Hence, the
        transferred (initial) scriptstate (or a request generated by a
        Location header, see below) cannot contain a subterm $l \in L$
        or $k \in K$.

        Now, assuming that the domain in the Location header was not
        $\str{CHALLENGE}$, then the new request is term-equivalent
        under $\theta$ between both browsers. A new DNS request is
        generated (which conforms to Condition~\ref{eqe:prototagequiv}
        of Definition~\ref{def:eq-of-states}). It is sent out and the
        HTTP request is stored in the pending DNS requests of each
        browser. It is clear that in this case, the conditions for
        $\gamma$-equivalence under $(\theta, H)$ (in particular,
        Condition~\ref{eqs:b:pendingDNS}) and $\beta$-equivalence
        under $(\theta, H, L)$ are satisfied. The same number of
        nonces is chosen. Altogether, $\alpha$-equivalence is given.

        If, however, the domain is $\str{CHALLENGE}$ (and the browser
        has not started a request to $\str{CHALLENGE}$ before; in this
        case the browser would behave as above), then the domain is
        $\mi{dr}_1$ in $b_1$ and $\mi{dr}_2$ in $b_2$. In particular,
        in the resulting requests, the Host header is exchanged in
        this way. For alpha equivalence to hold for the new
        configuration, we have $H' = H \cup \{n\}$, where $n$ is the
        nonce chosen for the HTTP(S) request. A new DNS request is
        generated (which in this case conforms to
        Condition~\ref{eqe:dns-req} of
        Definition~\ref{def:eq-of-states}). Therefore, we have
        $\gamma$-equivalence under $(\theta, H')$ and
        $\beta$-equivalence under $(\theta, H', L)$. The same number
        of nonces is chosen, and we indeed have $\alpha$-equivalence.

        If there is no Location header in $m_1$ (and therefore none in
        $m_2$), a new document is constructed just as in the case when
        the nonce in $m_1$ is in $H$.

        The scriptstate is initially equal, and the script inputs are
        empty. The document's referer is constructed from the referer
        header of the request, which is equal in both cases (up to
        proto-tags in $\theta$).

        To sum up, $\gamma$-equivalence under $(\theta, H)$ is
        preserved in this case as well. $\alpha$-equivalence is
        preserved as no output event is generated and the exact same
        number of nonces are chosen.
      \end{enumerate}
    \item\label{browser-http-response-xhr}
      In Case~(II), i.e., the response is a response to an XHR, we
      have that $\mi{reference}$ is a tupel, say,
      $\mi{reference} = \an{\mi{docnonce}, \mi{xhrref}}$, and we again
      distinguish between the two cases as above:
      \begin{enumerate}
      \item The HTTP nonce in $m_1$ is in $H$: In this case, only
        Case~\ref{eqe:http-res} of Definition~\ref{def:eq-of-events}
        can apply. We therefore have that there is no Location,
        Set-Cookie or Strict-Transport-Security header in the
        response, and that the responses $m_1$ and $m_2$ are equal up
        to proto-tags in $\theta$. From
        Case~\ref{eqs:b:pendingRequests} of
        Definition~\ref{def:eq-of-states} we have that in both
        browsers $b_1$ and $b_2$ the encryption keys stored in
        $\str{pendingRequests}$ are the same and that the expected sender
        in $e_i^{(1)}$ is $r_1$ and in $e_i^{(2)}$ is $r_2$.

        With this, we observe that both browsers either accept and
        successfully decrypt the messages and call the function
        $\mathsf{PROCESSRESPONSE}$, or both browsers stop with not
        state change and no output event (in which case the
        $\alpha$-equivalence is given trivially). In particular we
        note that the expected sender in both cases matches precisely
        the sender of the message (compare Case~\ref{eqe:http-res} of
        Definition~\ref{def:eq-of-events}).

        In $\mathsf{PROCESSRESPONSE}$, we see that no changes in the
        browsers' cookies are performed (as no cookies are in the
        response), the $\str{sts}$ subterm is not changed, and no
        redirection is performed (as no Location header is present).

        A new input is constructed for the document that is identified
        by $\mi{docnonce}$. We note that such a document exists either
        in both browsers or in none (in which, again, both browsers
        stop with no output or state change). As the input events may
        contain a subterm $l \in L$ (as we know from HTTP nonce in
        $m_1$ being in $H$ that the host of this document is
        $\mi{dr}_1$ in $b_1$ and $\mi{dr}_2$ in $b_2$), the
        constructed scriptinput may also contain a subterm $l \in L$.
        The same holds true for keys $k \in K$.

        For $j \in \{1,2\}$, we have that the $\str{scriptinput}$ term
        for the document in $b_j$ is $\an{\tXMLHTTPRequest,
          \mi{g_j}.\str{body}, \mi{xhrref}}$, where $g_j$ is the HTTP
        body of $m_j$.  With $g_1 \prototagequiv{\theta} g_2$ and
        $\mi{xhrref} \in \nonces \cup \{\bot\}$, it is easy to see
        that the resulting $\str{scriptinput}$ term of the document is
        term-equivalent under proto-tags $\theta$ (as it was before).
        This satisfies $\gamma$-equivalence on the new browser state.

        No output event is generated, and no nonces are chosen.
        Therefore we have $\alpha$-equivalence on the new
        configuration.

      \item The HTTP nonce in $m_1$ is not in $H$: In this case we
        have that $e_i^{(1)} \prototagequiv{\theta} e_i^{(2)}$
        (Case~\ref{eqe:prototagequiv} of
        Definition~\ref{def:eq-of-events}), and that the HTTP nonces,
        senders, encryption keys (if any) and original requests in the
        pending requests of both browsers are either equal or
        equivalent up to proto-tags $\theta$.

        With this, we observe that both browsers either accept and
        successfully decrypt the messages and call the function
        $\mathsf{PROCESSRESPONSE}$, or both browsers stop with not
        state change and no output event (in which case the
        $\alpha$-equivalence is given trivially). In particular we
        note that the expected sender in both cases matches precisely
        the sender the message has (as it is equal).

        If there is a Set-Cookie header in one of the responses, a new
        entry in the cookies of each browsers is created (which
        obviously is term-equivalent up to $\theta$, and therefore is
        in compliance with the requirements for $\gamma$-equivalence).
        The same holds true for any Strict-Transport-Security headers.

        Now, if there is a Location header in $m_1$ (and therefore
        also in $m_2$), both browsers stop with not state change and
        no output event (in which case the $\alpha$-equivalence is
        given trivially), as XHR cannot be redirected in the browser.

        If there is no Location header in $m_1$ (and therefore none in
        $m_2$), a new input is constructed for the document that is
        identified by $\mi{docnonce}$. We note that such a document
        exists either in both browsers or in none. For $j \in
        \{1,2\}$, we have that the $\str{scriptinput}$ for the
        document in $b_j$ is $\an{\tXMLHTTPRequest,
          \mi{g_j}.\str{body}, \mi{xhrref}}$, where $g_j$ is the HTTP
        body of $m_j$. With $e_i^{(1)} \prototagequiv{\theta}
        e_i^{(2)}$ (which may not contain a subterm $l \in L$ or $k \in K$), it is
        easy to see that the resulting $\str{scriptinput}$ term of the
        document is term-equivalent under proto-tags $\theta$ (as it
        was before). This satisfies $\gamma$-equivalence on the new
        browser state.

        No output event is generated, and no nonces are chosen.
        Therefore we have $\alpha$-equivalence on the new
        configuration.
      \end{enumerate}
    \end{enumerate}
  \item[TRIGGER] We now distinguish between the possible values for
    $\mi{cmd}_\text{switch}$.

    \begin{description}
    \item[1 (trigger script):] In this case, the script in the window
      indexed by $\mi{cmd}_\text{window}$ is triggered. Let $j$ be a
      pointer to that window.

      We first note that such a window exists in $b_1$ iff it exists
      in $b_2$ and that
      $S_1(b_1).j.\str{script} \equiv S_2(b_2).j.\str{script}$. We now
      distinguish between the following cases, which cover all
      possible states of the windows/documents:

      \begin{enumerate}

      \item
        $S_1(b_1).j.\str{origin} \in \{\an{\mi{dr}_1,
          \https},\an{\mi{dr}_2, \https}\}$
        and $S_1(b_1).j.\str{script} \equiv \str{script\_rp}$.
        
        Similar to the following scripts, the main distinction in this
        script is between the script's internal states (named
        $\str{q}$). With the term-equivalence under proto-tags
        $\theta$ we have that either
        $S_1(b_1).j.\str{scriptstate}.\str{q} =
        S_2(b_2).j.\str{scriptstate}.\str{q}$ or the script's state
        contains a tag and is therefore in an illegal state (in which
        case the script will stop without producing output or changing
        its state).

        We can therefore now distinguish between the possible values
        of
        $S_1(b_1).j.\str{scriptstate}.\str{q} =
        S_2(b_2).j.\str{scriptstate}.\str{q}$:
        \begin{description}
        \item[start:] In this case, the script chooses one nonce in
          both processing steps and creates an \xhr addressed to its
          own origin, which is in both cases either (a) equal and
          $\an{\mi{dr}_1,\https}$ or $\an{\mi{dr}_2,\https}$ or it is
          (b) $\an{\mi{dr}_1,\https}$ in $b_1$ and
          $\an{\mi{dr}_2,\https}$ in $b_2$. The path is the (static)
          string $\str{/startLogin}$. The script saves the freshly
          chosen nonce (referencing the \xhr) and a (static) value for
          $\str{q}$ in its scriptstate. We note that if a $k \in K$ is
          contained in the script's state, it is never sent out in
          this state.

          In Case~(a), we have that the command is term-equivalent
          under proto-tags $\theta$ and hence, the browser emits a DNS
          request which is term-equivalent, and appends the \xhr in
          $\str{pendingDNS}$. Hence, we have $\gamma$-equivalence
          under $(\theta,H)$ for the new states, $\beta$-equivalence
          under $(\theta,H,L)$ for the new events, and
          $\alpha$-equivalence for the new configuration.

          In Case~(b), we have that the prepared HTTP request is
          $\delta$-equivalent under $\theta$, and is added to
          $\str{pendingDNS}$ in the browser's state, we set $H' := H
          \cup \{n\}$ with $n$ being the nonce that the browser
          chooses for $\lambda_1$. The browser emits a DNS request that
          fulfills Condition~\ref{eqe:dns-req} of
          Definition~\ref{def:eq-of-events}. Therefore, we have
          $\gamma$-equivalence under $(\theta,H')$ for the new states.
          As the request added to $\str{pendingDNS}$ in the browser's
          state fulfills Condition~\ref{eqs:b:pendingDNS}, we have
          $\beta$-equivalence under $(\theta,H',L)$ for the new
          events, and $\alpha$-equivalence for the new configuration.

        \item[expectStartLoginResponse:] In this case, the script
          retrieves the result of an \xhr from $\mi{scriptinputs}$ that matches
          the reference contained in $\mi{scriptstate}$. From
          Condition~\ref{eqs:b:w:script_rp} of
          Definition~\ref{def:eq-of-states} we know that all results
          from \xhr{}s in $\mi{scriptinput}$ are term-equivalent up to
          $\theta$ and that $\mi{scriptstate}$ is term-equivalent up
          to $\theta$. Hence, in both browsers, both scripts stop with
          an empty command or both continue as they successfully
          retrieved such an \xhr.

          Now, a URL is constructed (exactly the same) for an (HTTPS)
          origin that is the origin of the document. We have to
          distinguish between to cases: Either the origin is (i) equal in
          $b_1$ and $b_2$, or (ii) the origin is $\an{\mi{dr}_1,\https}$ in $b_1$
          and $\an{\mi{dr}_2,\https}$ in $b_2$. In the first case (i), no subterm $l
          \in L$ is contained in $\mi{scriptinput}$ and hence, no such
          subterm is contained in the constructed URL. In the second
          case (ii), however, we have that such a subterm may be contained
          in $\mi{scriptinput}$. But as the URL commands the browser
          to prepare a request to $\mi{dr}_1$ and $\mi{dr}_2$,
          respectively, the request may be stored in
          $\str{pendingDNS}$ of the browser's state.

          To store the prepared HTTP request in $\str{pendingDNS}$,
          the browser chooses a nonce $n$. We construct the set $H'$
          as follows: In (i) $H' := H$ and in (ii) $H' := H \cup \{ n
          \}$.  Thus, Condition~\ref{eqs:b:pendingDNS} of
          Definition~\ref{def:eq-of-states} holds true under
          $(\theta,H')$.

          In Case (i) we also have that no $k \in K$ is written into
          $\str{scriptstate}.\str{tagKey}$. Otherwise, a $k \in K$ may
          be written there. In no case, a $k \in K$ is contained in
          the output event or generated HTTP request (as
          $\str{scriptstate}.\str{tagKey}$ is not used to create such
          event or request).

          The output events of both browsers are either a DNS request
          that is equal in (i) or a DNS request that matches
          Condition~\ref{eqe:dns-req} of
          Definition~\ref{def:eq-of-events}.

          We now have that $S_1'$ and $S_2'$ are $\gamma$-equivalent
          under $(\theta,H')$, $E_1'$ and $E_2'$ are
          $\beta$-equivalent under $(\theta,H',L)$, and as exactly the
          same number of nonces is chosen, we have that the new
          configuration is $\alpha$-equivalent.

        \item[expectFWDReady:] In this case, the script retrieves the
          result of a \pm from $\mi{scriptinputs}$. As we know that
          $S_1(b_1).j.\str{scriptstate} \prototagequiv{\theta}
          S_2(b_2).j.\str{scriptstate}$ and that for all matching \pms
          that they also have to be term-equivalent up to $\theta$ and
          that the window structure is equal in both browsers, we have
          that either the same \pm is retrieved from
          $\mi{scriptinputs}$ or none in both browsers.

          The script now constructs a \pm that is sent to exactly the
          same window in both browsers and that requires that the
          receiver origin has to be $\an{\str{fwddomain},\https}$
          \gs{technically, this is not correct, as the forwarder
            domain is retrieved dynamically and our equivalences have
            to be extended such that it actually is
            $\str{fwddomain}$.} The postMessage is only sent to this
          origin, we have that $\gamma$-equivalence cannot be violated
          even if a $k \in K$ is contained in the postMessage (as
          there are no constraints concerning $K$ in the script inputs
          of this origin).

          We now have that $S_1'$ and $S_2'$ are $\gamma$-equivalent
          under $(\theta,H)$, $E_1'$ and $E_2'$ are
          $\beta$-equivalent under $(\theta,H,L)$, and as exactly the
          same number of nonces is chosen, we have that the new
          configuration is $\alpha$-equivalent.

        \item[expectEIA:] 
          In this case, the script retrieves the
          result of a \pm from $\mi{scriptinputs}$. As we know that
          $S_1(b_1).j.\str{scriptstate} \prototagequiv{\theta}
          S_2(b_2).j.\str{scriptstate}$ and that for all matching \pms
          that they also have to be term-equivalent up to $\theta$ and
          that the window structure is equal in both browsers, we have
          that either the same \pm is retrieved from
          $\mi{scriptinputs}$ or none in both browsers.

          The script chooses one nonce in
          both processing steps and creates an \xhr addressed to its
          own origin, which is in both cases either (a) equal and
          $\an{\mi{dr}_1,\https}$ or $\an{mi{dr}_2,\https}$ or it is
          (b) $\an{\mi{dr}_1,\https}$ in $b_1$ and
          $\an{mi{dr}_2,\https}$ in $b_2$. The path is the (static)
          string $\str{/login}$.  The script saves the freshly
          chosen nonce (referencing the \xhr) and a (static) value for
          $\str{q}$ in its scriptstate.

          In Case~(a), we have that the command is term-equivalent
          under proto-tags $\theta$ and hence, the browser emits a DNS
          request which is term-equivalent, and appends the \xhr in
          $\str{pendingDNS}$. Hence, we have $\gamma$-equivalence
          under $(\theta,H)$ for the new states, $\beta$-equivalence
          under $(\theta,H,L)$ for the new events, and
          $\alpha$-equivalence for the new configuration.

          For Case~(b), we note that for $j \in \{1,2\}$, the body
          $g_j$ of the prepared HTTP request may contain an
          (encrypted) identity assertion such that
          \[g_j[\str{eia}] \sim \encs{\sig{
              \an{\encs{\an{\mi{dr}_j,*}}{*}, *, \str{fwddomain}}
            }{*}}{*}\,.\]
          As the receiver of this message is always determined to be
          $\mi{dr}_j$ (in $b_j$) and the Origin header is set
          accordingly, we have that the prepared HTTP request is
          $\delta$-equivalent under $\theta$. The (prepared) request
          is added to $\str{pendingDNS}$ in the browser's state, we
          set $H' := H \cup \{n\}$ with $n$ being the nonce that the
          browser chooses for $\lambda_1$. The browser emits a DNS
          request that fulfills Condition~\ref{eqe:dns-req} of
          Definition~\ref{def:eq-of-events}. In no case is a
          $k \in K$, which can only be stored in
          $\str{scriptstate}.\str{tagKey}$ used to construct either
          output events or state changes. Therefore, we have
          $\gamma$-equivalence under $(\theta,H')$ for the new states.
          As the request added to $\str{pendingDNS}$ in the browser's
          state fulfills Condition~\ref{eqs:b:pendingDNS}, we have
          $\beta$-equivalence under $(\theta,H',L)$ for the new
          events, and $\alpha$-equivalence for the new configuration.
        \end{description}

      \item
        $S_1(b_1).j.\str{origin} \in \{\an{\mi{dr}_1,
          \https},\an{\mi{dr}_2, \https}\}$
        and $S_1(b_1).j.\str{script} \not\equiv \str{script\_rp}$. It
        immediately follows that
        $S_1(b_1).j.\str{script} \equiv \str{script\_rp\_redir}$ in
        this case. This script always outputs the same command to the
        browser: it commands to navigate its window to a URL saved as
        the script's $\mi{scriptstate}$, which is term-equivalent
        under proto-tags between the two browsers. As the HTTP
        request that is generated (and then stored in
        $\str{pendingDNS}$) contains neither an Origin nor a Referer
        header, they are term-equivalent under $\theta$ if the
        domain of this HTTP request is not $\str{CHALLENGE}$\footnote{This also applies
          when the browser challenged before, i.e., $\str{challenge}$
          in the browser's state is $\bot$.} and $\delta$-equivalent
        otherwise.\gs{hier muessen wir wahrscheinlich anders unterscheiden}

        The resulting states of both browsers are therefore
        $\gamma$-equivalent under $\theta$.

        In both cases (challenged or not), we have that
        $E ^{(1)}_\text{out} \prototagequiv{\theta}
        E^{(2)}_\text{out}$
        and hence Condition~\ref{eqe:prototagequiv} of
        Definition~\ref{def:eq-of-events} is fulfilled, or the output
        messages fulfill Condition~\ref{eqe:http-req}.

        Clearly, no $k \in K$ is contained in the script's state, in the
        generated DNS request, or the HTTP request.

        In both processing steps, exactly the same number of nonces is
        chosen. We therefore have $\alpha$-equivalence.

      \item $S_1(b_1).j.\str{origin} = \an{ \str{fwddomain} ,
          \https}$.

        It immediately follows that $S_1(b_1).j.\str{script} \equiv
        \str{script\_fwd}$ in this case.

        As above, we have that either
        $S_1(b_1).j.\str{scriptstate}.\str{q} =
        S_2(b_2).j.\str{scriptstate}.\str{q}$ or the script's state
        contains a tag and is therefore in an illegal state (in which
        case the script will stop without producing output or changing
        its state).

        With the equivalence of the window structures we have that the
        $\mi{target}$ variable in the algorithm of $\str{script\_fwd}$
        in both runs points to a window containing the same script in
        both runs.

        We can now distinguish between the possible values of
        $S_1(b_1).j.\str{scriptstate}.\str{q} =
        S_2(b_2).j.\str{scriptstate}.\str{q}$:
        
        \begin{description}
        \item[start] In this case, a postMessage with exactly the same
          contents is sent to the same window. We therefore trivially
          have $\gamma$-equivalence under $(\theta,H)$ on the states
          in this case. No output events are generated, and no nonces
          are used.  Therefore, $\alpha$-equivalence holds on the new
          states.
        \item[expectTagKey] In this case, for any change in the state
          to occur, a \pm containing some term under the (dictionary)
          key $\str{tagKey}$ sent from exactly the same window has to
          be in $\mi{scriptinputs}$. From
          Condition~\ref{eqs:b:w:script_fwd} of
          Definition~\ref{def:eq-of-states} we know that in
          $\mi{scriptinputs}$ either such a \pm exists in both
          browsers or in none.
                    
          As the $\mi{tag}$ contained in the \pms is term-equivalent
          under proto-tags $\mi{theta}$, we have that the RP domain
          inside the tag is either the same in both processing steps
          or $\mi{dr}_1$ in $b_1$ and $\mi{dr}_2$ in $b_2$.
          Additionally, $\mi{eia}$ (if contained in the URL parameters
          in the location of the script) is term-equivalent under
          proto-tags $\mi{theta}$. It follows that the resulting
          postMessage, which contains $\mi{eia}$, is either delivered
          to the receiver (which is either equal in both browsers or
          $\mi{dr}_1$ in $b_1$ and $\mi{dr}_2$ in $b_2$).
          Additionally, no $k \in K$ can be contained in the
          $\mi{eia}$, as it is taken from the document's location. In
          any case, the resulting browser states are
          $\gamma$-equivalent under $(\theta,H)$.

          As no output events are produced, we have that $E_1'$ and
          $E_2'$ are $\beta$-equivalent under $(\theta,H,L)$. As
          exactly the same number of nonces are chosen (in fact, no
          nonces are chosen), we have that the resulting configuration
          is $\alpha$-equivalent.

          \df{need to apply filter syntax here! wir wissen: nur RP,
          sein passendes script\_rp und letztlich forwarder kennen
          passenden key zum tag. in message vorher wurde key gesendet,
          d.h. in beiden faellen der richtige oder in keinem fall.
          d.h. beide empfangsorigins passen oder keiner. d.h. es
          kommen beide an oder in keinem fall.}

        \end{description}

      \item
        $S_1(b_1).j.\str{origin} \not\in
        \{\an{\mi{dr}_1,\https},\an{\mi{dr}_2,\https},\an{\str{fwddomain},\https}\}$.

        Here, we assumte that the script in this case is the the
        attacker script $R^\text{att}$, as it subsumes all other
        scripts. 

        We first observe, that its ``view'', i.e., the input terms it
        gets from the browser, is term-equivalent up to proto-tags
        $\theta$ between (the scripts running in) $S_1(b_1)$ and
        $S_2(b_2)$. From the equivalence definition of states
        (Definition~\ref{def:eq-of-states}) we can see that:
        \begin{itemize}
        \item The window tree has the same structure in both
          processing steps. All window terms are equal (up to their
          $\str{documents}$ subterm). All same-origin documents
          contain only subterms that are term-equivalent up to $\theta$ (again,
          up to their $\str{subwindows}$ subterms). All
          non-same-origin documents become limited documents and
          therefore are equal (up to the subwindows, limited documents
          only contain the subwindows and the document nonce).
        \item The subterms $\str{cookies}$, $\str{localStorage}$,
          $\str{sessionStorage}$, $\str{scriptstate}$, and
          $\str{scriptinputs}$ are term-equivalent up to $\theta$.
        \item The subterms $\str{ids}$ and $\str{secrets}$ are equal.
        \item There is not $k \in K$ as a subterm (except as keys for
          tags) in this view. We therefore have that no such term can
          be contained in the output command of the script, or in the
          new scriptstate.
        \end{itemize}

        As the input of the script as a whole is term-equivalent up to
        $\theta$, does not contain any placeholders in
        $V_\text{script}$, and does not contain a key for any tag in
        $\theta$, we have that the output of the script, i.e.,
        $\mi{scriptstate}'$, $\mi{cookies}'$, $\mi{localStorage}'$,
        $\mi{sessionStorage}'$, $\mi{command}'$, must be
        term-equivalent up to proto-tags $\theta$ (in particular, the
        same number of nonces is replaced in both output terms in both
        processing steps). Note that the first element of the command
        output must be equal between the two browsers (as it must be
        string) or otherwise the browsers will ignore the command in
        both processing steps.

        Analogously, we see that the input does not contain any
        subterm $l \in L$. \gs{sollte eigentlich als begruendung fuer den rest reichen, dass er nirgendwo ein solches $l$ reinpacken kann.}

        We can now distinguish the possible commands the script can
        output (again, all parameters for these commands must be
        term-equivalent under $\theta$):

        \begin{enumerate}
        \item Empty or invalid command: In this case, the browser
          outputs no message and its state is not changed.
          $\alpha$-equivalence is therefore trivially given.
        \item\label{ratt-sends-href-command}
          $\an{\tHref, \mi{url}, \mi{hrefwindow}, \mi{noreferrer}}$:
          Here, the browser calls $\mathsf{GETNAVIGABLEWINDOW}$ to
          determine the window in which the document will be loaded.
          Due to the synchronous window structure between the two
          browsers, the result will be the same in both processing
          steps (which may include creating a new window with a new
          nonce).

          Now, a new HTTP(S) request is assembled from the URL. A
          Referer header is added to the request from the document's
          current $\str{location}$ (which is term-equivalent under
          $\theta$) and given to the $\mathsf{SEND}$ function. There,
          if the $\str{host}$ part of the URL is $\str{CHALLENGE}$, it
          will be replaced by $\mi{dr}_1$ in $b_1$ and by $\mi{dr}_2$
          in $b_2$. (In this case, the $\alpha$-equivalence in the
          following holds for $H' := H \cup \{n\}$, where $n$ is the
          nonce of the generated HTTP request. Otherwise, it holds for
          $H' := H$.). Afterwards, for domains that are in the
          $\str{sts}$ subterm of the browser's state, the request will
          be rewritten to HTTPS. Any cookies for the domain in the
          requests are added. Note that both latter steps never apply
          to requests to $\mi{dr}_1$ or $\mi{dr}_2$ as per definition,
          there are no entries for these domains in $\mi{sts}$ and
          $\mi{cookies}$.\df{check that we don't send a cookie in the
            model, adapt paper if needed}. The requests, which are
          $\delta$-equivalent under $\theta$ are added to the pending
          DNS requests and fulfill Condition~\ref{eqs:b:pendingDNS} of
          Definition~\ref{def:eq-of-states}. A DNS request is created
          in accordance with Condition~\ref{eqe:dns-req} or
          Condition~\ref{eqe:prototagequiv} of
          Definition~\ref{def:eq-of-events}. The same number of nonce
          is chosen in both processing steps, and therefore
          $\alpha$-equivalence holds.
        \item $\an{\tIframe, \mi{url}, \mi{window}}$
          This case is completely parallel to Case~\ref{ratt-sends-href-command}. 
        \item\label{ratt-sends-form-command}
          $\an{\tForm, \mi{url}, \mi{method}, \mi{data},
            \mi{hrefwindow}}$
          This case is parallel to Case~\ref{ratt-sends-href-command},
          except that an Origin header is added. Its properties are
          the same as those of the Referer header in
          Case~\ref{ratt-sends-href-command}.
        \item\label{ratt-sends-setscript-command} $\an{\tSetScript,
            \mi{window}, \mi{script}}$ In this case, the same document
          is manipulated in both processing steps in the same
          way. Note that only same-origin documents, i.e., attacker
          documents, can be manipulated. No output event is generated,
          and no nonces are chosen. $\alpha$-equivalence is given
          trivially.
        \item $\an{\tSetScriptState, \mi{window}, \mi{scriptstate}}$
          This case is parallel to
          Case~\ref{ratt-sends-setscript-command}.
        \item
          $\an{\tXMLHTTPRequest, \mi{url}, \mi{method}, \mi{data},
            \mi{xhrreference}}$
          This case is parallel to Case~\ref{ratt-sends-href-command}
          with the addition of the Origin header (see
          Case~\ref{ratt-sends-form-command}) and the addition of a
          reference parameter, which is transferred into
          $\str{pendingDNS}$ inside the browser ($\mi{xhrreference}$).
          Therefore, for $\gamma$-equivalence, it is important to note
          that this reference can only be a nonce (and therefore is
          equal in both processing steps). Otherwise, the browser
          stops in both processing steps.
        \item $\an{\tBack, \mi{window}}$,
          $\an{\tForward, \mi{window}}$, and
          $\an{\tClose, \mi{window}}$ If the script outputs one of
          these commands, in both processing steps, the browsers will
          be manipulated in exactly the same way. No output events are
          generated, and no nonces are chosen.
        \item $\an{\tPostMessage, \mi{window}, \mi{message},
            \mi{origin}}$ In this case, a term containing
          $\mi{message}$ (term-equivalent under $\theta$) is added to
          a document's $\str{scriptinput}$ term. If the $\mi{origin}$
          is $\bot$, the same term will be added to the same document
          in both processing steps. Otherwise, the term may only be
          added to one document (if, for example, the origin is
          $\an{\mi{dr}_1, \https}$ and the target documents in both
          browsers have the domain $\mi{dr}_1$ and $\mi{dr}_2$,
          respectively). In this case, however, the equivalence
          defined on the scriptinputs is preserved. This would only be
          possible for $\str{script\_rp}$ and only if the sender
          origin was $\an{\str{fwddomain}, \https}$.
        \end{enumerate}

      \end{enumerate}

    \item[2 (navigate to URL):] \gs{do we need to say that such an url cannot contain a subterm $l \in L$?} In this case, a new window is opened
      in the browser and a document is loaded from $\mi{url}$.

      The states of both browsers are changed in the same way except
      if the domain of the URL is $\str{CHALLENGE}$. In both cases, a
      new (at this point empty) window is created and appended the
      $\str{windows}$ subterm of the browsers. This subterm is
      therefore changed in exactly the same way.

      A new HTTP request is created and appended to
      $\str{pendingDNS}$. The generated requests in both processing
      steps can only differ in the host part iff the domain is
      $\str{CHALLENGE}$. In this case, in $b_1$ the domain is replaced
      by $\mi{dr}_1$ and in $b_2$ by $\mi{dr}_2$ and the
      $\alpha$-equivalence in the following holds for $H' := H \{n\}$,
      where $n$ is the nonce of the generated HTTP request. In both
      cases, the Condition~\ref{eqs:b:pendingDNS} of
      Definition~\ref{def:eq-of-states} is satisfied.

      The request cannot contain any $l \in L$ or $k \in K$.

      The generated DNS requests are equivalent under
      Condition~\ref{eqe:dns-req} or Condition~\ref{eqe:prototagequiv}
      of Definition~\ref{def:eq-of-events}.

      In both processing steps, three nonces are chosen.

      Therefore, we have $\alpha$-equivalence for $(S_1',E_1',N_1')$
      and $(S_2',E_2',N_2')$.
    \item[3 (reload document):] Here, an existing document is
      retrieved from its original location again. From the definition
      of $\gamma$-equivalence under $(\theta,H)$ we can see that
      whatever document is reloaded, its location is either (I)
      term-equivalent under $\theta$, or (II) it is term-equivalent
      under $\theta$ except for the domain, which is $\mi{dr}_1$ in
      $b_1$ and $\mi{dr}_2$ in $b_2$. 

      We note that in either case, the requests are constructed from
      the location and referrer properties of the document that is to
      be reloaded, and therefore, cannot contain any $k\in K$.

      In Case~(I), we note that the domain cannot be
      $\str{CHALLENGE}$. If the document is reloaded, the same DNS
      request is issued in both browsers (therefore,
      $\beta$-equivalence under $(\theta, H, L)$ is given), and an
      entry is added to the pending DNS messages such that we have
      $\gamma$-equivalence under $(\theta, H)$. The same number of
      nonces is chosen in both runs, and we have
      $\alpha$-equivalence.\gs{We have to adopt our equivalences such
        that no subterm $l \in L$ may be contained in these cases and
        it may be contained in the case below.}

      Case~(II) is similar, but we have $H' := H \cup \{n\}$, where
      $n$ is the nonce of the HTTP request that is added to the
      pending DNS entries. Then we have $\gamma$-equivalence under
      $(\theta, H')$. Again, the same number of nonces is chosen and
      we have $\alpha$-equivalence. \df{For fwddomain, the referer
        prototagequiv referer condition is missing. Technically, we
        should add it and test for it.}
    \end{description}

  \item[Other] Any other message is discarded by the browsers without
    any change to state or output events.
  \end{description}

  \paragraph{\underline{Case $p_1$ is some attacker:}}
  
  Here, only Case~\ref{eqe:prototagequiv} from
  Definition~\ref{def:eq-of-events} can apply to the input events,
  i.e., the input events are term-equivalent under proto-tags
  $\theta$. This implies that the message was delivered to the same
  attacker process in both processing steps. Let $A$ be that attacker
  process. With Case~\ref{eqs:att} of
  Definition~\ref{def:eq-of-states} we have that
  $S_1(A) \prototagequiv\theta S_2(A)$ and with
  Case~\ref{eqs:att-not-k} and Case~\ref{eqe:pre:k} of
  Definition~\ref{def:eq-of-events} it follows immediately that the
  attacker cannot decrypt any of the tags in $\theta$ in its
  knowledge. Further, in the attackers state there are no variables
  (from $V_\text{process}$).\df{$\nf$ Add this to equivalence? But
    should be very clear anyway.}

  With the output term being a fixed term (with variables)
  $\tau_{\text{process}} \in \terms(\{x\} \cup V_\text{process})$ and
  $x$ being $S_1(A)$ or $S_2(A)$, respectively, and there is no
  subterm $l \in L$ contained in either $S_1(A)$ or $S_2(A)$
  (Condition~\ref{eqs:att-not-l} of
  Definition~\ref{def:eq-of-states}), it is easy to see that the
  output events are $\beta$-equivalent under $\theta$, i.e.,
  $E ^{(1)}_\text{out} \prototagequiv\theta E^{(2)}_\text{out}$, there
  are not $k \in K$ contained in the output events (except as
  encryption keys for tags) and the used nonces are the same, i.e.,
  $N_1' = N_2'$. The new state of the attacker in both processing
  steps consists of the input events, the output events, and the
  former state of the event, and, as such, is $\beta$-equivalent under
  proto-tags $\theta$. Therefore we have $\alpha$-equivalence on the
  new configurations.
  
\qed
\end{proof}

\noindent
This proves Theorem~\ref{thm:privacy-main}.
\QED

\end{document}
